\documentclass[11pt]{article}
\usepackage{siunitx}

\usepackage{amsmath,amssymb,amsthm,bm,mathtools}
\usepackage{graphicx,booktabs,enumitem,geometry,rotating}
\usepackage{microtype}
\usepackage{needspace}
\usepackage{placeins}
\usepackage{algorithm,algorithmic}
\usepackage[font=small,labelfont=bf,labelsep=period]{caption}
\usepackage[authoryear,round]{natbib}
\usepackage{authblk}
\usepackage{xcolor}
\usepackage{hyperref}
\renewcommand{\arraystretch}{1}
\definecolor{linktone}{RGB}{24,69,92}
\hypersetup{colorlinks=true,linkcolor=blue,citecolor=blue,urlcolor=blue}
\newtheorem{assumption}{Assumption}
\newtheorem{theorem}{Theorem}
\newtheorem{lemma}{Lemma}
\newtheorem{proposition}{Proposition}
\newtheorem{corollary}{Corollary}
\theoremstyle{definition}
\newtheorem{remark}{Remark}
\newcommand{\E}{\mathbb{E}}
\newcommand{\1}{\mathbb{I}}
\newcommand{\cT}{\mathcal{T}}
\newcommand{\cI}{\mathcal{I}}
\newcommand{\cH}{\mathcal{H}}
\newcommand{\cV}{\mathcal{V}}
\newcommand{\cE}{\mathcal{E}}
\newcommand{\Cov}{\operatorname{Cov}}

\newcommand{\expit}{\operatorname{expit}}
\title{\textbf{Semiparametric Inference for Functional Difference-in-Differences}}
\author[1]{Junzhu Nie}
\author[1]{Chengxiu Ling}
\author[1]{Mengfei Ran\footnote{Corresponding author. \url{mengfei.ran@xjtlu.edu.cn}}}
\affil[1]{Academy of Life and Natural Sciences, Xi'an Jiaotong-Liverpool University}
\date{}
\begin{document}
\maketitle
\begin{abstract}
We study the functional average treatment effect on the treated in a two-period difference-in-differences design. Under conditional functional parallel trends, we derive the full-curve efficient influence function and efficiency bound and construct a cross-fitted, doubly robust estimator. Under ultra-dense observation, bounds on curve-estimation error and learner stability preserve its full-data first-order behavior. Nuisance-rate and moment conditions yield asymptotic linearity in $L^2$; additional Sobolev regularity gives uniform Gaussian limits and valid continuous-domain bands computed on a refining grid. In simulations, estimation errors decrease with sample size under single-model misspecification, and bands accounting for nuisance estimation have near-nominal or conservative coverage. Further experiments examine basis-truncation bias and nuisance-estimation uncertainty. An application to electricity load profiles from $4389$ eligible households estimates a modest reduction concentrated in evening consumption among dynamic-tariff participants.
\end{abstract}
\vspace{0.1in}

\section{Introduction}\label{sec:intro}
Interventions can affect an outcome's temporal profile as well as its overall level. A dynamic electricity tariff may lower evening demand while leaving overnight consumption unchanged, a distinction scalar averages can conceal. The potential-outcomes framework defines counterfactual comparisons \citep{rubin1974,holland1986statistics,imbens2015causal,hernan2020causal}, and functional data analysis (FDA) represents trajectories and their variation \citep{ramsay2005fda,hsing2015theoretical,wang2016functional}. Curve-level estimation, simultaneous bands, and interval-based tests describe the location and duration of changes and their uncertainty \citep{degras2011scb,choi2018geometric,liebl2023ffscb,pini2016itp,pini2017iwt}.

Difference-in-differences (DiD) compares changes in treated and comparison groups, permitting baseline differences while restricting their untreated evolution. Covariate adjustment allows these changes to depend on observed characteristics. Building on propensity-score methods \citep{rosenbaum1983}, \citet{abadie2005semiparametric} develops semiparametric DiD weighting, and \citet{santanna2020drdid} combine weighting and outcome regression in doubly robust estimators of the average treatment effect on the treated. Extensions address multiple periods, staggered adoption, and heterogeneous effects \citep{callaway2021did,sunabraham2021eventstudy,borusyak2024revisiting,roth2023whats}. Extending pointwise DiD to curves requires joint inference across locations and control of the error from discrete measurements.

Recent studies connect causal inference with functional outcomes. \citet{ecker2024functional} develop outcome-regression estimators and simultaneous inference, while \citet{testa2025drfate} develop doubly robust estimation with Gaussian-process inference. Both use ignorability of potential outcome levels. For functional outcomes missing at random, \citet{liu2026functional} establish simultaneous inference under single-model correctness. \citet{fangliebl2025honest} study honest bands for functional event-study plots indexed by event time, whereas our domain indexes within-period outcomes, such as hour of day. \citet{belloni2017program} develop uniform causal inference with estimated nuisances. We address joint inference under conditional parallel trends for fully or ultra-densely observed outcomes.

We consider a two-period panel targeting the functional average treatment effect on the treated (fATT). Calendar period determines the before--after comparison; the functional argument describes within-period variation. Conditional functional parallel trends restricts untreated changes while allowing time-invariant group differences. We combine a doubly robust DiD moment with cross-fitting and orthogonality \citep{bickel1993efficient,belloni2014post,chernozhukov2018dml}. Uniform inference requires stronger control than convergence in integrated squared error. Estimating individual curves affects both nuisance training and score evaluation; negligible trajectory errors and outcome-learner stability control these effects. Rates for population features under sparse sampling \citep{yao2005fpca,hall2006properties} differ from the trajectory-level rates required here. Our discretely observed theory uses the ultra-dense conditions in Assumption~\ref{ass:asym_recon}.

We identify the fATT and derive the canonical gradient and efficiency covariance bound in the full-curve model. For fully or ultra-densely observed curves, we establish asymptotic linearity with bounds on nuisance and trajectory errors, and give constructions satisfying the rate and stability conditions. Under stronger smoothness conditions, we establish continuous-domain weak convergence, covariance consistency, and valid Gaussian multiplier bands, including the approximation error from a refining evaluation grid.

Sections~\ref{sec:methodology}--\ref{sec:inference} develop the method and theory; Sections~\ref{sec:simulation}--\ref{sec:discussion} present simulations, the electricity application, and discussion. The Appendix contains proofs and sufficient constructions. Supplementary material develops working-model inference and reports additional simulation and application results.

\section{Methodology}\label{sec:methodology}
\subsection{Setup and Notation}\label{sec:notation}
After rescaling the functional domain, let $\cT=[0,1]$ and $\cH=L^2[0,1]$, with inner product $\langle f,g\rangle=\int_0^1 f(t)g(t)\,dt$. Independent units have baseline covariates $X_i$, group indicator $D_i$, and potential curves $Y_{iT}^d\in\cH$, where $T=0,1$ indexes calendar period, $d=0,1$ indexes treatment regime, and $t$ indexes within-period time. Write $\Delta Y_i=Y_{i1}-Y_{i0}$ and $W_i=(\Delta Y_i,D_i,X_i)$ for the full-curve variables, with law $P$.

For discretely observed data, the available outcomes are
\begin{equation}\label{eq:obs_model}
 Z_{iTj}=Y_{iT}(t_{iTj})+\varepsilon_{iTj},\qquad j=1,\ldots,m_{iT}.
\end{equation}
Point evaluations in \eqref{eq:obs_model} use a continuous representative of $Y_{iT}$, assumed to exist almost surely; Appendix Proposition~\ref{prop:curve_rates} gives sufficient conditions. Let $\widehat Y_{iT}$ estimate $Y_{iT}$ from these measurements, and write $\widehat{\Delta Y}_i=\widehat Y_{i1}-\widehat Y_{i0}$. Throughout, $n$ counts independent units, $K$ the cross-fitting folds, $m_n$ a lower bound on observations per curve, $L_n$ the basis dimension, $M_n$ the inference-grid size, $B$ the multiplier draws, and $R$ the Monte Carlo replications.

\subsection{Target Estimand}\label{sec:estimand}
The target is the population functional average treatment effect on the treated (fATT),
\begin{equation*}
 \tau_0=\E(Y_1^1-Y_1^0\mid D=1)\in\cH.
\end{equation*}
This population target differs from the realized sample average. Equalities in $\cH$ hold almost everywhere in $t$; under Assumption~\ref{ass:asym_paths}, their continuous versions hold throughout $\cT$. A fixed basis projection is a different target unless $\tau_0$ belongs to that span.

\subsection{Identification}\label{sec:identification}
Let $\pi_0(x)=P(D=1\mid X=x)$, $p_0=P(D=1)$, and $\mu_a(x)=\E(\Delta Y\mid X=x,D=a)$.
\begin{assumption}[Identification]\label{ass:ident}
The potential curves have integrable $\cH$ norms. In addition:
\begin{enumerate}[label=(A\arabic*),leftmargin=*]
\item \textit{Consistency and no interference.} Potential outcomes for a unit do not depend on other units' assignments, and $Y_0=Y_0^0$, $Y_1=DY_1^1+(1-D)Y_1^0$.
\item \textit{No anticipation.} $Y_0^1=Y_0^0$ almost surely.
\item \textit{Conditional functional parallel trends.}
\begin{equation*}
 \E(Y_1^0-Y_0^0\mid X,D=1)=\E(Y_1^0-Y_0^0\mid X,D=0).
\end{equation*}
\item \textit{Overlap.} For a constant $c_0>0$, $c_0\leq\pi_0(X)\leq1-c_0$ almost surely.
\end{enumerate}
\end{assumption}
Under Assumption~\ref{ass:ident}, iterated expectation and consistency imply
\begin{equation}\label{eq:ident_reg}
 \tau_0=\E\{\Delta Y-\mu_0(X)\mid D=1\}
       =\E\{\mu_1(X)-\mu_0(X)\mid D=1\}.
\end{equation}
Conditional parallel trends identifies the counterfactual mean change as $\E\{\mu_0(X)\mid D=1\}$. Since
$\E\{(1-D)\pi_0(X)g(X)/(1-\pi_0(X))\}=\E\{Dg(X)\}$ for any integrable $g$, an equivalent representation is
\begin{equation*}
 \tau_0=\frac{1}{p_0}\E\left[\left\{D-\frac{(1-D)\pi_0(X)}{1-\pi_0(X)}\right\}\Delta Y\right].
\end{equation*}

\subsection{Efficient Influence Function}\label{sec:eif}
Semiparametric efficiency is characterized through derivatives along statistical submodels \citep{newey1990semiparametric,tsiatis2006semiparametric}. Propensity scores enter both identification and efficiency for treatment-effect targets \citep{hahn1998role,hirano2003efficient}. Write $\eta_0=(\pi_0,\mu_0)$. The full-curve influence function is
\begin{equation}\label{eq:eif_point}
 \phi(W)=\frac{D}{p_0}\{\Delta Y-\mu_0(X)-\tau_0\}
 -\frac{(1-D)\pi_0(X)}{p_0\{1-\pi_0(X)\}}\{\Delta Y-\mu_0(X)\}.
\end{equation}
The term $D\tau_0/p_0$ accounts for estimating the treated-group probability. The treated regression $\mu_1$ cancels from the augmented representation and is not needed by CF--DR.

\begin{proposition}[Canonical gradient and efficiency bound]\label{prop:eif}
In the nonparametric model for the full-curve law $P$, suppose overlap and $\E\|\Delta Y\|_{\cH}^2<\infty$ hold. The identified functional in \eqref{eq:ident_reg} has canonical gradient \eqref{eq:eif_point}. Its efficiency covariance operator is
\begin{equation*}
 \Sigma_{\mathrm{eff}}f=\E\{\langle\phi,f\rangle\phi\},\qquad
 V_{\mathrm{eff}}(f)=\langle f,\Sigma_{\mathrm{eff}}f\rangle
                   =\E\langle f,\phi\rangle^2.
\end{equation*}
For each fixed $f\in\cH$, $V_{\mathrm{eff}}(f)$ is the lower bound on the asymptotic variance of regular asymptotically linear estimators of $\langle f,\tau_0\rangle$.
More generally, a regular $\cH$-valued asymptotically linear estimator with square-integrable influence function $\chi$ has covariance operator $\Sigma_\chi f=\E\{\langle\chi,f\rangle\chi\}$ satisfying
\[
 \Sigma_\chi-\Sigma_{\mathrm{eff}}\succeq0,
 \quad\text{that is,}\quad
 \langle f,(\Sigma_\chi-\Sigma_{\mathrm{eff}})f\rangle\geq0
 \quad(f\in\cH).
\]
\end{proposition}
Under the Sobolev regularity and moment conditions in Assumption~\ref{ass:asym_paths}, its kernel is $C(s,t)=\E\{\phi(W)(s)\phi(W)(t)\}$. The Appendix gives its conditional-variance decomposition in \eqref{eq:effkernel}. This full-curve bound also describes the first-order covariance under discrete observation when the rate and stability conditions below hold.

\subsection{Cross-fitted Doubly Robust Estimator}\label{sec:estimation}
Augmentation combines regression and weighting in a doubly robust moment \citep{robins1994estimation,bang2005doubly}. For $\eta=(\pi,\mu)$ define the unnormalized contribution and normalized score
\begin{equation}\label{eq:drscore_simplified}
 q(W;\eta)=\left\{D-\frac{(1-D)\pi(X)}{1-\pi(X)}\right\}\{\Delta Y-\mu(X)\},
 \qquad \psi_p(W;\eta)=q(W;\eta)/p.
\end{equation}
Lemma~\ref{lem:orth} in the Appendix establishes the exact drift identity
\begin{equation}\label{eq:drift}
 \E q(W;\eta)-p_0\tau_0
 =\E\left\{\frac{\pi(X)-\pi_0(X)}{1-\pi(X)}[\mu(X)-\mu_0(X)]\right\}.
\end{equation}
This identity gives double robustness and Neyman orthogonality at $\eta_0$. Inference under single-model misspecification additionally requires accounting for nuisance estimation \citep{santanna2020drdid}.

Partition the units independently of their data into $K$ folds $\cI_k$, with $n_k/n\to\rho_k\in(0,1)$. Fit $\widehat\eta^{(-k)}=(\widehat\pi^{(-k)},\widehat\mu_0^{(-k)})$ outside fold $k$. Let $\widetilde q_i$ denote \eqref{eq:drscore_simplified} with these fits and the observed change curve, or its fold-specific estimate for discrete measurements. With $\widehat p=n^{-1}\sum_iD_i$, set
\begin{equation}\label{eq:tauhat}
 \widehat\tau=\frac{1}{n\widehat p}\sum_{k=1}^K\sum_{i\in\cI_k}\widetilde q_i.
\end{equation}
Both terms use the common denominator $n\widehat p$.

When $\widehat{\Delta Y}_i(t)=\bm b_n(t)^\top\widehat{\bm\theta}_i$ and $\widehat\mu_0^{(-k)}(x)(t)=\bm b_n(t)^\top\widehat{\bm m}_0^{(-k)}(x)$ use a common deterministic basis, compute
\begin{equation}\label{eq:coefficient_score}
 \bm v_i=\left\{D_i-\frac{(1-D_i)\widehat\pi^{(-k)}(X_i)}{1-\widehat\pi^{(-k)}(X_i)}\right\}
 \{\widehat{\bm\theta}_i-\widehat{\bm m}_0^{(-k)}(X_i)\},\qquad
 \widehat\tau(t)=\bm b_n(t)^\top\frac{\sum_i\bm v_i}{n\widehat p}.
\end{equation}
This equivalent coefficient form avoids separate regressions at every grid point. The error bounds below include basis approximation bias.

\subsection{Ultra-dense Observation}\label{sec:recon}
Let $\widehat Y_{iT}^{(-k)}$ denote the curve estimate used in fold $k$, and define the maximum error across its training and evaluation records by
\[
 r_{n,\cE}=\max_{1\leq k\leq K}\max_{1\leq i\leq n}\max_{T\in\{0,1\}}
 \|\widehat Y_{iT}^{(-k)}-Y_{iT}\|_{\cE}.
\]
Shared components, including any learned basis or tuning parameters, are fitted on the training sample; a held-out curve may use that unit's own measurements. We require $\sqrt n\,r_{n,\cE}=o_p(1)$, including basis approximation bias. Proposition~\ref{prop:curve_rates} verifies this rate for a well-conditioned fixed-basis design and a regular midpoint design for nonparametric curves; it is not guaranteed for fixed-sparse or arbitrary irregular designs. Assumption~\ref{ass:asym_recon} also bounds how errors in the training responses affect nuisance fits.

\section{Asymptotic Theory}\label{sec:asymptotic}
\subsection{Preliminaries}\label{ssec:asym_prelim}
For a separable Hilbert space $\cE$, let $\|h\|_{2,\cE}^2=\E\|h(X)\|_{\cE}^2$, and use $\|\cdot\|_2$ for scalar $L^2(P_X)$ norms. Initially $\cE=\cH$. Write $P_{n,k}$ for the empirical measure on fold $k$. Conditioning on a training fold includes its noisy observations and all fitted objects. The corresponding held-out $W_i$ retain law $P$.

\subsection{Assumptions}\label{ssec:asym_assumptions}
\begin{assumption}[Sampling and moments]\label{ass:asym_basic}
The $W_i$ are i.i.d.; the observed records are independent across units, although their within-unit observation density may increase with $n$. Assumption~\ref{ass:ident} holds, $\E\|\Delta Y\|_{\cE}^2<\infty$, and
\[
 \mathop{\mathrm{ess\,sup}}_x\E\{\|\Delta Y-\mu_0(X)\|_{\cE}^2\mid X=x,D=0\}<\infty.
\]
\end{assumption}
\begin{assumption}[Cross-fitting]\label{ass:asym_cf}
The number of folds is fixed, their limiting proportions are positive, and fold assignment is independent of the data. Nuisance fits and any data-adaptive components shared across curve estimates are trained without the evaluation units. For a fixed $0<c<c_0$, fitted propensities lie in $[c,1-c]$ on the support of $X$, possibly after clipping.
\end{assumption}
\begin{assumption}[Nuisance rates]\label{ass:asym_rates}
Let $\widehat\mu_0^{\mathrm{lat},(-k)}$ be the outcome learner fitted to the training units' true changes $\Delta Y_i$, using the same covariates and tuning rule as the implemented fit. Define
\[
 a_n=\max_k\|\widehat\pi^{(-k)}-\pi_0\|_2,\qquad
 b_{n,\cE}=\max_k\|\widehat\mu_0^{\mathrm{lat},(-k)}-\mu_0\|_{2,\cE}.
\]
Then $a_n=o_p(1)$, $b_{n,\cE}=o_p(1)$ and $a_nb_{n,\cE}=o_p(n^{-1/2})$.
\end{assumption}
\begin{assumption}[Curve accuracy and response stability]\label{ass:asym_recon}
The maximum curve error satisfies $r_{n,\cE}=o_p(n^{-1/2})$. The implemented and full-curve outcome fits obey
\begin{equation}\label{eq:learner_stability}
 \max_k\sup_x\|\widehat\mu_0^{(-k)}(x)-\widehat\mu_0^{\mathrm{lat},(-k)}(x)\|_{\cE}
 \leq C_n r_{n,\cE},\qquad C_n=O_p(1).
\end{equation}
For fully observed curves set $r_{n,\cE}=0$.
\end{assumption}
Response-linear learners $\widehat\mu(x)=\sum_jw_j(x)\Delta Y_j$ satisfy \eqref{eq:learner_stability} if the weights depend only on training covariates and group labels and $\sup_x\sum_j|w_j(x)|=O_p(1)$. Response-adaptive learners, including response-tuned smoothing or partition selection, require a separate stability argument.

\begin{remark}[A sufficient nuisance-learning construction]
Proposition~\ref{prop:histogram} verifies $a_n,b_{n,\cE}=O_p(n^{-1/3})$ and \eqref{eq:learner_stability} for a one-dimensional bin learner. Its nuisance product is $O_p(n^{-2/3})=o_p(n^{-1/2})$; Proposition~\ref{prop:curve_rates} gives observation schemes satisfying the curve-error condition.
\end{remark}

\subsection{Asymptotic Linearity}\label{ssec:asym_linear}
\begin{theorem}[Hilbert-space expansion]\label{thm:asym_linear}
Under Assumptions~\ref{ass:asym_basic}--\ref{ass:asym_recon} with $\cE=\cH$, the estimator \eqref{eq:tauhat} satisfies
\begin{equation}\label{eq:asym_linear}
 \sqrt n(\widehat\tau-\tau_0)=\frac1{\sqrt n}\sum_{i=1}^n\phi(W_i)+o_p(1)
 \quad\text{in }\cH.
\end{equation}
\end{theorem}
The same expansion holds in any separable Hilbert space $\cE$ satisfying these assumptions. The Appendix bounds each remainder, including the effects of estimating $p_0$ and the individual curves.

\subsection{Weak Convergence and Efficiency}\label{ssec:asym_clt}
\begin{theorem}[Hilbert-space limit]\label{thm:asym_clt}
Under the assumptions of Theorem~\ref{thm:asym_linear}, $\sqrt n(\widehat\tau-\tau_0)\Rightarrow\mathbb G$ in $\cH$, where $\mathbb G$ is centered Gaussian with covariance $\Sigma_{\mathrm{eff}}$. For every fixed $f\in\cH$,
\[
 \sqrt n\langle f,\widehat\tau-\tau_0\rangle\Rightarrow N\{0,V_{\mathrm{eff}}(f)\}.
\]
\end{theorem}
The covariance equals the full-curve efficiency bound; efficiency among regular estimators additionally requires local regularity under perturbations of $P$. Hilbert-space convergence supports continuous linear summaries \citep{horvath2012inference}; pointwise and simultaneous inference require the stronger conditions below.

\subsection{Continuous-domain Convergence}\label{ssec:asym_uniform}
For simultaneous inference use the separable Hilbert space $\cV=H^1[0,1]$ with $\|f\|_{\cV}^2=\|f\|_{\cH}^2+\|f'\|_{\cH}^2$. Each element has a continuous representative, and
\begin{equation}\label{eq:sobolev_bounds}
 \|f\|_\infty\leq C\|f\|_{\cV},\qquad
 |f(t)-f(s)|\leq\|f'\|_{\cH}|t-s|^{1/2}.
\end{equation}
These one-dimensional bounds follow from absolute continuity and Cauchy--Schwarz; see \citet{adams2003sobolev}.
\begin{assumption}[Sobolev regularity]\label{ass:asym_paths}
Assumptions~\ref{ass:asym_basic}--\ref{ass:asym_recon} hold with $\cE=\cV$, and $\inf_{t\in[0,1]}\sigma(t)>0$, where $\sigma^2(t)=\E\phi(W)(t)^2$.
\end{assumption}
In particular, conditional Jensen's inequality and overlap give $\mu_0\in L^2(P_X;\cV)$, $\tau_0\in\cV$, and $\E\|\phi\|_{\cV}^2<\infty$; thus all point evaluations below use their continuous versions.
\begin{theorem}[Uniform expansion and limit]\label{thm:asym_sup}
Under Assumption~\ref{ass:asym_paths}, \eqref{eq:asym_linear} holds in $\cV$ and hence in $C[0,1]$. Consequently,
\[
 \sqrt n(\widehat\tau-\tau_0)\Rightarrow\mathbb Z\quad\text{in }C[0,1],
\]
where $\mathbb Z$ is the continuous image of the $\cV$-valued Gaussian limit and has kernel \eqref{eq:effkernel}.
\end{theorem}
Sobolev consistency controls both regressions and their functional derivatives, yielding continuous-domain convergence by embedding.

\section{Inference and Confidence Bands}\label{sec:inference}
\subsection{Estimated Influence Contributions}\label{ssec:influence_contrib}
For $i\in\cI_k$, let
\begin{equation*}
 \widehat\phi_i=\frac{\widetilde q_i-D_i\widehat\tau}{\widehat p},\qquad
 \bar{\widehat\phi}=n^{-1}\sum_i\widehat\phi_i,
 \qquad \widehat C(s,t)=\frac1n\sum_i\{\widehat\phi_i(s)-\bar{\widehat\phi}(s)\}
                                      \{\widehat\phi_i(t)-\bar{\widehat\phi}(t)\}.
\end{equation*}
The mean $\bar{\widehat\phi}$ is algebraically zero for \eqref{eq:tauhat}; retaining centering is convenient numerically. Set $\widehat\sigma^2(t)=\widehat C(t,t)$.

\subsection{Pointwise Confidence Intervals}\label{ssec:pointwise_ci}
Under Assumption~\ref{ass:asym_paths}, the Appendix proves $\sup_t|\widehat\sigma(t)-\sigma(t)|=o_p(1)$. Thus, at any prespecified $t$, an asymptotic $100(1-\alpha)\%$ interval is
\begin{equation*}
 \widehat\tau(t)\pm z_{1-\alpha/2}\widehat\sigma(t)/\sqrt n.
\end{equation*}
Here $\widehat\sigma/\sqrt n$ is the estimator's standard error; coverage concerns the chosen $t$.

\subsection{Simultaneous Confidence Bands}\label{ssec:scb}
Uniform inference can be developed through empirical-process and Gaussian-approximation arguments \citep{vandervaartwellner1996weak,chernozhukov2014gaussian}. Here Sobolev regularity permits a conditional Gaussian argument in a separable Hilbert space. Use independent standard-normal multipliers $\xi_i$ and define
\begin{equation*}
 \mathbb G_n^*(t)=n^{-1/2}\sum_i\xi_i\{\widehat\phi_i(t)-\bar{\widehat\phi}(t)\}.
\end{equation*}
Let $\cT_n\subset[0,1]$ include both endpoints, with fill distance $\delta_n=\sup_t\min_{u\in\cT_n}|t-u|\to0$. Define
\[
 T_n^*=\max_{t\in\cT_n}|\mathbb G_n^*(t)/\widehat\sigma(t)|,
 \qquad T=\sup_{t\in[0,1]}|\mathbb Z(t)/\sigma(t)|.
\]
\Needspace{10\baselineskip}
\begin{lemma}[Variance estimation and multiplier convergence]\label{lem:bootstrap}
Under Assumption~\ref{ass:asym_paths}, the following hold.
\begin{enumerate}[label=(\roman*),leftmargin=*]
\item \textit{Influence-function and covariance consistency:}
\[
 \frac1n\sum_i\|\widehat\phi_i-\phi(W_i)\|_{\cV}^2=o_p(1),\qquad
 \sup_{s,t}|\widehat C(s,t)-C(s,t)|=o_p(1).
\]
\item \textit{Conditional multiplier convergence:} conditionally on the data, $\mathbb G_n^*$ converges weakly in $\cV$ in probability to the $\cV$-valued limit in Theorem~\ref{thm:asym_sup}.
\item \textit{Grid approximation:} the grid-to-continuum maximum errors of the studentized estimation and multiplier processes are $O_p(\delta_n^{1/2})$ and $O_{P_\xi}(\delta_n^{1/2})$ in probability, respectively.
\end{enumerate}
\end{lemma}
\begin{theorem}[Validity of continuous-domain bands]\label{thm:band}
Suppose Assumption~\ref{ass:asym_paths} holds, $\delta_n\to0$, and the distribution of $T$ is continuous and strictly increasing in a neighborhood of its $(1-\alpha)$ quantile. Let $B=B_n\to\infty$ and let $\widehat c_{1-\alpha}$ be the empirical quantile of $B_n$ independent conditional draws of $T_n^*$. Then
\begin{equation}\label{eq:scb}
 P\left\{\tau_0(t)\in\left[\widehat\tau(t)\pm\widehat c_{1-\alpha}\widehat\sigma(t)/\sqrt n\right]
 \text{ for every }t\in[0,1]\right\}\longrightarrow1-\alpha.
\end{equation}
\end{theorem}
A numerical variance floor can be chosen to tend to zero; uniform consistency and $\inf_t\sigma(t)>0$ make it asymptotically inactive. Increasing $M_n$ refines numerical evaluation; increasing $m_n$ supplies additional measurements per curve.

\subsection{Algorithm}\label{ssec:algorithm}
The coefficient version below assumes a common deterministic basis. If basis construction is data-adaptive, train it within folds and evaluate the resulting functions in a common domain before averaging.

\begin{algorithm}[H]
\small
\caption{Cross-fitted functional DiD with Gaussian multiplier bands}\label{alg:dr_scb}
\begin{algorithmic}[1]
\STATE \textbf{Input:} unit records, curve-fitting specification, $K$, basis dimension $L_n$, grid $\cT_n$, $B$, clipping level $c$, and $\alpha$.
\STATE Partition units independently of the data; set $\widehat p=n^{-1}\sum_iD_i$.
\FOR{each training fold complement $\cI_k^c$}
\STATE Fit any shared curve-fitting components on $\cI_k^c$; estimate coefficients for the training and held-out curves.
\STATE Fit $\widehat\pi^{(-k)}$ and clip to $[c,1-c]$; regress control coefficients on $X$ to obtain $\widehat{\bm m}_0^{(-k)}$.
\STATE Compute held-out coefficient contributions $\bm v_i$ from \eqref{eq:coefficient_score}.
\ENDFOR
\STATE Set $\widehat{\bm\tau}=\sum_i\bm v_i/(n\widehat p)$ and $\bm u_i=(\bm v_i-D_i\widehat{\bm\tau})/\widehat p$; center the $\bm u_i$.
\STATE Evaluate $\widehat\tau(t)=\bm b_n(t)^\top\widehat{\bm\tau}$ and $\widehat\sigma^2(t)=n^{-1}\sum_i\{\bm b_n(t)^\top\bm u_i\}^2$ on $\cT_n$.
\FOR{$b=1,\ldots,B$}
\STATE Draw $\xi_i^{(b)}\sim N(0,1)$ independently; set $\bm g_b=n^{-1/2}\sum_i\xi_i^{(b)}\bm u_i$.
\STATE Set $T_b^*=\max_{t\in\cT_n}|\bm b_n(t)^\top\bm g_b/\widehat\sigma(t)|$.
\ENDFOR
\STATE Return the curve and bands using the empirical $(1-\alpha)$ quantile of $\{T_b^*\}$.
\end{algorithmic}
\end{algorithm}

For computational efficiency, $\bm g_b$ can equivalently be sampled from $N(0,n^{-1}\sum_i\bm u_i\bm u_i^\top)$; this is the same conditional multiplier law. Algorithm~\ref{alg:dr_scb} uses efficient-score inference under Theorem~\ref{thm:band}. For regular finite-dimensional working models, Supplementary Propositions~S1--S2 give a foldwise variance correction accommodating single-model misspecification; the point estimator is unchanged.
\FloatBarrier
\section{Simulation Study}\label{sec:simulation}
The simulation study examines estimation error, coverage, and interval width under nuisance misspecification, heterogeneous effects, limited overlap, and discrete observation.

\subsection{Data Generation}\label{ssec:sim_dgp}
The covariates $X_1,\ldots,X_5$ are independent $\operatorname{Unif}[-1,1]$ variables. Define
\[
 Z_X=X_1X_2+\tfrac12(X_3^2-1/3)+\tfrac12\sin(\pi X_4),\qquad
 \ell_L(X)=0.7X_1-0.5X_2,\quad \ell_F(X)=\ell_L(X)+0.7Z_X.
\]
For a specified index $\ell$, let $\pi_0(X)=\expit\{\gamma_0+\ell(X)\}$ and $D\mid X\sim\operatorname{Bernoulli}\{\pi_0(X)\}$. The intercept $\gamma_0$ solves the population equation $\E\pi_0(X)=1/2$ and is fixed across replications within a scenario. The limited-overlap index is $\ell_W(X)=4X_1$, for which $\gamma_0=0$ and $\pi_0\in[\expit(-4),\expit(4)]$. Compact covariate support ensures strict overlap in all scenarios.

On $[0,1]$, use the orthonormal Fourier basis $b_0(t)=1$, $b_{2k-1}(t)=\sqrt2\sin(2\pi kt)$ and $b_{2k}(t)=\sqrt2\cos(2\pi kt)$, $k=1,\ldots,6$. Let $U_i$ and $V_i$ be independent centered Gaussian curves in this 13-dimensional span, independent of $X_i,D_i$, with independent coefficients of variances $\lambda_0=0.25$ and $\lambda_{2k-1}=\lambda_{2k}=0.20/k^2$. Define
\begin{align*}
 Y_{i0}^0(t)&=2+0.5\sin(2\pi t)+0.2\cos(4\pi t)+0.2X_{i1}\cos(2\pi t)+U_i(t),\\
 Y_{i1}^0(t)&=Y_{i0}^0(t)+\mu_0(X_i)(t)+V_i(t),\\
 \mu_L(x)(t)&=0.15+0.10\sin(2\pi t)+0.30x_1\sin(2\pi t)+0.20x_2\cos(2\pi t),\\
 \mu_F(x)(t)&=\mu_L(x)(t)+0.25Z_x\{1+\cos(4\pi t)\}.
\end{align*}
Set $Y_{i0}^1=Y_{i0}^0$ and $Y_{i1}^1=Y_{i1}^0+\tau_i$. Because $V_i$ is independent of $D_i$ conditional on $X_i$, conditional functional parallel trends holds exactly. The common $U_i$ cancels in the period difference.

Two effect shapes are used. The smooth shape is $\tau_S(t)=0.6\sin(2\pi t)$. For a localized peak and dip, define
\[
 F_6(u)=\frac17\left\{1+2\sum_{k=1}^6(1-k/7)\cos(2\pi ku)\right\},\qquad
 \tau_B(t)=0.8F_6(t-0.30)-0.5F_6(t-0.75).
\]
Effects are homogeneous, $\tau_i=\tau_S$ or $\tau_B$, except in the heterogeneous scenario, where $\tau_i(t)=(1+0.4X_{i1})\tau_B(t)$. Its population target is
\[
 \tau_0(t)=\tau_B(t)\left\{1+0.4\frac{\E[X_1\pi_0(X)]}{\E\pi_0(X)}\right\}.
\]
Population integrals and intercepts use deterministic numerical integration independent of the Monte Carlo samples, with a tolerance of $10^{-5}$ under doubled resolution. The target is fixed across replications.

\subsection{Observation Scheme and Curve Fitting}\label{ssec:sim_recon}
For $n\in\{200,400,800\}$, observe each period at $t_j=j/m_n$, $j=0,\ldots,m_n-1$, with
\[
 m_n=\lceil n^2/100\rceil\in\{400,1600,6400\},\qquad
 \varepsilon_{iTj}\overset{\mathrm{i.i.d.}}\sim N(0,0.20^2).
\]
Each curve is estimated by least squares in the 13-dimensional Fourier span. The generated curves lie in this span, the discrete Gram matrix is the identity, and $m_n/\{n\log(n+1)\}\to\infty$, satisfying Proposition~\ref{prop:curve_rates}(i) without truncation bias.

The inference grid is $\cT_n=\{j/n:j=0,\ldots,n\}$, with $M_n=n+1$ and vanishing mesh. A doubled-resolution grid assesses numerical sensitivity. Feasible methods use $\widehat Y_{iT}$, whereas Oracle uses $Y_{iT}$; paired runs of the feasible methods on the true curves measure the effect of estimating them. Sensitivity analyses multiply $m_n$ by $1/4$ and $4$ and use measurement-noise standard deviations $0.10$ and $0.40$. A separate nine-dimensional fit retains frequencies $k\leq4$. The resulting truncation bias generally violates Assumption~\ref{ass:asym_recon}; evaluating errors and coverage against the full $\tau_0$ measures its effect.

\subsection{Estimators and Nuisance Learners}\label{ssec:sim_estimators}
Let $z_L(x)=(1,x_1,\ldots,x_5)^\top$ and augment it with $x_1x_2$, $x_3^2-1/3$ and $\sin(\pi x_4)$ to obtain $z_F(x)$. Assign $K=5$ balanced folds independently of the observations. On each training fold, fit propensity by Firth logistic regression \citep{firth1993bias}, maximizing the logistic log-likelihood plus one half of the log determinant of its information matrix. Fit the functional coefficients by group-specific least squares using the specified features. The same fitting rules apply in every principal scenario: CF--DR and IPW share propensity fits, and CF--DR and OR share control regressions.
\begin{itemize}[leftmargin=*]
\item \textbf{CF--DR:} the common-denominator estimator \eqref{eq:tauhat}, computed through \eqref{eq:coefficient_score}.
\item \textbf{OR:} the treated average of $\widehat\mu_1(X_i)-\widehat\mu_0(X_i)$, with both regressions cross-fitted.
\item \textbf{IPW:} the average of $\{D_i-(1-D_i)\widehat\pi(X_i)/(1-\widehat\pi(X_i))\}\widehat{\Delta Y}_i/\widehat p$.
\item \textbf{Naive DiD:} the difference between the treated and control sample mean changes. It is an unadjusted baseline and need not identify the conditional-trends ATT.
\item \textbf{Oracle:} \eqref{eq:tauhat} with true changes $\Delta Y_i$ and true $\pi_0,\mu_0$, estimating $p_0$ by $\widehat p$. Propensities are used without clipping.
\end{itemize}
The primary fitted-propensity clipping level is $c=0.001$; sensitivities use $0.005$ and $0.02$. In S6, $\expit(-4)\approx0.018<0.02$, so the larger threshold truncates true propensities on a set of positive probability. These sensitivities report point estimates and weight summaries, since Supplementary Proposition~S1 requires asymptotically inactive clipping.

A supplementary nonparametric experiment uses $X\sim\operatorname{Unif}[-1,1]$, $Z_X=\sin(\pi X)+\tfrac12(X^2-1/3)$, $\ell(X)=0.7X+0.7Z_X$, and
\[
 \mu_0(X)(t)=0.15+0.10\sin(2\pi t)+0.30X\sin(2\pi t)+0.25Z_X\{1+\cos(4\pi t)\}.
\]
The curve and noise construction is unchanged, with baseline covariate $X_1$ replaced by $X$ and effect $\tau_B$. Apply the stabilized bin learners of Proposition~\ref{prop:histogram} to $(X+1)/2$, with $J_X=\lceil n^{1/3}\rceil$. A training bin with $N_j$ units and $N_{1j}$ treated units has propensity estimate $(N_{1j}+1/2)/(N_j+1)$ before clipping. Outcome fits are group-specific bin means; an empty group bin uses that group's mean over the same training fold.

\subsection{Inference and Performance}\label{ssec:sim_metrics}
We use $R=1000$ independent replications per scenario and sample size. Pointwise intervals use standard-normal quantiles, and simultaneous bands use $B=2000$ Gaussian multiplier draws, both at the nominal 95\% level. Methods share the same simulated samples and folds. Failures count as noncoverage; none occurred.

For the principal CF--DR, OR, and IPW comparisons, inference uses foldwise influence contributions that account for nuisance estimation, as established in Supplementary Propositions~S1--S2. Naive DiD uses its two-group mean influence function, and Oracle uses \eqref{eq:eif_point}. The nonparametric experiment uses
Algorithm~\ref{alg:dr_scb} for CF--DR inference; its comparators are evaluated for point estimation. Coverage is evaluated against the same population fATT $\tau_0$ for every method, including methods whose probability limits differ from this target.

Let $\mathcal S$ index the successful replications, $R_s=|\mathcal S|$, and $e_r(t)=\widehat\tau_r(t)-\tau_0(t)$. Here $R_s=R$. We report
\begin{align*}
 \mathrm{Bias}&=\frac1{R_s}\sum_{r\in\mathcal S}\int_0^1 e_r(t)\,dt,&
 \mathrm{MAE}&=\frac1{R_s}\sum_{r\in\mathcal S}\int_0^1|e_r(t)|\,dt,\\
 \overline{\mathrm{ISE}}&=\frac1{R_s}\sum_{r\in\mathcal S}\int_0^1e_r(t)^2\,dt,&
 \mathrm{SupErr}&=\frac1{R_s}\sum_{r\in\mathcal S}\sup_{t\in[0,1]}|e_r(t)|.
\end{align*}
Pointwise RMSE is $\{R_s^{-1}\sum_{r\in\mathcal S}e_r(t)^2\}^{1/2}$. Integrals use the trapezoidal rule and suprema use the inference grid.

Inferential summaries comprise pointwise coverage at $t=0.25,0.50,0.75$, simultaneous coverage over the entire inference grid, and interval and band widths averaged over the domain and replications. The Monte Carlo standard error of coverage near 95\% is approximately 0.7 percentage points.

For paired coverage indicators $I_{a,r}$ and $I_{b,r}$
for methods $a$ and $b$ in replication $r$, let
$d_r=I_{a,r}-I_{b,r}$ and
$\bar d=R^{-1}\sum_{r=1}^R d_r$.
The coverage difference is $100\bar d$ percentage points,
with Monte Carlo standard error $100\{\sum_r(d_r-\bar d)^2/[R(R-1)]\}^{1/2}$. Supplementary Tables~S1--S3 report standardized bias, ratios of estimated standard errors to Monte Carlo standard deviations, and coverage MCSEs.

\subsection{Scenarios}\label{ssec:sim_scenarios}
Table~\ref{tab:sim_scenarios} defines the six configurations. S1 and S4--S6 use correctly specified finite-dimensional nuisance models; S2--S3 assess single-model misspecification. All use the inference rule in Section~\ref{ssec:sim_metrics}. Nonlinearity concerns covariate dependence. In S3 the omitted term vanishes at $t=0.25,0.75$ but is nonzero at $t=0.50$, so results at all three points and over the full curve are reported.
\begin{table}[H]
\centering
\caption{Simulation configurations. Working features are shown for propensity and outcome regression, respectively. All scenarios use ultra-dense observations.}
\label{tab:sim_scenarios}
\small
\begin{tabular}{llllll}
\toprule
Scenario & True index & True trend & Working features & Effect & Purpose\\
\midrule
S1 & $\ell_L$ & $\mu_L$ & $z_L,z_L$ & $\tau_S$ & Parametric baseline\\
S2 & $\ell_F$ & $\mu_L$ & $z_L,z_L$ & $\tau_S$ & Propensity misspecification\\
S3 & $\ell_L$ & $\mu_F$ & $z_L,z_L$ & $\tau_B$ & Outcome misspecification\\
S4 & $\ell_F$ & $\mu_F$ & $z_F,z_F$ & $\tau_B$ & Nonlinear nuisances\\
S5 & $\ell_F$ & $\mu_F$ & $z_F,z_F$ & $(1+0.4X_1)\tau_B$ & Heterogeneity\\
S6 & $\ell_W$ & $\mu_F$ & $z_F,z_F$ & $\tau_B$ & Limited overlap\\
\bottomrule
\end{tabular}
\end{table}

Supplementary comparisons vary nuisance fits and observation schemes. Paired corrected and uncorrected bands measure the effect of the variance correction on coverage and width; analyses of $\widehat Y_{iT}$ and $Y_{iT}$ measure the effect of curve estimation. An example band uses the first successful replication.

\subsection{Results}\label{ssec:sim_results}
CF--DR's MAE, ISE, and SupErr decrease with sample size in every scenario, including S2--S3 where one nuisance model is misspecified (Table~\ref{tab:sim_accuracy}). In S1--S5, its MAE falls from 0.122--0.140 at $n=200$ to 0.056--0.057 at $n=800$. At $n=800$, CF--DR and correctly specified OR have similar MAE in S1--S2. Figure~\ref{fig:sim_mean_n800} shows localized OR bias in S3 despite similar integrated error.

\begin{table}[!htbp]
\centering
\caption{Estimation accuracy in S1--S6.}
\label{tab:sim_accuracy}
\begingroup
\footnotesize
\setlength{\tabcolsep}{2pt}
\renewcommand{\arraystretch}{1}
\begin{tabular*}{\textwidth}{@{\extracolsep{\fill}}ll@{\hspace{7pt}}*{12}{r}@{}}
\toprule
Scenario & Method & \multicolumn{3}{c}{Bias} & \multicolumn{3}{c}{MAE} & \multicolumn{3}{c}{$\overline{\mathrm{ISE}}$} & \multicolumn{3}{c}{SupErr} \\
\cmidrule(lr){3-5} \cmidrule(lr){6-8} \cmidrule(lr){9-11} \cmidrule(lr){12-14}
 & & $n{=}200$ & $n{=}400$ & $n{=}800$ & $n{=}200$ & $n{=}400$ & $n{=}800$ & $n{=}200$ & $n{=}400$ & $n{=}800$ & $n{=}200$ & $n{=}400$ & $n{=}800$ \\
\midrule
S1 & CF--DR & 0.003 & 0.001 & 0.000 & 0.122 & 0.082 & 0.057 & 0.023 & 0.010 & 0.005 & 0.289 & 0.191 & 0.133 \\
 & OR & 0.001 & 0.000 & 0.000 & 0.114 & 0.079 & 0.056 & 0.020 & 0.010 & 0.005 & 0.269 & 0.186 & 0.131 \\
 & IPW & -0.011 & -0.006 & -0.003 & 0.123 & 0.082 & 0.057 & 0.024 & 0.011 & 0.005 & 0.290 & 0.193 & 0.133 \\
 & Naive DiD & -0.001 & 0.000 & 0.000 & 0.112 & 0.086 & 0.068 & 0.020 & 0.012 & 0.007 & 0.261 & 0.197 & 0.153 \\
 & Oracle & 0.000 & 0.000 & 0.000 & 0.111 & 0.078 & 0.056 & 0.019 & 0.010 & 0.005 & 0.262 & 0.184 & 0.130 \\
\addlinespace[4pt]
S2 & CF--DR & 0.002 & 0.001 & -0.001 & 0.122 & 0.084 & 0.057 & 0.024 & 0.011 & 0.005 & 0.290 & 0.198 & 0.133 \\
 & OR & 0.001 & 0.000 & -0.001 & 0.112 & 0.080 & 0.056 & 0.020 & 0.010 & 0.005 & 0.267 & 0.188 & 0.131 \\
 & IPW & -0.013 & -0.007 & -0.006 & 0.125 & 0.086 & 0.059 & 0.025 & 0.012 & 0.005 & 0.297 & 0.202 & 0.137 \\
 & Naive DiD & 0.001 & 0.000 & -0.002 & 0.111 & 0.087 & 0.068 & 0.019 & 0.012 & 0.007 & 0.260 & 0.198 & 0.150 \\
 & Oracle & 0.000 & 0.000 & -0.001 & 0.112 & 0.081 & 0.057 & 0.020 & 0.010 & 0.005 & 0.264 & 0.189 & 0.131 \\
\addlinespace[4pt]
S3 & CF--DR & 0.001 & 0.001 & 0.000 & 0.124 & 0.082 & 0.056 & 0.024 & 0.011 & 0.005 & 0.293 & 0.193 & 0.134 \\
 & OR & -0.007 & -0.009 & -0.009 & 0.117 & 0.080 & 0.056 & 0.021 & 0.010 & 0.005 & 0.275 & 0.189 & 0.133 \\
 & IPW & -0.011 & -0.006 & -0.003 & 0.124 & 0.083 & 0.057 & 0.025 & 0.011 & 0.005 & 0.294 & 0.194 & 0.135 \\
 & Naive DiD & 0.000 & 0.001 & 0.000 & 0.114 & 0.087 & 0.069 & 0.020 & 0.012 & 0.007 & 0.263 & 0.196 & 0.153 \\
 & Oracle & 0.002 & 0.000 & 0.001 & 0.111 & 0.078 & 0.054 & 0.020 & 0.010 & 0.005 & 0.262 & 0.184 & 0.129 \\
\addlinespace[4pt]
S4 & CF--DR & 0.007 & 0.004 & -0.002 & 0.139 & 0.087 & 0.057 & 0.032 & 0.012 & 0.005 & 0.327 & 0.203 & 0.134 \\
 & OR & 0.004 & 0.003 & -0.002 & 0.120 & 0.082 & 0.056 & 0.023 & 0.011 & 0.005 & 0.281 & 0.192 & 0.130 \\
 & IPW & -0.020 & -0.008 & -0.008 & 0.139 & 0.087 & 0.058 & 0.031 & 0.012 & 0.005 & 0.326 & 0.204 & 0.135 \\
 & Naive DiD & 0.045 & 0.045 & 0.039 & 0.122 & 0.098 & 0.078 & 0.023 & 0.015 & 0.009 & 0.282 & 0.221 & 0.173 \\
 & Oracle & 0.004 & 0.004 & -0.002 & 0.113 & 0.080 & 0.055 & 0.020 & 0.010 & 0.005 & 0.263 & 0.188 & 0.128 \\
\addlinespace[4pt]
S5 & CF--DR & 0.000 & -0.001 & -0.002 & 0.140 & 0.084 & 0.057 & 0.032 & 0.011 & 0.005 & 0.332 & 0.198 & 0.137 \\
 & OR & 0.002 & -0.001 & -0.002 & 0.119 & 0.079 & 0.056 & 0.022 & 0.010 & 0.005 & 0.277 & 0.187 & 0.132 \\
 & IPW & -0.027 & -0.013 & -0.008 & 0.141 & 0.086 & 0.058 & 0.032 & 0.012 & 0.005 & 0.329 & 0.202 & 0.138 \\
 & Naive DiD & 0.043 & 0.040 & 0.039 & 0.122 & 0.095 & 0.079 & 0.023 & 0.014 & 0.009 & 0.280 & 0.213 & 0.175 \\
 & Oracle & 0.002 & -0.001 & -0.002 & 0.112 & 0.078 & 0.056 & 0.020 & 0.010 & 0.005 & 0.263 & 0.183 & 0.132 \\
\addlinespace[4pt]
S6 & CF--DR & 0.016 & 0.004 & 0.003 & 0.667 & 0.304 & 0.165 & 2.150 & 0.311 & 0.047 & 1.578 & 0.718 & 0.392 \\
 & OR & 0.000 & -0.001 & -0.002 & 0.192 & 0.134 & 0.091 & 0.059 & 0.028 & 0.013 & 0.458 & 0.316 & 0.215 \\
 & IPW & -0.074 & -0.024 & -0.009 & 0.586 & 0.286 & 0.165 & 1.800 & 0.242 & 0.048 & 1.363 & 0.668 & 0.387 \\
 & Naive DiD & 0.001 & 0.000 & -0.001 & 0.182 & 0.171 & 0.163 & 0.048 & 0.040 & 0.035 & 0.378 & 0.335 & 0.300 \\
 & Oracle & -0.006 & -0.003 & 0.001 & 0.253 & 0.184 & 0.131 & 0.119 & 0.058 & 0.028 & 0.595 & 0.430 & 0.310 \\
\bottomrule
\end{tabular*}
\endgroup
\end{table}

In S1--S5, CF--DR simultaneous coverage is 97.8--98.4\% at $n=200$, compared with 88.8--93.0\% for OR, with wider, conservative bands (Tables~\ref{tab:sim_coverage}--\ref{tab:sim_width}). At $n=800$, CF--DR coverage is 95.6--97.3\% and its SEratio at $t=0.25$ is 1.00--1.06, down from 1.12--1.21 at $n=200$ (Supplementary Table~S1). Naive DiD coverage falls to 67.6\% in S4. Figure~\ref{fig:sim_coverage} displays CF--DR and Oracle coverage with Monte Carlo uncertainty. CF--DR mean band widths contract from 1.001--1.242 to 0.419--0.441 between $n=200$ and $800$.

Limited overlap in S6 increases estimation error and uncertainty, as reflected in the elevated RMSE in Figure~\ref{fig:sim_rmse_n800} and the wider example band in Figure~\ref{fig:sim_example_scb}. CF--DR simultaneous coverage is 97.7\% at $n=200$ and 95.3\% at $n=800$, whereas OR covers 86.3\% at $n=200$ (Table~\ref{tab:sim_coverage}). OR has smaller MAE and narrower bands despite its lower coverage at $n=200$. The foldwise correction raises CF--DR coverage by 17.6 and 5.6 percentage points at these sample sizes, with paired MCSEs of 1.20 and 0.74, respectively, without changing the point estimates (Supplementary Table~S5).

Using the true curves changes CF--DR coverage by at most 0.5 percentage points, and doubling the inference grid leaves coverage unchanged (Supplementary Tables~S6 and S8). With nine basis functions in S4, however, full-target coverage declines from 95.9\% to 60.1\% as $n$ increases from 200 to 800, while projected-target coverage stays at 96.9--98.5\% (Supplementary Table~S9). The truncation bias persists as sampling variability decreases, violating the error-rate condition in Assumption~\ref{ass:asym_recon}.

In the nonparametric learner experiment, CF--DR simultaneous coverage remains between 94.8\% and 95.7\%, and its MAE decreases from 0.146 to 0.064 (Supplementary Table~S11 and Figure~S2).

\begin{table}[h]
\centering
\caption{Pointwise and simultaneous coverage (\%) of nominal 95\% intervals and bands.}
\label{tab:sim_coverage}
\begingroup
\footnotesize
\setlength{\tabcolsep}{2pt}
\renewcommand{\arraystretch}{1}
\begin{tabular*}{\textwidth}{@{\extracolsep{\fill}}ll@{\hspace{7pt}}*{12}{r}@{}}
\toprule
Scenario & Method & \multicolumn{3}{c}{$t=0.25$} & \multicolumn{3}{c}{$t=0.50$} & \multicolumn{3}{c}{$t=0.75$} & \multicolumn{3}{c}{SCB} \\
\cmidrule(lr){3-5} \cmidrule(lr){6-8} \cmidrule(lr){9-11} \cmidrule(lr){12-14}
 & & $n{=}200$ & $n{=}400$ & $n{=}800$ & $n{=}200$ & $n{=}400$ & $n{=}800$ & $n{=}200$ & $n{=}400$ & $n{=}800$ & $n{=}200$ & $n{=}400$ & $n{=}800$ \\
\midrule
S1 & CF--DR & 96.8 & 95.9 & 96.4 & 97.5 & 96.4 & 95.6 & 97.7 & 95.3 & 95.5 & 97.8 & 95.6 & 96.2 \\
 & OR & 93.6 & 94.3 & 94.8 & 93.9 & 95.0 & 94.3 & 94.9 & 93.3 & 95.4 & 91.9 & 93.4 & 94.7 \\
 & IPW & 95.0 & 95.3 & 95.8 & 96.1 & 96.1 & 94.9 & 96.5 & 94.9 & 95.5 & 96.2 & 94.3 & 95.2 \\
 & Naive DiD & 92.4 & 88.3 & 82.1 & 95.8 & 95.0 & 91.6 & 92.0 & 87.0 & 81.6 & 90.3 & 86.8 & 80.9 \\
 & Oracle & 94.0 & 95.4 & 95.8 & 95.7 & 96.1 & 94.2 & 96.6 & 94.1 & 95.1 & 95.5 & 95.0 & 94.9 \\
\addlinespace[4pt]
S2 & CF--DR & 97.3 & 95.7 & 95.9 & 96.6 & 95.3 & 95.5 & 97.7 & 96.7 & 95.3 & 98.1 & 95.8 & 96.2 \\
 & OR & 94.3 & 93.2 & 94.7 & 93.4 & 94.5 & 94.8 & 93.9 & 95.2 & 94.6 & 93.0 & 93.0 & 94.5 \\
 & IPW & 95.7 & 95.1 & 94.7 & 95.4 & 95.3 & 94.8 & 96.6 & 96.1 & 94.9 & 96.7 & 94.3 & 94.8 \\
 & Naive DiD & 92.5 & 87.8 & 83.1 & 94.4 & 94.0 & 92.1 & 92.2 & 88.3 & 81.0 & 90.8 & 86.8 & 82.0 \\
 & Oracle & 95.0 & 94.2 & 94.8 & 95.3 & 95.3 & 94.9 & 95.3 & 96.2 & 95.4 & 96.0 & 94.5 & 95.6 \\
\addlinespace[4pt]
S3 & CF--DR & 98.1 & 95.3 & 94.5 & 97.2 & 97.0 & 95.8 & 97.2 & 97.0 & 96.3 & 98.3 & 96.4 & 95.6 \\
 & OR & 95.0 & 93.7 & 93.2 & 93.3 & 94.8 & 94.4 & 93.9 & 94.9 & 95.7 & 92.1 & 92.2 & 93.5 \\
 & IPW & 96.7 & 94.5 & 94.3 & 95.8 & 95.8 & 95.6 & 95.7 & 96.3 & 95.9 & 96.0 & 95.4 & 95.1 \\
 & Naive DiD & 92.4 & 90.0 & 82.1 & 94.6 & 94.6 & 93.2 & 92.7 & 89.4 & 82.6 & 92.3 & 88.8 & 80.4 \\
 & Oracle & 96.0 & 93.6 & 94.8 & 95.4 & 95.8 & 96.0 & 94.8 & 95.1 & 95.6 & 96.1 & 94.0 & 94.9 \\
\addlinespace[4pt]
S4 & CF--DR & 97.9 & 97.1 & 95.6 & 98.1 & 96.6 & 95.8 & 98.6 & 96.7 & 96.8 & 98.4 & 97.4 & 97.3 \\
 & OR & 92.8 & 94.3 & 93.9 & 92.5 & 93.1 & 95.0 & 94.2 & 94.4 & 95.8 & 88.8 & 92.9 & 95.2 \\
 & IPW & 96.5 & 95.0 & 94.7 & 96.6 & 95.4 & 95.5 & 96.4 & 95.6 & 96.2 & 96.4 & 95.6 & 96.4 \\
 & Naive DiD & 90.4 & 89.4 & 84.9 & 84.8 & 76.1 & 61.0 & 93.7 & 89.2 & 83.5 & 87.6 & 77.6 & 67.6 \\
 & Oracle & 93.7 & 94.3 & 94.9 & 94.8 & 94.7 & 95.5 & 95.6 & 95.5 & 95.9 & 95.8 & 94.4 & 96.3 \\
\addlinespace[4pt]
S5 & CF--DR & 98.3 & 97.4 & 95.5 & 97.9 & 97.3 & 96.1 & 98.8 & 97.5 & 96.8 & 98.1 & 98.1 & 96.5 \\
 & OR & 93.6 & 95.3 & 93.8 & 93.6 & 94.8 & 95.0 & 94.9 & 94.4 & 96.1 & 91.1 & 93.4 & 93.8 \\
 & IPW & 96.6 & 96.4 & 94.3 & 97.5 & 95.7 & 95.3 & 97.9 & 96.2 & 96.5 & 97.1 & 96.5 & 95.9 \\
 & Naive DiD & 91.5 & 89.4 & 84.1 & 86.3 & 79.6 & 63.9 & 92.3 & 89.3 & 84.4 & 87.7 & 81.7 & 66.2 \\
 & Oracle & 95.1 & 95.7 & 94.4 & 96.2 & 95.2 & 95.4 & 95.9 & 95.3 & 96.5 & 95.3 & 95.0 & 95.5 \\
\addlinespace[4pt]
S6 & CF--DR & 98.8 & 98.0 & 96.2 & 98.2 & 97.6 & 97.0 & 98.3 & 96.9 & 96.8 & 97.7 & 95.3 & 95.3 \\
 & OR & 91.9 & 92.6 & 94.3 & 91.7 & 93.6 & 95.2 & 93.2 & 93.5 & 95.0 & 86.3 & 91.0 & 93.8 \\
 & IPW & 96.0 & 96.3 & 95.5 & 97.8 & 97.1 & 97.2 & 97.2 & 96.6 & 96.6 & 95.2 & 94.1 & 94.6 \\
 & Naive DiD & 52.2 & 23.7 & 4.6 & 95.5 & 94.3 & 94.9 & 53.2 & 23.3 & 3.3 & 49.9 & 17.3 & 1.3 \\
 & Oracle & 97.9 & 96.1 & 95.7 & 97.5 & 96.7 & 97.3 & 96.7 & 95.8 & 96.4 & 97.9 & 98.3 & 97.9 \\
\bottomrule
\end{tabular*}
\par\smallskip
\begin{minipage}{\textwidth}
\footnotesize SCB denotes simultaneous coverage over the inference grid. Entry-specific MCSEs are in Supplementary Table~S3.
\end{minipage}
\endgroup
\end{table}

\begin{table}[!]
\centering
\caption{Average pointwise interval and simultaneous band widths.}
\label{tab:sim_width}
\begingroup
\footnotesize
\setlength{\tabcolsep}{2pt}
\renewcommand{\arraystretch}{1}
\begin{tabular*}{\textwidth}{@{\extracolsep{\fill}}ll*{6}{r}@{}}
\toprule
Scenario & Method & \multicolumn{3}{c}{Pointwise Interval Width} & \multicolumn{3}{c}{SCB Width} \\
\cmidrule(lr){3-5} \cmidrule(lr){6-8}
 & & $n{=}200$ & $n{=}400$ & $n{=}800$ & $n{=}200$ & $n{=}400$ & $n{=}800$ \\
\midrule
S1 & CF--DR & 0.682 & 0.426 & 0.285 & 1.001 & 0.627 & 0.419 \\
 & OR & 0.537 & 0.381 & 0.269 & 0.790 & 0.560 & 0.396 \\
 & IPW & 0.631 & 0.412 & 0.281 & 0.926 & 0.607 & 0.413 \\
 & Naive DiD & 0.515 & 0.365 & 0.258 & 0.758 & 0.536 & 0.379 \\
 & Oracle & 0.543 & 0.384 & 0.271 & 0.799 & 0.565 & 0.399 \\
\addlinespace[4pt]
S2 & CF--DR & 0.691 & 0.430 & 0.287 & 1.015 & 0.633 & 0.423 \\
 & OR & 0.537 & 0.380 & 0.270 & 0.790 & 0.560 & 0.397 \\
 & IPW & 0.646 & 0.421 & 0.287 & 0.948 & 0.620 & 0.422 \\
 & Naive DiD & 0.515 & 0.364 & 0.258 & 0.758 & 0.535 & 0.379 \\
 & Oracle & 0.553 & 0.388 & 0.275 & 0.814 & 0.571 & 0.404 \\
\addlinespace[4pt]
S3 & CF--DR & 0.689 & 0.432 & 0.288 & 1.010 & 0.635 & 0.424 \\
 & OR & 0.544 & 0.385 & 0.272 & 0.800 & 0.566 & 0.401 \\
 & IPW & 0.636 & 0.416 & 0.284 & 0.933 & 0.612 & 0.417 \\
 & Naive DiD & 0.523 & 0.370 & 0.261 & 0.768 & 0.543 & 0.384 \\
 & Oracle & 0.546 & 0.383 & 0.271 & 0.803 & 0.564 & 0.399 \\
\addlinespace[4pt]
S4 & CF--DR & 0.844 & 0.469 & 0.300 & 1.233 & 0.689 & 0.441 \\
 & OR & 0.548 & 0.386 & 0.273 & 0.806 & 0.569 & 0.402 \\
 & IPW & 0.738 & 0.444 & 0.293 & 1.079 & 0.652 & 0.431 \\
 & Naive DiD & 0.522 & 0.369 & 0.262 & 0.767 & 0.543 & 0.384 \\
 & Oracle & 0.551 & 0.390 & 0.275 & 0.810 & 0.573 & 0.405 \\
\addlinespace[4pt]
S5 & CF--DR & 0.849 & 0.468 & 0.299 & 1.242 & 0.688 & 0.441 \\
 & OR & 0.548 & 0.386 & 0.273 & 0.807 & 0.568 & 0.402 \\
 & IPW & 0.746 & 0.443 & 0.293 & 1.091 & 0.652 & 0.431 \\
 & Naive DiD & 0.524 & 0.371 & 0.262 & 0.770 & 0.545 & 0.386 \\
 & Oracle & 0.551 & 0.390 & 0.276 & 0.811 & 0.573 & 0.406 \\
\addlinespace[4pt]
S6 & CF--DR & 4.506 & 1.788 & 0.894 & 6.386 & 2.551 & 1.291 \\
 & OR & 0.872 & 0.619 & 0.441 & 1.282 & 0.910 & 0.649 \\
 & IPW & 3.453 & 1.539 & 0.854 & 4.729 & 2.168 & 1.227 \\
 & Naive DiD & 0.520 & 0.368 & 0.261 & 0.765 & 0.541 & 0.383 \\
 & Oracle & 1.162 & 0.851 & 0.628 & 1.666 & 1.229 & 0.913 \\
\bottomrule
\end{tabular*}
\endgroup
\end{table}

\FloatBarrier
\begin{figure}[p]
\centering
\includegraphics[width=\textwidth]{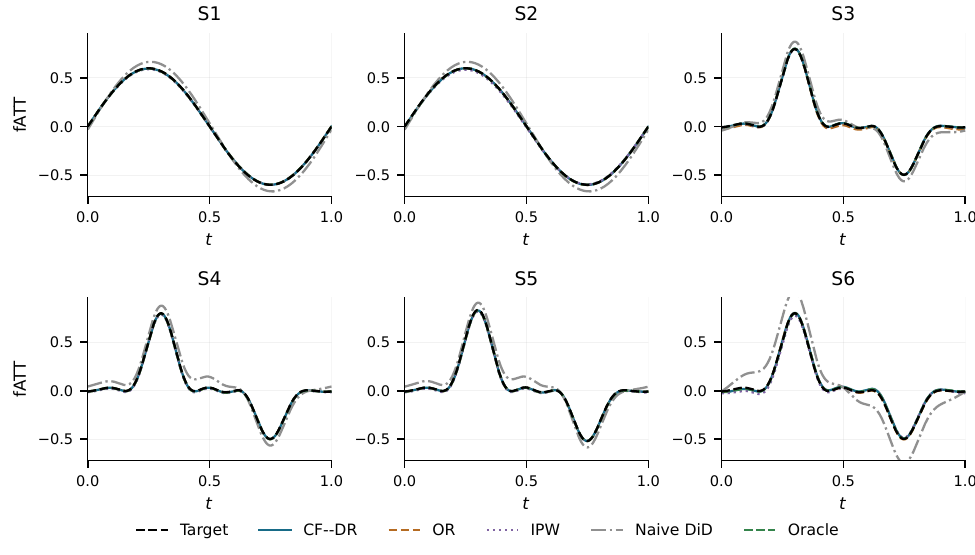}
\caption{Mean estimated fATT curves at $n=800$. The dashed curve is the population target.}
\label{fig:sim_mean_n800}
\end{figure}

\begin{figure}[p]
\centering
\includegraphics[width=\textwidth]{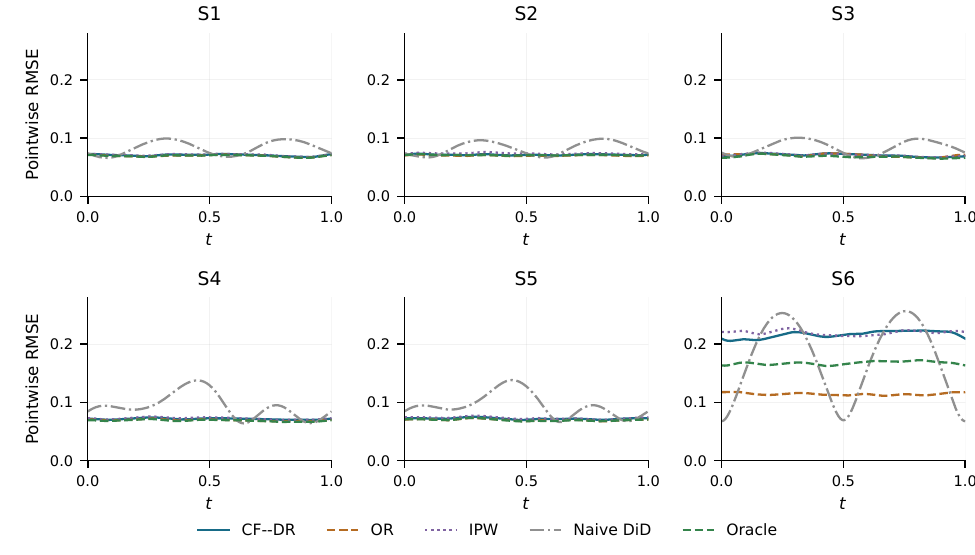}
\caption{Pointwise root mean squared error at $n=800$.}
\label{fig:sim_rmse_n800}
\end{figure}

\begin{figure}[p]
\centering
\includegraphics[width=\textwidth]{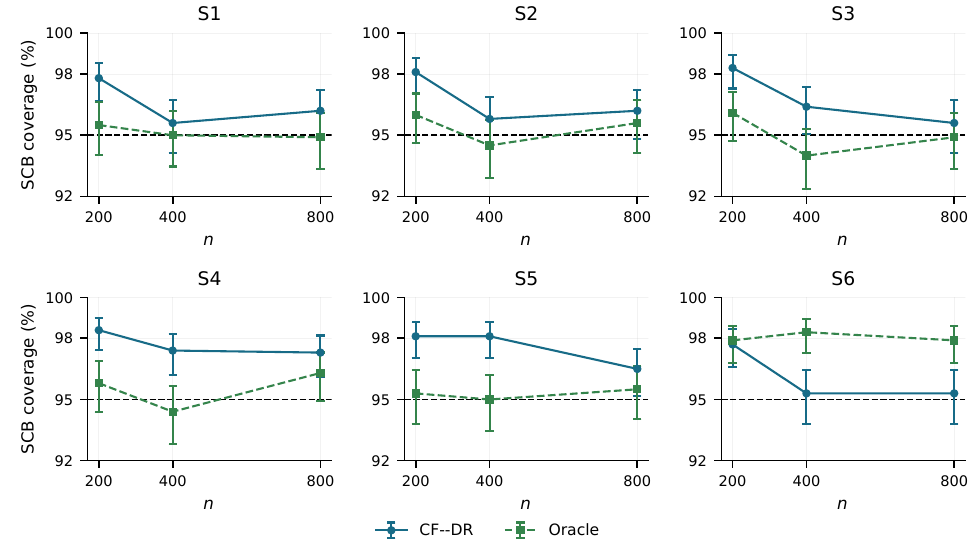}
\caption{Simultaneous confidence-band coverage for CF--DR and Oracle. Error bars are 95\% Wilson Monte Carlo intervals and the dashed line marks the nominal 95\% level.}
\label{fig:sim_coverage}
\end{figure}

\begin{figure}[p]
\centering
\includegraphics[width=\textwidth]{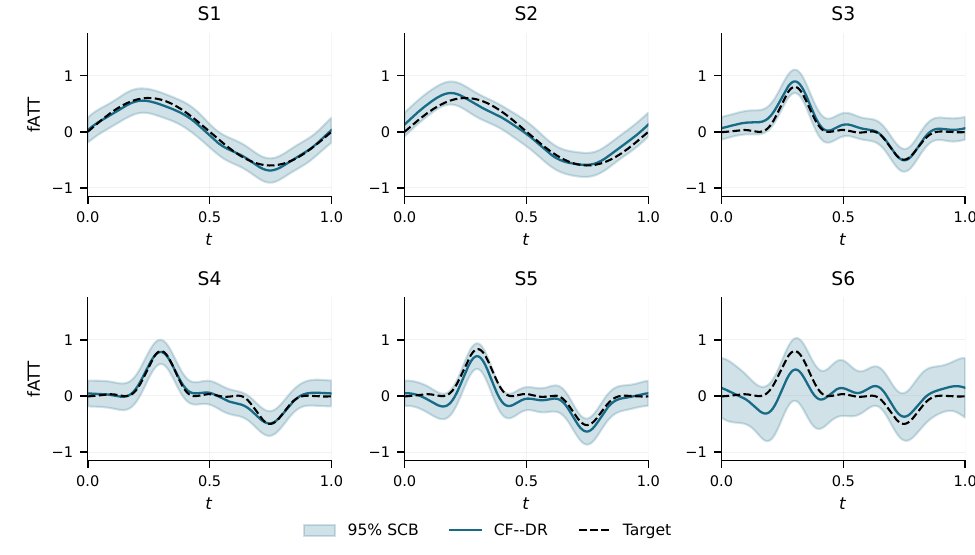}
\caption{CF--DR estimate and 95\% simultaneous confidence band from the first successful replication in each scenario at $n=800$.}
\label{fig:sim_example_scb}
\end{figure}

\FloatBarrier
\section{Application}\label{sec:application}
\subsection{Dynamic electricity pricing and household load profiles}
We apply the method to half-hourly electricity consumption from the Low Carbon London study
\citep{ukpn2014lcl,schofield2014residential}.\footnote{Public consumption records and tariff documentation: \href{https://data.london.gov.uk/dataset/smartmeter-energy-consumption-data-in-london-households-vqm0d}{London Datastore, SmartMeter Energy Consumption Data in London Households}. Recruitment is described in \citet{schofield2015guide}; its separately curated UK Data Service release is not used here.} Households with $D_i=1$ received a dynamic time-of-use tariff during 2013, whereas comparison households ($D_i=0$) did not; participation was voluntary. The pre-period profile $Y_{i0}$ uses September--October 2012 and the post-period profile $Y_{i1}$ the same months in 2013, with $\Delta Y_i=Y_{i1}-Y_{i0}$. August 2012 consumption supplies baseline covariates.

A complete day contains 48 valid half-hour readings. Eligibility requires at least 80\% complete eligible days in each of the five study months and at least two complete dates for each of the seven days of the week in every month. Baseline eligibility applies this rule to August 2012 alone. The analysis includes $n=4{,}389$ households (894 dynamic-tariff and 3,495 comparison households), each contributing one 48-coordinate change profile. Monthly profiles give equal weight to each day of the week, and the two monthly profiles within each comparison period are averaged equally. Six bounded covariates summarize the August profile's mean, first two Fourier harmonics, and weekend--weekday contrast. Detailed preprocessing and covariate definitions are provided in the supplement.

Let $R=1$ denote eligibility. The target is the ATT on standardized recorded profiles among eligible dynamic-tariff households,
\[
 \tau_{\rm E}(h)=\E\{Y_1^1(h)-Y_1^0(h)\mid D=1,R=1\},
 \qquad h\in\mathcal G_{48}:=\{0,\tfrac12,\ldots,23\tfrac12\}.
\]
Identification requires the finite-coordinate analogue of Assumption~\ref{ass:ident} within this selected population, including its observation process; the supplement states these conditions.

\subsection{Estimation and comparisons}
We compare CF--DR, OR, IPW, and Naive DiD. The first three share the six baseline covariates and five household-level folds assigned independently of the data, keeping each household's entire profile in one fold. Propensities are fitted by Firth logistic regression and clipped to $[0.001,0.999]$; conditional mean changes are fitted by group-specific least-squares regressions. CF--DR uses \eqref{eq:tauhat}; OR and IPW follow Section~\ref{ssec:sim_estimators}. Naive DiD is the unadjusted difference in mean changes.

For CF--DR, OR, and IPW, inference accounts for nuisance estimation through the foldwise correction in Supplementary Propositions~S1--S2. We construct 95\% joint intervals over the 48 recorded coordinates using Supplementary Corollary~S3 and 5000 Gaussian multiplier draws.

Energy contrasts sum the half-hour effects over the full day, morning $[7,10)$, or evening $[17,23)$, with marginal 95\% normal intervals. As each reading measures half-hour energy in kWh, summing gives the corresponding period's energy contrast. Sampling inference assumes independent households conditional on the realized calendar-level conditions and tariff schedule. Computation details and sensitivity analyses appear
in the supplement.

\subsection{Results}
No propensity prediction is clipped. ATT weighting retains a control effective sample size of 3,391 and reduces the largest absolute baseline standardized mean difference from 0.087 to 0.011 (Table~\ref{tab:lcl_sample} and Supplementary Figure~S3).

\begin{table}[!htbp]
\centering\small
\caption{Household sample, overlap, and baseline balance.}
\label{tab:lcl_sample}
\begin{tabular}{lrr}
\toprule
Measure & Dynamic tariff & Non-dynamic tariff \\
\midrule
Archive households & 1123 & 4443 \\
Baseline-eligible households & 961 & 3786 \\
Analysis households & 894 & 3495 \\
Baseline daily energy (kWh) & 7.683 & 8.145 \\
Minimum fitted propensity & 0.057 & 0.002 \\
Maximum fitted propensity & 0.293 & 0.429 \\
Clipped predictions (\%) & 0.00 & 0.00 \\
ATT-weighted effective sample size & 894.0 & 3391.1 \\
Maximum control odds & --- & 0.750 \\
Max. $|\mathrm{SMD}|$, before weighting & --- & 0.087 \\
Max. $|\mathrm{SMD}|$, after weighting & --- & 0.011 \\
\bottomrule
\end{tabular}
\par\smallskip
\begin{minipage}{\textwidth}
\footnotesize Baseline energy uses the standardized August 2012 profile. SMD denotes the standardized mean difference between groups.
\end{minipage}
\end{table}

CF--DR estimates a daily effect of $-0.249$ kWh, approximately 2.8\% of the treated group's pretreatment daily mean of 8.816 kWh. The evening effect is $-0.142$ kWh, accounting for approximately 57\% of the daily reduction (Table~\ref{tab:lcl_effects}). Figure~\ref{fig:lcl_effects} shows that joint intervals exclude zero at every recorded half-hour from 18:00 to 22:00. Overnight estimates are positive, but their joint intervals include zero.

CF--DR's daily standard error is 6.2\% smaller than Naive DiD's, while OR is slightly more precise.

\begin{table}[!htbp]
\centering\small
\setlength{\tabcolsep}{3.5pt}
\caption{Estimated energy contrasts in kWh per standardized day.}
\label{tab:lcl_effects}
\begin{tabular}{lrrrrrr}
\toprule
Method & Daily & 95\% CI & Morning & 95\% CI & Evening & 95\% CI \\
\midrule
CF--DR & -0.249 & [-0.413, -0.085] & -0.034 & [-0.060, -0.007] & -0.142 & [-0.208, -0.076] \\
OR & -0.266 & [-0.426, -0.107] & -0.038 & [-0.064, -0.012] & -0.148 & [-0.212, -0.084] \\
IPW & -0.262 & [-0.426, -0.097] & -0.035 & [-0.062, -0.008] & -0.145 & [-0.211, -0.079] \\
Naive DiD & -0.252 & [-0.426, -0.077] & -0.035 & [-0.063, -0.007] & -0.143 & [-0.212, -0.073] \\
\bottomrule
\end{tabular}
\par\smallskip
\begin{minipage}{\textwidth}
\footnotesize Intervals are marginal 95\% confidence intervals.
\end{minipage}
\end{table}

The October-minus-September 2012 pretreatment test gives $p=0.123$ for CF--DR and $p=0.0014$ for Naive DiD (Supplementary Table~S12). This test concerns one pretreatment comparison; conditional parallel trends during the treatment period remains an identifying assumption. An exploratory restriction to baseline night-energy share at most 50\% gives Naive DiD nominal $p=0.065$ in this restricted population.

Prespecified checks using individual comparison months, expanded nuisance models, or stricter completeness produce negative daily estimates ranging from $-0.339$ to $-0.160$ kWh, although the September interval includes zero (Supplementary Table~S13). The stricter completeness rule changes the eligible population. Exploratory night-load, midnight, and influence checks were added after inspecting the initial results; their intervals are nominal and do not account for this selection. Details and additional results appear in the supplement.

Under the stated identification conditions, the estimates indicate a modest reduction in electricity consumption concentrated in evening demand among eligible participants.

\begin{figure}[!htbp]
\centering
\includegraphics[width=\textwidth]{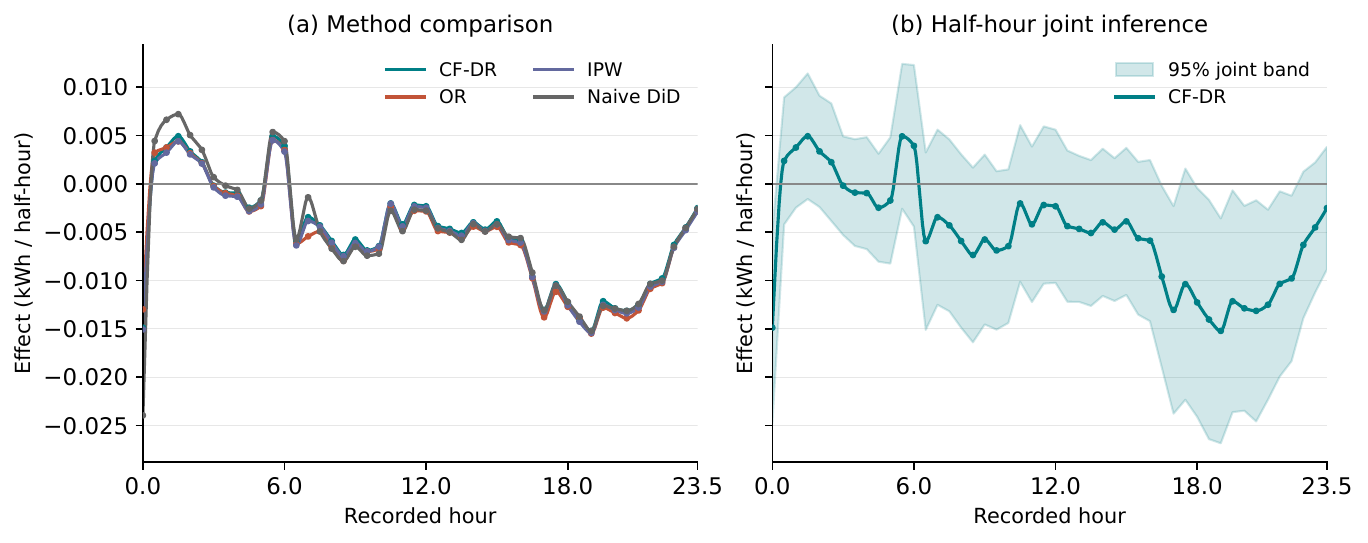}
\caption{Half-hour energy contrasts: (a) method comparison; (b) CF--DR estimate and 95\% joint intervals over the 48 recorded half-hours. Points represent estimates; curves interpolate these estimates.}
\label{fig:lcl_effects}
\end{figure}

\FloatBarrier

\section{Discussion}\label{sec:discussion}
We combine functional parallel-trends identification with cross-fitted, doubly robust estimation. Under ultra-dense observation, negligible trajectory errors and stable nuisance learners preserve the full-data first-order covariance. Sobolev regularity yields continuous-domain Gaussian limits and simultaneous bands. In simulations, estimation errors decrease under single-model misspecification. The variance correction reduces undercoverage under limited overlap, while fixed basis truncation causes coverage to decline with sample size. The household electricity analysis estimates a reduction concentrated in evening demand, with joint inference over the recorded half-hours.

The theory concerns independent units and full or ultra-dense observations. Fixed-sparse records call for observed-data methods that propagate measurement uncertainty without requiring first-order negligible trajectory errors. Dependent observations require covariance estimates and resampling procedures suited to their spatial, serial, or cluster structure. Extensions to multiple periods, staggered interventions, and adaptive basis selection would broaden the range of applications.

\section*{Supplementary materials}
The supplementary material contains additional proofs, computational details, and further results for the simulation study and electricity application.

\section*{Acknowledgements}
Ran's research is supported by Research Development Fund (RDF-24-02-036) of Xi'an Jiaotong-Liverpool University.

\section*{Disclosure statement}
No potential conflict of interest was reported by the authors.

\bibliographystyle{erae}
\bibliography{references}

@article{abadie2005semiparametric,
  author  = {Abadie, Alberto},
  title   = {Semiparametric Difference-in-Differences Estimators},
  journal = {The Review of Economic Studies},
  year    = {2005},
  volume  = {72},
  number  = {1},
  pages   = {1--19}
}

@article{bang2005doubly,
  author  = {Bang, Heejung and Robins, James M.},
  title   = {Doubly Robust Estimation in Missing Data and Causal Inference Models},
  journal = {Biometrics},
  year    = {2005},
  volume  = {61},
  number  = {4},
  pages   = {962--973}
}

@article{belloni2014post,
  author  = {Belloni, Alexandre and Chernozhukov, Victor and Hansen, Christian},
  title   = {Inference on Treatment Effects after Selection among High-Dimensional Controls},
  journal = {The Review of Economic Studies},
  year    = {2014},
  volume  = {81},
  number  = {2},
  pages   = {608--650}
}

@article{belloni2017program,
  author  = {Belloni, Alexandre and Chernozhukov, Victor and Fern{\'a}ndez-Val, Iv{\'a}n and Hansen, Christian},
  title   = {Program Evaluation and Causal Inference with High-Dimensional Data},
  journal = {Econometrica},
  year    = {2017},
  volume  = {85},
  number  = {1},
  pages   = {233--298}
}

@book{bickel1993efficient,
  author    = {Bickel, Peter J. and Klaassen, Chris A. J. and Ritov, Ya'acov and Wellner, Jon A.},
  title     = {Efficient and Adaptive Estimation for Semiparametric Models},
  publisher = {Johns Hopkins University Press},
  year      = {1993},
  address   = {Baltimore}
}

@article{borusyak2024revisiting,
  author  = {Borusyak, Kirill and Jaravel, Xavier and Spiess, Jann},
  title   = {Revisiting Event Study Designs: Robust and Efficient Estimation},
  journal = {The Review of Economic Studies},
  year    = {2024},
  volume  = {91},
  number  = {6},
  pages   = {3253--3285}
}

@article{callaway2021did,
  author  = {Callaway, Brantly and Sant'Anna, Pedro H. C.},
  title   = {Difference-in-Differences with Multiple Time Periods},
  journal = {Journal of Econometrics},
  year    = {2021},
  volume  = {225},
  number  = {2},
  pages   = {200--230}
}

@article{chernozhukov2014gaussian,
  author  = {Chernozhukov, Victor and Chetverikov, Denis and Kato, Kengo},
  title   = {Gaussian Approximation of Suprema of Empirical Processes},
  journal = {The Annals of Statistics},
  year    = {2014},
  volume  = {42},
  number  = {4},
  pages   = {1564--1597}
}

@article{chernozhukov2018dml,
  author  = {Chernozhukov, Victor and Chetverikov, Denis and Demirer, Mert and Duflo, Esther and Hansen, Christian and Newey, Whitney and Robins, James},
  title   = {Double/Debiased Machine Learning for Treatment and Structural Parameters},
  journal = {The Econometrics Journal},
  year    = {2018},
  volume  = {21},
  number  = {1},
  pages   = {C1--C68}
}

@article{choi2018geometric,
  author  = {Choi, Hyunphil and Reimherr, Matthew},
  title   = {A Geometric Approach to Confidence Regions and Bands for Functional Parameters},
  journal = {Journal of the Royal Statistical Society: Series B (Statistical Methodology)},
  year    = {2018},
  volume  = {80},
  number  = {1},
  pages   = {239--260}
}

@article{degras2011scb,
  author  = {Degras, David A.},
  title   = {Simultaneous Confidence Bands for Nonparametric Regression with Functional Data},
  journal = {Statistica Sinica},
  year    = {2011},
  volume  = {21},
  number  = {4},
  pages   = {1735--1765}
}

@article{ecker2024functional,
  author  = {Ecker, Kreske and {de Luna}, Xavier and Schelin, Lina},
  title   = {Causal Inference with a Functional Outcome},
  journal = {Journal of the Royal Statistical Society: Series C (Applied Statistics)},
  year    = {2024},
  volume  = {73},
  number  = {1},
  pages   = {221--240}
}

@misc{fangliebl2025honest,
  author        = {Fang, Chencheng and Liebl, Dominik},
  title         = {Making Event Study Plots Honest: A Functional Data Approach to Causal Inference},
  year          = {2025},
  howpublished  = {arXiv preprint arXiv:2512.06804},
  url           = {https://arxiv.org/abs/2512.06804},
  eprint        = {2512.06804},
  archivePrefix = {arXiv},
  primaryClass  = {econ.EM}
}

@article{hahn1998role,
  author  = {Hahn, Jinyong},
  title   = {On the Role of the Propensity Score in Efficient Semiparametric Estimation of Average Treatment Effects},
  journal = {Econometrica},
  year    = {1998},
  volume  = {66},
  number  = {2},
  pages   = {315--331}
}

@article{hall2006properties,
  author  = {Hall, Peter and M{\"u}ller, Hans-Georg and Wang, Jane-Ling},
  title   = {Properties of Principal Component Methods for Functional and Longitudinal Data Analysis},
  journal = {The Annals of Statistics},
  year    = {2006},
  volume  = {34},
  number  = {3},
  pages   = {1493--1517}
}

@book{hernan2020causal,
  author    = {Hern{\'a}n, Miguel A. and Robins, James M.},
  title     = {Causal Inference: What If},
  publisher = {Chapman \& Hall/CRC},
  year      = {2020},
  address   = {Boca Raton}
}

@article{hirano2003efficient,
  author  = {Hirano, Keisuke and Imbens, Guido W. and Ridder, Geert},
  title   = {Efficient Estimation of Average Treatment Effects Using the Estimated Propensity Score},
  journal = {Econometrica},
  year    = {2003},
  volume  = {71},
  number  = {4},
  pages   = {1161--1189}
}

@article{holland1986statistics,
  author  = {Holland, Paul W.},
  title   = {Statistics and Causal Inference},
  journal = {Journal of the American Statistical Association},
  year    = {1986},
  volume  = {81},
  number  = {396},
  pages   = {945--960}
}

@book{horvath2012inference,
  author    = {Horv{\'a}th, Lajos and Kokoszka, Piotr},
  title     = {Inference for Functional Data with Applications},
  publisher = {Springer},
  year      = {2012},
  address   = {New York}
}

@book{hsing2015theoretical,
  author    = {Hsing, Tailen and Eubank, Randall},
  title     = {Theoretical Foundations of Functional Data Analysis, with an Introduction to Linear Operators},
  publisher = {Wiley},
  year      = {2015}
}

@book{imbens2015causal,
  author    = {Imbens, Guido W. and Rubin, Donald B.},
  title     = {Causal Inference for Statistics, Social, and Biomedical Sciences: An Introduction},
  publisher = {Cambridge University Press},
  year      = {2015},
  address   = {Cambridge}
}

@article{liebl2023ffscb,
  author  = {Liebl, Dominik and Reimherr, Matthew},
  title   = {Fast and Fair Simultaneous Confidence Bands for Functional Parameters},
  journal = {Journal of the Royal Statistical Society: Series B (Statistical Methodology)},
  year    = {2023},
  volume  = {85},
  number  = {3},
  pages   = {842--868}
}

@misc{liu2026functional,
  author        = {Liu, Xijia and Ecker, Kreske Felix and Schelin, Lina and {de Luna}, Xavier},
  title         = {Doubly Robust Estimation with Functional Outcomes Missing at Random},
  year          = {2026},
  howpublished  = {arXiv preprint arXiv:2411.17224v3},
  doi           = {10.48550/arXiv.2411.17224},
  url           = {https://arxiv.org/abs/2411.17224v3}
}

@article{newey1990semiparametric,
  author  = {Newey, Whitney K.},
  title   = {Semiparametric Efficiency Bounds},
  journal = {Journal of Applied Econometrics},
  year    = {1990},
  volume  = {5},
  number  = {2},
  pages   = {99--135}
}

@article{pini2016itp,
  author  = {Pini, Alessia and Vantini, Simone},
  title   = {The Interval Testing Procedure: A General Framework for Inference in Functional Data Analysis},
  journal = {Biometrics},
  year    = {2016},
  volume  = {72},
  number  = {3},
  pages   = {835--845}
}

@article{pini2017iwt,
  author  = {Pini, Alessia and Vantini, Simone},
  title   = {Interval-Wise Testing for Functional Data},
  journal = {Journal of Nonparametric Statistics},
  year    = {2017},
  volume  = {29},
  number  = {2},
  pages   = {407--424}
}

@book{ramsay2005fda,
  author    = {Ramsay, James O. and Silverman, Bernard W.},
  title     = {Functional Data Analysis},
  edition   = {2},
  publisher = {Springer},
  year      = {2005},
  address   = {New York}
}

@article{robins1994estimation,
  author  = {Robins, James M. and Rotnitzky, Andrea and Zhao, Lue Ping},
  title   = {Estimation of Regression Coefficients When Some Regressors Are Not Always Observed},
  journal = {Journal of the American Statistical Association},
  year    = {1994},
  volume  = {89},
  number  = {427},
  pages   = {846--866}
}

@article{rosenbaum1983,
  author  = {Rosenbaum, Paul R. and Rubin, Donald B.},
  title   = {The Central Role of the Propensity Score in Observational Studies for Causal Effects},
  journal = {Biometrika},
  year    = {1983},
  volume  = {70},
  number  = {1},
  pages   = {41--55}
}

@article{rubin1974,
  author  = {Rubin, Donald B.},
  title   = {Estimating Causal Effects of Treatments in Randomized and Nonrandomized Studies},
  journal = {Journal of Educational Psychology},
  year    = {1974},
  volume  = {66},
  number  = {5},
  pages   = {688--701}
}

@article{santanna2020drdid,
  author  = {Sant'Anna, Pedro H. C. and Zhao, Jun},
  title   = {Doubly Robust Difference-in-Differences Estimators},
  journal = {Journal of Econometrics},
  year    = {2020},
  volume  = {219},
  number  = {1},
  pages   = {101--122}
}

@article{sunabraham2021eventstudy,
  author  = {Sun, Liyang and Abraham, Sarah},
  title   = {Estimating Dynamic Treatment Effects in Event Studies with Heterogeneous Treatment Effects},
  journal = {Journal of Econometrics},
  year    = {2021},
  volume  = {225},
  number  = {2},
  pages   = {175--199}
}

@article{testa2025drfate,
  author        = {Testa, Lorenzo and Boschi, Tobia and Chiaromonte, Francesca and Kennedy, Edward H. and Reimherr, Matthew},
  title         = {Doubly-Robust Functional Average Treatment Effect Estimation},
  year          = {2026},
  journal       = {Journal of Causal Inference},
  volume        = {14},
  number        = {1},
  pages         = {20250045},
  doi           = {10.1515/jci-2025-0045}
}

@book{tsiatis2006semiparametric,
  author    = {Tsiatis, Anastasios A.},
  title     = {Semiparametric Theory and Missing Data},
  publisher = {Springer},
  year      = {2006},
  address   = {New York}
}

@book{vandervaart1998asymptotic,
  author    = {{van der Vaart}, Aad W.},
  title     = {Asymptotic Statistics},
  publisher = {Cambridge University Press},
  year      = {1998},
  address   = {Cambridge}
}

@book{vandervaartwellner1996weak,
  author    = {{van der Vaart}, Aad W. and Wellner, Jon A.},
  title     = {Weak Convergence and Empirical Processes: With Applications to Statistics},
  publisher = {Springer},
  year      = {1996},
  address   = {New York}
}

@article{wang2016functional,
  author  = {Wang, Jane-Ling and Chiou, Jeng-Min and M{\"u}ller, Hans-Georg},
  title   = {Functional Data Analysis},
  journal = {Annual Review of Statistics and Its Application},
  year    = {2016},
  volume  = {3},
  pages   = {257--295}
}

@article{yao2005fpca,
  author  = {Yao, Fang and M{\"u}ller, Hans-Georg and Wang, Jane-Ling},
  title   = {Functional Data Analysis for Sparse Longitudinal Data},
  journal = {Journal of the American Statistical Association},
  year    = {2005},
  volume  = {100},
  number  = {470},
  pages   = {577--590}
}

@article{roth2023whats,
  author  = {Roth, Jonathan and Sant'Anna, Pedro H. C. and Bilinski, Alyssa and Poe, John},
  title   = {What's Trending in Difference-in-Differences? A Synthesis of the Recent Econometrics Literature},
  journal = {Journal of Econometrics},
  year    = {2023},
  volume  = {235},
  number  = {2},
  pages   = {2218--2244}
}

@book{adams2003sobolev,
  author    = {Adams, Robert A. and Fournier, John J. F.},
  title     = {Sobolev Spaces},
  edition   = {2},
  series    = {Pure and Applied Mathematics},
  volume    = {140},
  publisher = {Academic Press},
  address   = {Amsterdam},
  year      = {2003}
}

@article{firth1993bias,
  author = {Firth, David},
  title = {Bias reduction of maximum likelihood estimates},
  journal = {Biometrika},
  year = {1993},
  volume = {80},
  number = {1},
  pages = {27--38},
  doi = {10.1093/biomet/80.1.27}
}

@misc{ukpn2014lcl,
  author = {{UK Power Networks}},
  title = {{SmartMeter} Energy Consumption Data in {London} Households},
  year = {n.d.},
  howpublished = {London Datastore},
  note = {Public consumption records, November 2011--February 2014. Accessed 28 September 2026}
}

@techreport{schofield2014residential,
  author = {Schofield, James R. and Carmichael, Richard and Tindemans, Simon H. and Woolf, Matt and Bilton, Mark and Strbac, Goran},
  title = {Residential Consumer Responsiveness to Time-Varying Pricing},
  institution = {Imperial College London},
  type = {Low Carbon London Learning Report},
  number = {A3},
  year = {2014},
  doi = {10.25561/21252}
}

@techreport{schofield2015guide,
  author = {Schofield, James R. and Carmichael, Richard and Tindemans, Simon H. and Bilton, Mark and Woolf, Matt and Strbac, Goran},
  title = {{Low Carbon London} Project: Data from the Dynamic Time-of-Use Electricity Pricing Trial, 2013},
  institution = {UK Data Service},
  type = {User Guide, Study},
  number = {7857},
  year = {2015},
  note = {\url{https://doc.ukdataservice.ac.uk/doc/7857/mrdoc/pdf/7857_userguide.pdf}}
}

\clearpage
\appendix
\section*{Appendix: Proofs}\label{sec:appendix}
\begingroup
\setlength{\abovedisplayskip}{9pt plus 2pt minus 2pt}
\setlength{\belowdisplayskip}{9pt plus 2pt minus 2pt}
Write $P_nf=n^{-1}\sum_i f(W_i)$ and let $\mathcal F_{-k}$ contain the training records and fitted objects for fold $k$. Constants $C$ may change between displays and do not depend on $n$. All function-valued expectations are Bochner expectations.

\subsection*{Canonical gradient and efficiency bound}
\begin{proof}[Proof of Proposition~\ref{prop:eif}]
Let $R_0=\Delta Y-\mu_0(X)$ and consider $dP_\epsilon=(1+\epsilon S)dP$, where $PS=0$ and $\|S\|_\infty<\infty$. For $|\epsilon|\|S\|_\infty<1/2$, put $s_0(x)=\E(S\mid X=x,D=0)$. Conditional expectation under $P_\epsilon$ gives the exact identity
\begin{align*}
 \mu_{0,\epsilon}(x)
 &=\frac{\mu_0(x)+\epsilon\E(\Delta Y S\mid X=x,D=0)}{1+\epsilon s_0(x)},\\
 \frac{\mu_{0,\epsilon}(x)-\mu_0(x)}{\epsilon}
 &=\frac{\E(R_0S\mid X=x,D=0)}{1+\epsilon s_0(x)}
 \longrightarrow \dot\mu_0(x):=\E(R_0S\mid X=x,D=0).
\end{align*}
Conditional Jensen and overlap imply
\[
 \E\|\mu_0(X)\|_{\cH}^2\leq c_0^{-1}\E\|\Delta Y\|_{\cH}^2<\infty,
 \qquad
 \E\|\dot\mu_0(X)\|_{\cH}^2
 \leq c_0^{-1}\|S\|_\infty^2\E\{(1-D)\|R_0\|_{\cH}^2\}<\infty.
\]
The difference between the preceding difference quotient and $\dot\mu_0$ has $L^2(P_X;\cH)$ norm at most $2|\epsilon|\|S\|_\infty\|\dot\mu_0\|_{2,\cH}$. Thus differentiation can be passed through the expectations below. For $N_\epsilon=\E_\epsilon[D\{\Delta Y-\mu_{0,\epsilon}(X)\}]$ and $p_\epsilon=\E_\epsilon D$,
\begin{gather*}
 N_0=p_0\tau_0,\qquad \dot p_0=\E(DS),\qquad
 \dot N_0=\E(DR_0S)-\E\{D\dot\mu_0(X)\},\\
 \dot\tau_0
 =\frac{\dot N_0-\tau_0\dot p_0}{p_0}
 =\frac{\E\{D(R_0-\tau_0)S\}-\E\{D\dot\mu_0(X)\}}{p_0}.
\end{gather*}
Iterated expectation yields
\[
 \E\{D\dot\mu_0(X)\}
 =\E\{\pi_0(X)\E(R_0S\mid X,D=0)\}
 =\E\left\{\frac{(1-D)\pi_0(X)}{1-\pi_0(X)}R_0S\right\}.
\]
Consequently $\dot\tau_0=\E(\phi S)$, where $\phi$ is \eqref{eq:eif_point}. Moreover,
\[
 P\phi=0,\qquad
 \E\|\phi\|_{\cH}^2
 \leq C\{\E\|\Delta Y\|_{\cH}^2+\E\|\mu_0(X)\|_{\cH}^2+\|\tau_0\|_{\cH}^2\}<\infty.
\]
Each bounded centered $S$ generates a quadratic-mean differentiable tilt: the Taylor expansion of $\sqrt{1+\epsilon S}$ gives
\[
 \left\|\frac{\sqrt{1+\epsilon S}-1}{\epsilon}-\frac S2\right\|_2\longrightarrow0.
\]
For any $S\in L_0^2(P)$, the bounded centered truncations
$S_M=S\1(|S|\leq M)-P\{S\1(|S|\leq M)\}$ converge to $S$ in $L^2(P)$. Thus these tilts generate the full nonparametric tangent space $L_0^2(P)$. Since
\[
 \|\E\{\phi(S_1-S_2)\}\|_{\cH}
 \leq(\E\|\phi\|_{\cH}^2)^{1/2}\|S_1-S_2\|_2,
\]
the derivative has a unique continuous extension to this closure. For every $f\in\cH$, its scalar representative $\langle f,\phi\rangle$ belongs to $L_0^2(P)$ and is therefore the canonical gradient of $\langle f,\tau_0\rangle$ \citep{bickel1993efficient,tsiatis2006semiparametric}.

For a fixed $f\in\cH$, any regular asymptotically linear influence function $\chi_f$ for $\langle f,\tau_0\rangle$ satisfies
\[
 \E\{(\chi_f-\langle f,\phi\rangle)S\}=0\quad(S\in L_0^2(P)),\qquad
 \E\chi_f^2=\E\langle f,\phi\rangle^2+
            \E(\chi_f-\langle f,\phi\rangle)^2
 \geq V_{\mathrm{eff}}(f).
\]
The covariance operator is trace class since $\operatorname{tr}(\Sigma_{\mathrm{eff}})=\E\|\phi\|_{\cH}^2<\infty$.
Applying the directional inequality with $\chi_f=\langle f,\chi\rangle$ gives $\Sigma_\chi\succeq\Sigma_{\mathrm{eff}}$.

Under the Sobolev regularity and moment conditions in Assumption~\ref{ass:asym_paths}, let $R_a=\Delta Y-\mu_a(X)$, $d(X)=\mu_1(X)-\mu_0(X)$, and $V_a(s,t\mid X)=\Cov\{\Delta Y(s),\Delta Y(t)\mid X,D=a\}$. Then
\[
 p_0\phi=D(R_1+d-\tau_0)-(1-D)\frac{\pi_0}{1-\pi_0}R_0,
 \qquad \E(R_a\mid X,D=a)=0.
\]
Using $D(1-D)=0$ and taking conditional second moments gives
\begin{align*}
 p_0^2\E\{\phi(s)\phi(t)\mid X\}
 &=\pi_0 V_1(s,t\mid X)+\frac{\pi_0^2}{1-\pi_0}V_0(s,t\mid X)\\
 &\quad+\pi_0\{d(s)-\tau_0(s)\}\{d(t)-\tau_0(t)\}.
\end{align*}
Integration over $X$ gives the kernel decomposition
\begin{equation}\label{eq:effkernel}
 \begin{split}
 C(s,t)=\frac{1}{p_0^2}\E\bigg[&\pi_0 V_1(s,t\mid X)
 +\frac{\pi_0^2}{1-\pi_0}V_0(s,t\mid X)\\
 &+\pi_0\{d(X)(s)-\tau_0(s)\}\{d(X)(t)-\tau_0(t)\}\bigg].
 \end{split}
\end{equation}
\end{proof}

\subsection*{Drift identity and orthogonality}
\begin{lemma}[Drift identity and orthogonality]\label{lem:orth}
Under Assumption~\ref{ass:ident}, for $0\leq\pi\leq1-c$ with $0<c<c_0$ as in Assumption~\ref{ass:asym_cf} and integrable $\mu$, \eqref{eq:drift} holds. Its derivative at $(\pi_0,\mu_0)$ is zero along bounded propensity directions and square-integrable outcome directions.
\end{lemma}
\begin{proof}
Write $\delta_\pi=\pi-\pi_0$, $\delta_\mu=\mu-\mu_0$, $R_0=\Delta Y-\mu_0(X)$ and $w_\pi=D-(1-D)\pi/(1-\pi)$. Since
\[
 w_\pi-w_{\pi_0}=-\frac{(1-D)\delta_\pi}{(1-\pi)(1-\pi_0)},
\]
the score difference and its conditional expectation are
\begin{align*}
 q(W;\eta)-q(W;\eta_0)
 &=-\frac{(1-D)\delta_\pi}{(1-\pi)(1-\pi_0)}R_0-w_\pi\delta_\mu,\\
 \E\{(1-D)R_0\mid X\}&=0,\qquad
 \E(w_\pi\mid X)=\pi_0-\frac{(1-\pi_0)\pi}{1-\pi}
                =-\frac{\delta_\pi}{1-\pi},\\
 \E\{q(W;\eta)-q(W;\eta_0)\mid X\}
 &=\frac{\delta_\pi\delta_\mu}{1-\pi}.
\end{align*}
By \eqref{eq:ident_reg}, $Pq(W;\eta_0)=p_0\tau_0$, proving \eqref{eq:drift}. In any Hilbert space $\cE$ where the norms are finite,
\[
 \|Pq(W;\eta)-p_0\tau_0\|_{\cE}
 \leq c^{-1}\E\{|\delta_\pi|\|\delta_\mu\|_{\cE}\}
 \leq c^{-1}\|\delta_\pi\|_2\|\delta_\mu\|_{2,\cE}.
\]
For $\eta_\epsilon=(\pi_0+\epsilon h_\pi,\mu_0+\epsilon h_\mu)$ and sufficiently small $\epsilon$, the drift is
\[
 \epsilon^2\E\left\{\frac{h_\pi h_\mu}{1-\pi_0-\epsilon h_\pi}\right\},\qquad
 \left\|\frac{Pq(W;\eta_\epsilon)-Pq(W;\eta_0)}{\epsilon}\right\|_{\cE}
 \leq C|\epsilon|\|h_\pi\|_2\|h_\mu\|_{2,\cE}\longrightarrow0.\qedhere
\]
\end{proof}

\subsection*{Curve accuracy and response stability}
\begin{proposition}[Rates for curve estimation]\label{prop:curve_rates}
Suppose $m_n$ observations are available per curve. Conditional on the underlying curves and observation times, measurement errors are independent, mean zero, and sub-Gaussian with variance proxy at most $\sigma^2$, uniformly in units, periods, and $n$.
\begin{enumerate}[label=(\roman*),leftmargin=*]
\item If the curves belong to a fixed span of $q$ deterministic continuously differentiable basis functions and the discrete least-squares Gram matrices have eigenvalues uniformly bounded away from zero and infinity, coefficient least squares satisfies
\begin{equation}\label{eq:finite_curve_rate}
 r_{n,H^1}=O_p\{\sqrt{\log(n+1)/m_n}\}.
\end{equation}
Thus $m_n/\{n\log(n+1)\}\to\infty$ gives $r_{n,H^1}=o_p(n^{-1/2})$.
\item Alternatively, suppose the curves are twice continuously differentiable and $\sup_{i,T}\|Y_{iT}''\|_\infty\leq A<\infty$ almost surely, for a deterministic $A$. Observe $t_j=(j-1/2)/m_n$, $j=1,\ldots,m_n$, with $J_n\geq2$ dividing $m_n$. Average observations in $J_n$ equal bins, interpolate linearly between bin centers, and extrapolate linearly to the two endpoints. Then
\begin{equation*}
 r_{n,H^1}=O_p\left\{J_n^{-1}+\sqrt{\frac{J_n^3\log(nJ_n+1)}{m_n}}\right\}.
\end{equation*}
For $r_{n,H^1}=o_p(n^{-1/2})$, it suffices that
\[
 \frac{\sqrt n}{J_n}\longrightarrow0,\qquad
 \frac{nJ_n^3\log(nJ_n+1)}{m_n}\longrightarrow0.
\]
For example, take $J_n=\lceil n^{3/4}\rceil$ and $m_n=J_n\lceil n^4/J_n\rceil\asymp n^4$.
\end{enumerate}
Here $H^1=W^{1,2}[0,1]$. The rates also imply $L^2$ and uniform error bounds by continuous embedding.
\end{proposition}
\begin{proof}[Proof of Proposition~\ref{prop:curve_rates}]
In the fixed-basis case, let $B_{j\ell}=b_\ell(t_j)$ and $G_{m_n}=B^\top B/m_n$. Exact span membership implies
\[
 \widehat{\bm\theta}-\bm\theta=G_{m_n}^{-1}B^\top\bm\varepsilon/m_n.
\]
Conditional on the curves and design, coordinate $\ell$ of this error is sub-Gaussian with variance proxy
\[
 \frac{\sigma^2}{m_n^2}\bm e_\ell^\top G_{m_n}^{-1}B^\top B G_{m_n}^{-1}\bm e_\ell
 =\frac{\sigma^2}{m_n}(G_{m_n}^{-1})_{\ell\ell}\leq C\sigma^2/m_n.
\]
For the fixed basis, $\|\bm b^\top\bm v\|_{H^1}\leq C\|\bm v\|_2$. A union bound over coordinates, periods, units and the fixed number of folds gives
\[
 P(r_{n,H^1}>u)\leq Cnq\exp(-c m_n u^2/\sigma^2).
\]
Taking $u=C_*\sqrt{\log(n+1)/m_n}$ with a sufficiently large constant $C_*$ proves \eqref{eq:finite_curve_rate}.

For the nonparametric construction, put $h_n=1/J_n$ and $c_j=(j-1/2)h_n$. Divisibility and the regular midpoint grid place exactly $m_n/J_n$ measurements in each bin. Define $\bar Z_j=(J_n/m_n)\sum_{t_\ell\text{ in bin }j}Z_\ell$. Symmetry about $c_j$ cancels the linear Taylor term, giving
\begin{align*}
 \bar Z_j&=Y(c_j)+a_j+e_j,\\
 |a_j|&\leq\frac{AJ_n}{2m_n}\sum_{t_\ell\text{ in bin }j}(t_\ell-c_j)^2\leq C A h_n^2,\\
 P\left(\max_{i,T,j}|e_{iTj}|>u\right)
 &\leq CnJ_n\exp\{-c m_n u^2/(J_n\sigma^2)\}.
\end{align*}
The union bound includes both periods and all $K$ fold-specific fits, changing the constant by a fixed factor. Thus $\max|e_{iTj}|=O_p\{\sqrt{J_n\log(nJ_n+1)/m_n}\}$. Let $I_{J_n}\bm v$ be the linear interpolant with endpoint extrapolation. On $[c_j,c_{j+1}]$, its derivative is $(v_{j+1}-v_j)/h_n$, so
\[
 \|I_{J_n}\bm v\|_2\leq C\max_j|v_j|,\qquad
 \|(I_{J_n}\bm v)'\|_2^2\leq C J_n^2\max_j|v_j|^2.
\]
The extrapolation coefficients at either endpoint are $3/2$ and $-1/2$, so the same bounds hold there with constants independent of $n$. For $\bm y=(Y(c_1),\ldots,Y(c_{J_n}))$, the derivative of $I_{J_n}\bm y$ is an interval average of $Y'$. The bound on $Y''$ therefore yields, also on the extrapolated boundary intervals,
\[
 \|(I_{J_n}\bm y)'-Y'\|_\infty\leq C A h_n,
 \qquad \|I_{J_n}\bm y-Y\|_\infty\leq C A h_n^2.
\]
Combining these bounds,
\begin{align*}
 \|\widehat Y-Y\|_{H^1}
 &\leq\|I_{J_n}\bm y-Y\|_{H^1}+\|I_{J_n}\bm a\|_{H^1}+\|I_{J_n}\bm e\|_{H^1}\\
 &\leq C A/J_n+CJ_n\max|a_j|+CJ_n\max|e_j|
 =O_p\left\{J_n^{-1}+\sqrt{J_n^3\log(nJ_n+1)/m_n}\right\}.
\end{align*}
These inequalities are uniform over the curves. In particular,
\[
 \sqrt n\,r_{n,H^1}
 =O_p\left\{\frac{\sqrt n}{J_n}
       +\sqrt{\frac{nJ_n^3\log(nJ_n+1)}{m_n}}\right\}=o_p(1)
\]
under the stated growth conditions. In the displayed example the two deterministic rates are $O(n^{-1/4})$ and $O\{n^{-3/8}\sqrt{\log(n+1)}\}$, respectively. Finally, $\|f\|_2\leq\|f\|_{H^1}$ and \eqref{eq:sobolev_bounds} give the other norm bounds.
\end{proof}

\begin{proposition}[A nonparametric learner satisfying the rate conditions]\label{prop:histogram}
Suppose $X\in[0,1]$ has a density bounded above and below, $\pi_0$ is Lipschitz, and $x\mapsto\mu_0(x)$ is Lipschitz in $\cE$. Under overlap and the control residual-moment condition in Assumption~\ref{ass:asym_basic}, partition $[0,1]$ into $J_X\asymp n^{1/3}$ equal bins. In a training bin with counts $N_j,N_{1j}$, estimate propensity by $(N_{1j}+1/2)/(N_j+1)$, clipped to $[c,1-c]$ with $0<c<c_0$ as in Assumption~\ref{ass:asym_cf}. Estimate the control regression by the bin's control mean, using the training-fold control mean when that bin has no controls. Define the fit as zero on the event that the entire training fold has no controls. Then $a_n,b_{n,\cE}=O_p(n^{-1/3})$, and the response-stability constant in \eqref{eq:learner_stability} is at most two.
\end{proposition}
\begin{proof}[Proof of Proposition~\ref{prop:histogram}]
Let $N\asymp n$ be a training-set size, $h=1/J_X$, and $B_j$ a bin. Write
\[
 p_j=P(X\in B_j),\quad q_j=P(X\in B_j,D=0),\quad
 m_j=\E(\Delta Y\mid X\in B_j,D=0).
\]
Density and overlap give $ch\leq p_j,q_j\leq Ch$, and Lipschitz continuity gives $\sup_{x\in B_j}\|m_j-\mu_0(x)\|_{\cE}\leq Ch$. The control count $N_{0j}\sim\operatorname{Binomial}(N,q_j)$ satisfies
\begin{align*}
 P(N_{0j}=0)&\leq e^{-cNh},\\
 \E\left\{\frac{\1(N_{0j}>0)}{N_{0j}}\right\}
 &\leq2\E\{(N_{0j}+1)^{-1}\}
 =\frac{2\{1-(1-q_j)^{N+1}\}}{(N+1)q_j}\leq\frac{C}{Nh}.
\end{align*}
Given a positive count $N_{0j}$, the control mean has expectation $m_j$ and squared Hilbert risk at most $C/N_{0j}$. Lipschitz continuity on a compact interval and the residual-moment bound also give $\sup_x\E(\|\Delta Y\|_{\cE}^2\mid X=x,D=0)\leq C$. Conditional on the training covariates and labels, Jensen's inequality therefore bounds the second moment of the fallback control mean by $C$, including on an empty-bin event. Consequently,
\begin{align*}
 \E\|\widehat\mu_0^{\mathrm{lat}}-\mu_0\|_{2,\cE}^2
 &\leq C\sum_jp_j\left\{h^2+
   \E\frac{\1(N_{0j}>0)}{N_{0j}}+P(N_{0j}=0)\right\}\\
 &\leq C\{h^2+(Nh)^{-1}+e^{-cNh}\}.
\end{align*}
For propensity, set $\bar\pi_j=P(D=1\mid X\in B_j)$. Given $N_j=m$, the treated count is binomial with parameters $m,\bar\pi_j$, so the unclipped estimate obeys
\[
 \E\left[\left\{\frac{N_{1j}+1/2}{m+1}-\bar\pi_j\right\}^2\Bigm|N_j=m\right]
 =\frac{m\bar\pi_j(1-\bar\pi_j)+(1/2-\bar\pi_j)^2}{(m+1)^2}
 \leq\frac{C}{m+1}.
\]
Moreover, $\sup_{x\in B_j}|\bar\pi_j-\pi_0(x)|\leq Ch$ and
$\E(N_j+1)^{-1}\leq\{(N+1)p_j\}^{-1}$. Clipping is a contraction toward every value in $[c,1-c]$, giving propensity risk at most $C\{h^2+(Nh)^{-1}\}$. With $h\asymp n^{-1/3}$, both squared risks are $O(n^{-2/3})$. Markov's inequality and fixed $K$ yield
\[
 a_n,b_{n,\cE}=O_p(n^{-1/3}),\qquad
 a_nb_{n,\cE}=O_p(n^{-2/3})=o_p(n^{-1/2}).
\]
Both the nonempty-bin mean and its fallback are convex combinations of training responses, with weights determined only by covariates and group labels. Hence
\[
 \sup_x\|\widehat\mu_0(x)-\widehat\mu_0^{\mathrm{lat}}(x)\|_{\cE}
 \leq\max_i\|\widehat{\Delta Y}_i-\Delta Y_i\|_{\cE}
 \leq 2r_{n,\cE}.
\]
If the training fold has no controls, both fits are zero and the bound still holds. This event has probability at most $(1-c_0)^N$ and is recorded as a fitting failure in the simulations.
\end{proof}

\subsection*{Asymptotic linearity}
\begin{proof}[Proof of Theorem~\ref{thm:asym_linear}]
Work in a separable Hilbert space $\cE$ satisfying the assumptions. Define $w_\pi(W)=D-(1-D)\pi(X)/\{1-\pi(X)\}$ and $q_0(W)=q(W;\eta_0)$. For $i\in\cI_k$, distinguish the scores
\begin{align*}
 q_i^{\mathrm{lat}}&=w_{\widehat\pi^{(-k)}}(W_i)
           \{\Delta Y_i-\widehat\mu_0^{(-k)}(X_i)\},\\
 q_i^{\mathrm{rec}}&=w_{\widehat\pi^{(-k)}}(W_i)
           \{\widehat{\Delta Y}_i-\widehat\mu_0^{(-k)}(X_i)\}=\widetilde q_i.
\end{align*}
Both scores use the same nuisance fit and differ in their outcome input. Put $h_k(W)=q(W;\widehat\eta^{(-k)})-q_0(W)$, with true change $\Delta Y$, and $\omega_k=n_k/n$. Since $P(q_0-D\tau_0)=0$, the exact decomposition is
\begin{equation}\label{eq:four_term_expansion}
 \widehat p(\widehat\tau-\tau_0)
 =(P_n-P)(q_0-D\tau_0)+R_{\mathrm{emp}}+R_{\mathrm{drift}}+R_{\mathrm{rec}},
\end{equation}
where
\[
 \begin{split}
 R_{\mathrm{emp}}&=\sum_k\omega_k(P_{n,k}-P)h_k,\qquad
 R_{\mathrm{drift}}=\sum_k\omega_kPh_k,\\
 R_{\mathrm{rec}}&=\frac1n\sum_k\sum_{i\in\cI_k}
             (q_i^{\mathrm{rec}}-q_i^{\mathrm{lat}}).
 \end{split}
\]
For each fold, let $a_k=\|\widehat\pi^{(-k)}-\pi_0\|_2$, $\widetilde b_k=\|\widehat\mu_0^{(-k)}-\mu_0\|_{2,\cE}$, and $\widetilde b_n=\max_k\widetilde b_k$. Response stability gives
\[
 \widetilde b_n\leq b_{n,\cE}+C_nr_{n,\cE}=o_p(1),\qquad
 a_n\widetilde b_n\leq a_nb_{n,\cE}+a_nC_nr_{n,\cE}=o_p(n^{-1/2}).
\]
With $R_0=\Delta Y-\mu_0(X)$, $\delta_{\pi,k}=\widehat\pi^{(-k)}-\pi_0$ and $\delta_{\mu,k}=\widehat\mu_0^{(-k)}-\mu_0$, Lemma~\ref{lem:orth} gives
\[
 h_k=-\frac{(1-D)\delta_{\pi,k}}{(1-\widehat\pi^{(-k)})(1-\pi_0)}R_0
      -w_{\widehat\pi^{(-k)}}\delta_{\mu,k}.
\]
In particular, the residual term obeys
\[
 \E\{(1-D)\delta_{\pi,k}^2\|R_0\|_{\cE}^2\}
 =\E\left[\delta_{\pi,k}^2(1-\pi_0)
     \E\{\|R_0\|_{\cE}^2\mid X,D=0\}\right]\leq C a_k^2.
\]
Bounded weights and $\|u+v\|^2\leq2\|u\|^2+2\|v\|^2$ therefore yield
\begin{equation}\label{eq:score_L2_bound}
 \|h_k\|_{L^2(P;\cE)}^2\leq C(a_k^2+\widetilde b_k^2),\qquad
 \|Ph_k\|_{\cE}\leq C a_k\widetilde b_k.
\end{equation}
The second inequality follows directly from the drift identity and Cauchy--Schwarz.

Let $H_k=(P_{n,k}-P)h_k$. Conditional on $\mathcal F_{-k}$, the centered summands are independent, so their cross inner products have zero mean and
\begin{align*}
 \E(\|H_k\|_{\cE}^2\mid\mathcal F_{-k})
 &=\frac1{n_k}\E\{\|h_k-Ph_k\|_{\cE}^2\mid\mathcal F_{-k}\}\\
 &\leq\frac{C}{n_k}(a_k^2+\widetilde b_k^2).
\end{align*}
For $\epsilon,\delta>0$, conditional Markov's inequality and $n_k/n\to\rho_k>0$ give
\[
 P(\sqrt n\|H_k\|_{\cE}>\epsilon)
 \leq P(a_k^2+\widetilde b_k^2>\delta)+P(n_k<\rho_kn/2)
       +\frac{2C\delta}{\rho_k\epsilon^2}.
\]
Taking $n\to\infty$ and then $\delta\downarrow0$ proves $H_k=o_p(n^{-1/2})$. Since $K$ is fixed,
\[
 \|R_{\mathrm{emp}}\|_{\cE}=o_p(n^{-1/2}),\qquad
 \|R_{\mathrm{drift}}\|_{\cE}\leq Ca_n\widetilde b_n=o_p(n^{-1/2}).
\]

For $i\in\cI_k$, define $e_i=\widehat{\Delta Y}_i-\Delta Y_i$. The contribution from replacing $\Delta Y_i$ by $\widehat{\Delta Y}_i$ satisfies
\[
 R_{\mathrm{rec}}=\frac1n\sum_{k}\sum_{i\in\cI_k}w_{\widehat\pi^{(-k)}}(W_i)e_i,
 \qquad \|R_{\mathrm{rec}}\|_{\cE}\leq C\max_i\|e_i\|_{\cE}\leq Cr_{n,\cE}.
\]
Thus $\|R_{\mathrm{rec}}\|_{\cE}=o_p(n^{-1/2})$. Substituting these bounds and $q_0-D\tau_0=p_0\phi$ into \eqref{eq:four_term_expansion} gives
\[
 \widehat p(\widehat\tau-\tau_0)=p_0P_n\phi+R_n,
 \qquad R_n=R_{\mathrm{emp}}+R_{\mathrm{drift}}+R_{\mathrm{rec}},
 \quad\|R_n\|_{\cE}=o_p(n^{-1/2}).
\]
The influence function has mean zero and finite second moment, whence
$\E\|P_n\phi\|_{\cE}^2=n^{-1}\E\|\phi\|_{\cE}^2$.
On the event $\widehat p\geq p_0/2$, whose probability tends to one,
\[
 \widehat\tau-\tau_0-P_n\phi
 =\frac{p_0-\widehat p}{\widehat p}P_n\phi+\frac{R_n}{\widehat p}
 =O_p(n^{-1})+o_p(n^{-1/2}).
\]
Here $\widehat p-p_0=O_p(n^{-1/2})$. Multiplication by $\sqrt n$ proves the expansion.
\end{proof}

\subsection*{Hilbert-space limit}
\begin{proof}[Proof of Theorem~\ref{thm:asym_clt}]
Let $S_n=n^{-1/2}\sum_i\phi(W_i)$ in $\cH$. Proposition~\ref{prop:eif} gives $P\phi=0$ and $P\|\phi\|_{\cH}^2<\infty$. The positive trace-class operator $\Sigma_{\mathrm{eff}}$ therefore defines a centered Gaussian element $\mathbb G$ in $\cH$ \citep{hsing2015theoretical}. Let $(e_j)$ be an orthonormal basis and $Q_J$ its first-$J$ projection. For fixed $J$, the multivariate central limit theorem yields
\[
 (\langle S_n,e_j\rangle)_{j=1}^J
 \Rightarrow N_J(0,\Gamma_J),\qquad
 (\Gamma_J)_{j\ell}=P\{\langle\phi,e_j\rangle\langle\phi,e_\ell\rangle\}.
\]
This is the law of $Q_J\mathbb G$ in the projected subspace. Independence and centering give the uniform-in-$n$ tail identity
\begin{align*}
 \E\|(I-Q_J)S_n\|_{\cH}^2
 &=\E\|(I-Q_J)\phi\|_{\cH}^2
 =\sum_{j>J}\E\langle\phi,e_j\rangle^2=:v_J\longrightarrow0,\\
 \E\|(I-Q_J)\mathbb G\|_{\cH}^2&=v_J.
\end{align*}
Parseval's identity and $P\|\phi\|_{\cH}^2<\infty$ justify the limit. In bounded-Lipschitz distance,
\[
 d_{\mathrm{BL}}\{\mathcal L(S_n),\mathcal L(\mathbb G)\}
 \leq 2\sqrt{v_J}
 +d_{\mathrm{BL}}\{\mathcal L(Q_JS_n),\mathcal L(Q_J\mathbb G)\}.
\]
Taking $n\to\infty$ and then $J\to\infty$ proves $S_n\Rightarrow\mathbb G$ in $\cH$. Slutsky's theorem and continuity of $g\mapsto\langle f,g\rangle$ give the estimator's limit and each directional limit \citep{vandervaart1998asymptotic}. This projection argument applies to any separable Hilbert space with a mean-zero square-integrable influence function.
\end{proof}

\subsection*{Continuous-domain limit}
\begin{proof}[Proof of Theorem~\ref{thm:asym_sup}]
Assumption~\ref{ass:asym_paths} imposes all rate and moment conditions in $\cV=H^1[0,1]$. Theorem~\ref{thm:asym_linear}, whose proof holds for a general separable Hilbert space, consequently gives
\[
 \sqrt n(\widehat\tau-\tau_0)=S_n+o_p(1)\quad\text{in }\cV,
 \qquad P\phi=0,\quad P\|\phi\|_{\cV}^2<\infty.
\]
Applying the preceding projection argument in $\cV$ yields $S_n\Rightarrow\mathbb G_{\cV}$ there. Its covariance is $f\mapsto\E\{\langle\phi,f\rangle_{\cV}\phi\}$; it is defined using the $\cV$ inner product, not the $L^2$ inner product.

Each $f\in H^1$ has an absolutely continuous representative and a point $u$ with $|f(u)|\leq\|f\|_2$. Hence
\begin{align*}
 |f(t)|&\leq|f(u)|+\int_0^1|f'(v)|\,dv\leq\|f\|_2+\|f'\|_2\leq\sqrt2\|f\|_{\cV},\\
 |f(t)-f(s)|&\leq\left|\int_s^tf'(v)\,dv\right|
 \leq |t-s|^{1/2}\|f'\|_2.
\end{align*}
Thus the embedding $\iota:\cV\to C[0,1]$ is continuous. Applying it to the expansion and limit gives
\[
 \begin{gathered}
 \|\sqrt n(\widehat\tau-\tau_0)-S_n\|_\infty
 \leq\sqrt2\|\sqrt n(\widehat\tau-\tau_0)-S_n\|_{\cV}=o_p(1),\\
 \iota S_n\Rightarrow\mathbb Z:=\iota\mathbb G_{\cV}\quad\text{in }C[0,1].
 \end{gathered}
\]
Continuous mapping and Slutsky's theorem prove the asserted uniform limit. Since evaluation $f\mapsto f(t)$ is a bounded linear functional on $\cV$, $\mathbb Z$ is centered Gaussian with $\Cov\{\mathbb Z(s),\mathbb Z(t)\}=\E\{\phi(s)\phi(t)\}=C(s,t)$. Its paths are continuous by construction.
\end{proof}

\subsection*{Covariance estimation and multiplier convergence}
\begin{proof}[Proof of Lemma~\ref{lem:bootstrap}]
Let $\phi_i=\phi(W_i)$ and $q_{0i}=q_0(W_i)$. Conditional Markov bounds applied to \eqref{eq:score_L2_bound} give
\[
 P_{n,k}\|h_k\|_{\cV}^2=o_p(1),\qquad
 P_n\|\widetilde q_i-q_{0i}\|_{\cV}^2
 \leq C\sum_k\frac{n_k}{n}P_{n,k}\|h_k\|_{\cV}^2+Cr_{n,\cV}^2=o_p(1).
\]
Indeed, for a nonnegative conditional mean $Z_n=o_p(1)$, the bound
$P(U_n>\epsilon)\leq P(Z_n>\delta)+\delta/\epsilon$ applies whenever $\E(U_n\mid\mathcal F_{-k})\leq Z_n$.
The exact influence-function difference is
\[
 \widehat\phi_i-\phi_i
 =\frac{\widetilde q_i-q_{0i}-D_i(\widehat\tau-\tau_0)}{\widehat p}
  +\left(\frac1{\widehat p}-\frac1{p_0}\right)(q_{0i}-D_i\tau_0).
\]
On $\widehat p\geq p_0/2$, its empirical squared norm is bounded by
\begin{align*}
 P_n\|\widehat\phi_i-\phi_i\|_{\cV}^2
 &\leq C\big\{P_n\|\widetilde q_i-q_{0i}\|_{\cV}^2+\|\widehat\tau-\tau_0\|_{\cV}^2\\
 &\hspace{35mm}+(\widehat p-p_0)^2P_n\|q_{0i}-D_i\tau_0\|_{\cV}^2\big\}=o_p(1).
\end{align*}
Therefore
\begin{equation*}
 e_n^2:=P_n\|\widehat\phi_i-\phi_i\|_{\cV}^2=o_p(1),\qquad
 P_n\|\widehat\phi_i\|_{\cV}^2\leq2P_n\|\phi_i\|_{\cV}^2+2e_n^2=O_p(1).
\end{equation*}

Set $v_i=\widehat\phi_i-P_n\widehat\phi$ and $v_i^0=\phi_i-P_n\phi$. Hilbert centering gives $P_n\|v_i-v_i^0\|_{\cV}^2\leq e_n^2$. By the Sobolev bound and Cauchy--Schwarz,
\[
 \sup_{s,t}\left|P_nv_i(s)v_i(t)-P_nv_i^0(s)v_i^0(t)\right|
 \leq C e_n\{(P_n\|\phi_i\|_{\cV}^2)^{1/2}+e_n\}=o_p(1).
\]
For $F_{s,t}(W)=\phi(s)\phi(t)$,
\[
 |F_{s,t}-F_{u,v}|\leq C\|\phi\|_{\cV}^2(|s-u|^{1/2}+|t-v|^{1/2}).
\]
Let $\mathcal A_\delta$ be a finite $\delta$-net of $[0,1]^2$. Then
\begin{align*}
 \sup_{s,t}|(P_n-P)F_{s,t}|
 &\leq\max_{(u,v)\in\mathcal A_\delta}|(P_n-P)F_{u,v}|\\
 &\quad+C\delta^{1/2}\{P_n\|\phi\|_{\cV}^2+P\|\phi\|_{\cV}^2\}.
\end{align*}
The finite maximum tends to zero by the law of large numbers. Taking $n\to\infty$ and then $\delta\downarrow0$, and using $\|P_n\phi\|_\infty^2=O_p(n^{-1})$, proves
\[
 \|\widehat C-C\|_\infty=o_p(1),\qquad
 \|\widehat\sigma-\sigma\|_\infty
 \leq\frac{\|\widehat\sigma^2-\sigma^2\|_\infty}{\inf_t\sigma(t)}=o_p(1).
\]

Fix an orthonormal basis $(e_j)$ of $\cV$, let $Q_J$ be its first-$J$ projection, and use the $\cV$ inner product in this projection argument. Define $\mathbb G_{n,0}^*=n^{-1/2}\sum_i\xi_i v_i^0$. For fixed $J$, its projected conditional covariance entries are
\[
 P_n\langle v_i^0,e_j\rangle\langle v_i^0,e_\ell\rangle
 \longrightarrow P\langle\phi,e_j\rangle\langle\phi,e_\ell\rangle,
 \qquad 1\leq j,\ell\leq J.
\]
The projected laws are Gaussian, hence converge weakly. For the tails,
\[
 \E_\xi\|(I-Q_J)\mathbb G_{n,0}^*\|_{\cV}^2
 =P_n\|(I-Q_J)v_i^0\|_{\cV}^2
 \leq P_n\|(I-Q_J)\phi_i\|_{\cV}^2.
\]
Let $d_{\mathrm{BL}}$ be bounded-Lipschitz distance and $\mathbb G_{\cV}$ the limit. Projection gives
\begin{align*}
 d_{\mathrm{BL}}\{\mathcal L_\xi(\mathbb G_{n,0}^*),\mathcal L(\mathbb G_{\cV})\}
 &\leq\{P_n\|(I-Q_J)\phi\|_{\cV}^2\}^{1/2}
       +\{P\|(I-Q_J)\phi\|_{\cV}^2\}^{1/2}\\
 &\quad+d_{\mathrm{BL}}\{\mathcal L_\xi(Q_J\mathbb G_{n,0}^*),
                         \mathcal L(Q_J\mathbb G_{\cV})\}.
\end{align*}
For fixed $J$, the last term tends to zero in probability and the empirical tail tends to its expectation. Parseval's identity and integrability make the expected tail tend to zero as $J\to\infty$. This proves conditional weak convergence in probability. Finally,
\[
 \E_\xi\|\mathbb G_n^*-\mathbb G_{n,0}^*\|_{\cV}^2
 =P_n\|v_i-v_i^0\|_{\cV}^2\leq e_n^2=o_p(1),
\]
which transfers the limit to the estimated-influence-function process.

For grid approximation, let $s_*=\inf_t\sigma(t)>0$ and $\Lambda_n=(P_n\|v_i'\|_2^2)^{1/2}=O_p(1)$. The reverse triangle inequality for Euclidean norms gives
\[
 |\widehat\sigma(t)-\widehat\sigma(u)|
 \leq\{P_n|v_i(t)-v_i(u)|^2\}^{1/2}\leq \Lambda_n|t-u|^{1/2}.
\]
On $E_n=\{\inf_t\widehat\sigma(t)\geq s_*/2\}$, which has probability tending to one, every $f\in\cV$ satisfies
\begin{align*}
 0&\leq\sup_t\left|\frac{f(t)}{\widehat\sigma(t)}\right|
       -\max_{u\in\cT_n}\left|\frac{f(u)}{\widehat\sigma(u)}\right|\\
 &\leq\delta_n^{1/2}\left\{\frac{2\|f'\|_2}{s_*}
              +\frac{4\Lambda_n\|f\|_\infty}{s_*^2}\right\}.
\end{align*}
Take $f=\sqrt n(\widehat\tau-\tau_0)$, whose $\cV$ norm is $O_p(1)$, to obtain the estimation-process bound. For $f=\mathbb G_n^*$,
$\E_\xi\|\mathbb G_n^*\|_{\cV}^2=P_n\|v_i\|_{\cV}^2=O_p(1)$;
conditional Markov's inequality gives the claimed multiplier bound.
\end{proof}

\subsection*{Critical values and simultaneous coverage}
\begin{proof}[Proof of Theorem~\ref{thm:band}]
Write $\widehat T_n=\|\sqrt n(\widehat\tau-\tau_0)/\widehat\sigma\|_\infty$ and let $F$ be the distribution function of $T$. Theorem~\ref{thm:asym_sup}, Lemma~\ref{lem:bootstrap} and $s_*>0$ imply
\[
 \widehat T_n\Rightarrow T,\qquad
 \|\widehat\sigma^{-1}-\sigma^{-1}\|_\infty=o_p(1).
\]
For a multiplier draw, the preceding grid bound yields
\[
 \left|T_n^*-\|\mathbb G_n^*/\sigma\|_\infty\right|
 \leq O_{P_\xi}(\delta_n^{1/2})
       +\|\mathbb G_n^*\|_\infty\|\widehat\sigma^{-1}-\sigma^{-1}\|_\infty
 =o_{P_\xi}(1)
\]
in probability. Conditional weak convergence and continuous mapping therefore give $F_n^*(x):=P_\xi(T_n^*\leq x)\to_p F(x)$ at every continuity point of $F$.

For the empirical distribution $\widehat F_{n,B}^*$ of the $B$ conditional i.i.d. draws, the empirical-distribution inequality gives
\[
 P_\xi\left\{\sup_x|\widehat F_{n,B}^*(x)-F_n^*(x)|>u\right\}
 \leq2e^{-2Bu^2}\qquad(u>0).
\]
Since $B_n\to\infty$, $\widehat F_{n,B}^*(x)\to_p F(x)$ at continuity points. Let $c=F^{-1}(1-\alpha)$. For sufficiently small $\epsilon>0$, continuity and strict increase near $c$ imply
\[
 F(c-\epsilon)<1-\alpha<F(c+\epsilon),\qquad
 P\{\widehat F_{n,B}^*(c-\epsilon)<1-\alpha<\widehat F_{n,B}^*(c+\epsilon)\}\to1.
\]
Thus $\widehat c_{1-\alpha}\to_p c$. Finally,
\begin{align*}
 P(\widehat T_n\leq c-\epsilon)-P(|\widehat c_{1-\alpha}-c|>\epsilon)
 &\leq P(\widehat T_n\leq\widehat c_{1-\alpha})\\
 &\leq P(\widehat T_n\leq c+\epsilon)+P(|\widehat c_{1-\alpha}-c|>\epsilon).
\end{align*}
Letting $n\to\infty$ and then $\epsilon\downarrow0$ gives $F(c)=1-\alpha$. The middle event is exactly the simultaneous containment event in \eqref{eq:scb}.
\end{proof}
\endgroup

\clearpage

\begin{center}
  {\Large\bfseries Supplementary Material for\\[0.5em]
  Semiparametric Inference for Functional Difference-in-Differences\par}
  \vspace{1em}
  Junzhu Nie \quad Chengxiu Ling \quad Mengfei Ran
\end{center}
\vspace{1em}

\setcounter{equation}{0}
\setcounter{table}{0}
\setcounter{figure}{0}
\setcounter{proposition}{0}

\makeatletter
\let\c@corollary\c@proposition
\makeatother

\renewcommand{\theproposition}{S\arabic{proposition}}
\renewcommand{\thecorollary}{\theproposition}
\renewcommand{\theequation}{S\arabic{equation}}
\renewcommand{\thetable}{S\arabic{table}}
\renewcommand{\thefigure}{S\arabic{figure}}

\renewcommand{\theHproposition}{supp.\arabic{proposition}}
\renewcommand{\theHcorollary}{supp.\arabic{proposition}}
\renewcommand{\theHequation}{supp.\arabic{equation}}
\renewcommand{\theHtable}{supp.\arabic{table}}
\renewcommand{\theHfigure}{supp.\arabic{figure}}

\section*{Variance estimation and confidence bands}
Let $W_1,\ldots,W_n$ be i.i.d. full-curve records $W=(\Delta Y,D,X)$, with $p_0=\E D>0$, and put $\cV=H^1[0,1]$. Change curves and regression fits lie in a common fixed finite-dimensional subspace of $\cV$. Write $P_nf=n^{-1}\sum_i f(W_i)$, and let $P_{n,k}$ and $P_{n,-k}$ average over fold $\cI_k$ and its complement. Folds are independent of the data and asymptotically balanced: $K$ is fixed and $n_k/n\to_p1/K$ for every $k$. Stack fitted nuisance coefficients in $\theta$, with probability limit $\theta_*$ and estimating equations satisfying $Pg(W;\theta_*)=0$.

\subsection*{Estimating equations}
Write $\Delta Y(t)=\bm b(t)^\top\bm c$, where $\bm c\in\mathbb R^L$. For feature vectors $z_\pi\in\mathbb R^{d_\pi}$ and $z_a\in\mathbb R^{d_a}$, define
\[
 \rho_\beta(x)=\operatorname{expit}\{z_\pi(x)^\top\beta\},\qquad
 \mu_{a,\theta}(x)(t)=\bm b(t)^\top B_a^\top z_a(x),\qquad a=0,1,
\]
with $B_a\in\mathbb R^{d_a\times L}$. The score uses $\pi_\theta=\min(1-c,\max(c,\rho_\beta))$, whereas logistic estimation uses the untruncated $\rho_\beta$. With column-wise vectorization, the full estimating-equation stack is
\begin{equation*}
 \begin{split}
 g_\pi(W;\theta)&=z_\pi(X)\{D-\rho_\beta(X)\},\\
 g_a(W;\theta)&=\1(D=a)\operatorname{vec}\left[z_a(X)
                  \{\bm c-B_a^\top z_a(X)\}^\top\right],\qquad a=0,1,\\
 g(W;\theta)&=(g_\pi^\top,g_0^\top,g_1^\top)^\top.
 \end{split}
\end{equation*}
Here $\1(\cdot)$ denotes the indicator function. At $\theta_*$, its Jacobian is block diagonal:
\begin{equation*}
 A=-\operatorname{diag}\left(
  \E[\rho_*(1-\rho_*)z_\pi z_\pi^\top],\,
  I_L\otimes\E[(1-D)z_0z_0^\top],\,
  I_L\otimes\E[Dz_1z_1^\top]\right).
\end{equation*}
CF--DR uses the $\pi,0$ blocks, OR the $0,1$ blocks, and IPW the $\pi$ block. Unused or known nuisance blocks are omitted; with no estimated nuisance the correction is zero.

For propensity and group regressions $\pi_\theta,\mu_{0,\theta},\mu_{1,\theta}$, the unnormalized method contributions are
\begin{align*}
 q_{\mathrm{OR}}(W;\theta)&=D\{\mu_{1,\theta}(X)-\mu_{0,\theta}(X)\},\\
 q_{\mathrm{IPW}}(W;\theta)&=\left\{D-\frac{(1-D)\pi_\theta(X)}{1-\pi_\theta(X)}\right\}\Delta Y,\\
 q_{\mathrm{CF\text{--}DR}}(W;\theta)&=\left\{D-\frac{(1-D)\pi_\theta(X)}{1-\pi_\theta(X)}\right\}
                  \{\Delta Y-\mu_{0,\theta}(X)\}.
\end{align*}
For method $a$, set
\[
 \widehat\tau_a=\frac{\sum_k(n_k/n)P_{n,k}q_a(W;\widehat\theta^{(-k)})}{\widehat p},
 \quad\widehat p=P_nD,\quad
 \tau_a^*=\frac{Pq_a(W;\theta_*)}{p_0}.
\]
For discrete observations, use $\widehat{\Delta Y}$ in both nuisance fitting and score evaluation. Define
\[
 A=P\partial_\theta g(W;\theta_*),\qquad
 Q_a=P\partial_\theta q_a(W;\theta_*).
\]
Here $Q_a$ is a linear map from the finite-dimensional parameter space into $\cV$.

\begin{proposition}[Finite-dimensional working-model correction]\label{prop:supp_working}
Assume bounded regressors, consistent local roots of the logistic and coefficient least-squares equations, nonsingular $A$, and finite second moments. Logistic fitting may include the Firth penalty. Suppose propensities stay uniformly away from zero and one in a neighborhood of $\theta_*$, clipping is inactive there, and $q_a$ is twice differentiable, with square-integrable first derivative and an integrable envelope for the second derivative in $\cV$ norm. The maximum period-wise error of $\widehat Y_{iT}$ satisfies $r_{n,\cV}=o_p(n^{-1/2})$. Then
\begin{equation}\label{eq:supp_expansion}
 \sqrt n(\widehat\tau_a-\tau_a^*)
 =\frac1{\sqrt n}\sum_i\phi_a(W_i)+o_p(1)\quad\text{in }\cV,
\end{equation}
where
\begin{equation}\label{eq:supp_if}
 \phi_a(W)=\frac{q_a(W;\theta_*)-D\tau_a^*-Q_aA^{-1}g(W;\theta_*)}{p_0}.
\end{equation}
Consistent plug-in influence functions give Gaussian multiplier inference for $\tau_a^*$ if its pointwise variances are positive, the evaluation-grid fill distance tends to zero, the multiplier repetition count tends to infinity, and the limiting studentized supremum distribution is continuous and strictly increasing near its required quantile. Inference targets the fATT when $\tau_a^*=\tau_0$.
\end{proposition}
\begin{proof}
Write $T_k=\cI_k^c$, $N_k=|T_k|$, $\Delta_k=\widehat\theta^{(-k)}-\theta_*$, and $g_*=g(W;\theta_*)$. For logistic fitting, the information matrix is $I_k(\beta)=\sum_{j\in T_k}\rho_j(1-\rho_j)z_{\pi j}z_{\pi j}^\top$. Bounded regressors and nonsingularity imply $I_k^{-1}=O_p(N_k^{-1})$ locally, while derivatives of $I_k$ are $O_p(N_k)$. Hence the gradient of $\tfrac12\log\det I_k$, and its first derivative, are $O_p(1)$. After division by $N_k$, the penalty changes the estimating equation by $O_p(n^{-1})$. Taylor expansion at $\theta_*$ therefore gives, with or without this penalty,
\[
 \Delta_k=-A^{-1}P_{n,-k}g_*+o_p(n^{-1/2})=O_p(n^{-1/2}).
\]
Let $\dot q_a=\partial_\theta q_a(W;\theta_*)$. Taylor expansion of a held-out score yields
\begin{align*}
 P_{n,k}q_a(W;\widehat\theta^{(-k)})
 &=P_{n,k}q_a(W;\theta_*)+Q_a\Delta_k
   +\{(P_{n,k}-P)\dot q_a\}\Delta_k+P_{n,k}r_{a,k},\\
 \|r_{a,k}(W)\|_{\cV}&\leq H(W)\|\Delta_k\|^2,\qquad PH<\infty.
\end{align*}
The Hilbert-space second-moment bound gives $\|(P_{n,k}-P)\dot q_a\|=O_p(n^{-1/2})$. Also $P_{n,k}H=O_p(1)$, so both remainder terms are $O_p(n^{-1})$. Define
\[
 \lambda_{j,n}=\sum_{k:j\in T_k}\frac{n_k}{n-n_k}.
\]
Since $n_k/n\to_p1/K$ and $K$ is fixed, $\max_j|\lambda_{j,n}-1|=o_p(1)$. The training averages satisfy the exact identity
\[
 \sum_k\frac{n_k}{n}P_{n,-k}g_*
 =P_ng_*+\frac1n\sum_j(\lambda_{j,n}-1)g_*(W_j).
\]
Conditionally on the fold partition $\mathcal I$, the summands in the last term are independent and centered. Hence
\[
 \E\left(\left\|\frac1{\sqrt n}\sum_j(\lambda_{j,n}-1)g_*(W_j)\right\|^2
       \Bigm|\mathcal I\right)
 =\frac{P\|g_*\|^2}{n}\sum_j(\lambda_{j,n}-1)^2=o_p(1).
\]
Conditional Markov's inequality makes the difference between the training average and $P_ng_*$ equal to $o_p(n^{-1/2})$. Consequently,
\begin{equation}\label{eq:supp_numerator}
 \widehat p\widehat\tau_a
 =P_nq_a(W;\theta_*)-Q_aA^{-1}P_ng_*+o_p(n^{-1/2}).
\end{equation}

Replacing $\Delta Y$ by $\widehat{\Delta Y}$ leaves the logistic equations unchanged. Bounded regressors and fixed-basis norm equivalence give, locally uniformly in $\theta$,
\[
 \|P_{n,-k}(\widehat g-g)\|\leq Cr_{n,\cV},\qquad
 \|P_{n,k}(\widetilde q_a-q_a)\|_{\cV}\leq Cr_{n,\cV}.
\]
The first bound and nonsingularity of the local regression matrices imply that fits based on $\widehat{\Delta Y}$ and $\Delta Y$ differ by $O_p(r_{n,\cV})$. The second bound controls score evaluation with the same coefficients. Both perturbations are $o_p(n^{-1/2})$, so \eqref{eq:supp_numerator} remains valid. Subtracting $\widehat p\tau_a^*$ gives
\[
 \widehat p(\widehat\tau_a-\tau_a^*)=p_0P_n\phi_a+o_p(n^{-1/2}).
\]
Here $P\phi_a=0$, $P\|\phi_a\|_{\cV}^2<\infty$, $P_n\phi_a=O_p(n^{-1/2})$, and $\widehat p-p_0=O_p(n^{-1/2})$. Replacing $\widehat p^{-1}$ by $p_0^{-1}$ contributes $O_p(n^{-1})$, proving \eqref{eq:supp_expansion}.

Write $\phi_a=\bm b^\top u_a$ in the common fixed basis. Smoothness, bounded regressors, and consistency of $\widehat A,\widehat Q_a,\widehat\theta,\widehat p,\widehat\tau_a$ imply
\[
 \frac1n\sum_i\|\widehat\phi_{a,i}-\phi_a(W_i)\|_{\cV}^2=o_p(1).
\]
Local differences are bounded by a constant times
$(1+\|\Delta Y\|_{\cV})$ multiplied by the parameter errors, plus $O_p(r_{n,\cV})$; the squared envelope is integrable. Norm equivalence and the law of large numbers then give
\[
 \widehat S_a:=\frac1n\sum_i(\widehat u_{a,i}-\bar{\widehat u}_a)
                 (\widehat u_{a,i}-\bar{\widehat u}_a)^\top
 \longrightarrow_p S_a:=P(u_au_a^\top).
\]
Conditionally, $n^{-1/2}\sum_i\xi_i(\widehat u_{a,i}-\bar{\widehat u}_a)$ has law $N(0,\widehat S_a)$, which converges to $N(0,S_a)$. The map $v\mapsto\bm b^\top v$ is continuous into $H^1$ and into $C[0,1]$. Thus the multiplier curve has the same limit as the estimation error, and its variance function converges uniformly. Sobolev continuity and variance positivity make grid approximation negligible. Conditional weak convergence, followed by the empirical-quantile law of large numbers and the stated quantile condition, proves band validity for $\tau_a^*$.

For CF--DR, the exact drift is
\[
 Pq_{\mathrm{CF\text{--}DR}}(W;\theta)-p_0\tau_0
 =\E\left\{\frac{\pi_\theta-\pi_0}{1-\pi_\theta}(\mu_{0,\theta}-\mu_0)\right\}.
\]
When both nuisances are correct, $Q_a=0$. With only one correct nuisance, the drift vanishes at $\theta_*$ but its derivative need not: consistency is retained, while nuisance estimation may contribute to \eqref{eq:supp_if}.
\end{proof}

\subsection*{Computing the correction}
For a perturbation $h=(h_\beta,H_0,H_1)$, put
\[
 \dot\pi_h=\1(c<\rho_\beta<1-c)\rho_\beta(1-\rho_\beta)z_\pi^\top h_\beta,
 \qquad
 \dot\mu_{a,h}(t)=\bm b(t)^\top H_a^\top z_a.
\]
Away from the clipping boundaries, differentiation gives
\begin{equation*}
 \begin{split}
 \dot q_{\mathrm{OR}}[h]&=D(\dot\mu_{1,h}-\dot\mu_{0,h}),\\
 \dot q_{\mathrm{IPW}}[h]&=-\frac{(1-D)\dot\pi_h}{(1-\pi_\theta)^2}\Delta Y,\\
 \dot q_{\mathrm{CF\text{--}DR}}[h]
 &=-\frac{(1-D)\dot\pi_h}{(1-\pi_\theta)^2}(\Delta Y-\mu_{0,\theta})
   -\left\{D-\frac{(1-D)\pi_\theta}{1-\pi_\theta}\right\}\dot\mu_{0,h}.
 \end{split}
\end{equation*}
Thus $Q_a[h]=\E\dot q_a[h]$. The derivative indicator is zero at a clipping boundary, an event of probability zero in these continuous designs.

The implementation differentiates each training fit separately. Assign case weights $\bm a=(a_1,\ldots,a_n)$ and let the training equation be
\[
 U_k(\theta,\bm a)=\sum_{j\in T_k}a_jg(\widehat W_j;\theta)+r_k(\theta,\bm a)=0,
\]
where $r_k$ is the Firth adjustment in the propensity block and zero in the regression blocks. At $(\widehat\theta^{(-k)},\bm1)$, define
\[
 J_k=\partial_\theta U_k,\qquad s_{kj}=\partial_{a_j}U_k,\qquad
 v_{kj}=\frac{\partial\widehat\theta^{(-k)}}{\partial a_j}=-J_k^{-1}s_{kj},\qquad j\in T_k.
\]
For $i\in\cI_k$, write $\widehat q_{a,i}=q_a(\widehat W_i;\widehat\theta^{(-k)})$ and
$Q_{ak}^{\Sigma}=\sum_{i\in\cI_k}\partial_\theta q_a(\widehat W_i;\widehat\theta^{(-k)})$.
Differentiating the weighted numerator and denominator gives
\begin{equation}\label{eq:supp_plugin}
 \widehat\phi_{a,j}
 =n\left.\frac{\partial\widehat\tau_a(\bm a)}{\partial a_j}\right|_{\bm a=\bm1}
 =\frac{\widehat q_{a,j}-D_j\widehat\tau_a+
       \sum_{k:j\in T_k}Q_{ak}^{\Sigma}v_{kj}}{\widehat p}.
\end{equation}
The first two terms account for evaluation and normalization; the sum accounts for every training fold containing observation $j$.

In a propensity training fold, suppress $k$ and set $z_j=z_\pi(X_j)$, $w_j=\rho_j(1-\rho_j)$,
$I=\sum_j a_jw_jz_jz_j^\top$, and $h_j=a_jw_jz_j^\top I^{-1}z_j$. The weighted Firth score is
\[
 U_\pi=\sum_j z_j\{a_j(D_j-\rho_j)+h_j(1/2-\rho_j)\}.
\]
At $\bm a=\bm1$, its case-weight derivative is
\begin{equation}\label{eq:supp_firth_weight}
 s_{\pi j}=z_j\{D_j-\rho_j+h_j(1/2-\rho_j)\}
 -w_j\sum_{\ell\in T_k}w_\ell(1/2-\rho_\ell)z_\ell
                  (z_\ell^\top I^{-1}z_j)^2.
\end{equation}
This follows from $\partial_{a_j}I^{-1}=-I^{-1}w_jz_jz_j^\top I^{-1}$. The Jacobian $J_k$ includes the derivative of the Firth adjustment. Regression blocks use $s_{kj}=g_j$ and the corresponding least-squares Jacobian. Known nuisance blocks are omitted.

\begin{proposition}[Consistency of the foldwise correction]\label{prop:supp_foldwise}
Under Proposition~\ref{prop:supp_working}, the contributions in \eqref{eq:supp_plugin} satisfy
\[
 n^{-1}\sum_j\|\widehat\phi_{a,j}-\phi_a(W_j)\|_{\cV}^2=o_p(1).
\]
Their centered empirical covariance therefore yields the multiplier inference stated there.
\end{proposition}
\begin{proof}
Foldwise laws of large numbers and local consistency give
$J_k/N_k\to_p A$ and $Q_{ak}^{\Sigma}/n_k\to_p Q_a$.
In \eqref{eq:supp_firth_weight}, $h_j=O_p(N_k^{-1})$ uniformly in $j$, and each term of the sum is $O_p(N_k^{-2})$ before summation. Thus $s_{\pi j}-g_{\pi j}=O_p(N_k^{-1})$ uniformly. The regression equations have no such adjustment. Smoothness and the square-integrable envelopes imply
\[
 \frac1{N_k}\sum_{j\in T_k}
 \|N_kv_{kj}+A^{-1}g(W_j;\theta_*)\|^2=o_p(1).
\]
By fixed-basis norm equivalence, replacing $\Delta Y$ by $\widehat{\Delta Y}$ adds $o_p(1)$ to this empirical norm. The fold-proportion condition gives
\[
 \max_j\left|\sum_{k:j\in T_k}\frac{n_k}{N_k}-1\right|=o_p(1).
\]
Together with $P_n\|g_*\|^2=O_p(1)$, this shows that the training contribution in \eqref{eq:supp_plugin} converges in empirical squared $\cV$ norm to $-Q_aA^{-1}g(W_j;\theta_*)$. The same argument applies to the held-out score, and $\widehat p\to_p p_0$, $\widehat\tau_a\to_p\tau_a^*$. Combining these limits proves the claim. The finite-basis covariance and Gaussian-multiplier argument in Proposition~\ref{prop:supp_working} then applies.
\end{proof}

Under a fixed unit-level eligibility rule with positive retention probability, independent fold assignment remains independent of retained records conditional on membership, and $n_k^{\rm ret}/n^{\rm ret}\to_p1/K$. Propositions~\ref{prop:supp_working}--\ref{prop:supp_foldwise} thus apply to the retained-population law when their remaining conditions hold there. Intervals for rules chosen after inspecting outcomes remain exploratory and nominal.

Standard errors and multiplier bands use empirically centered contributions from \eqref{eq:supp_plugin}. Uncorrected CF--DR bands omit the training-fold sum, retaining the same point estimate, data, folds, and multiplier draws. For Naive DiD, put $\nu_d=\E(\Delta Y\mid D=d)$; its influence function about $\nu_1-\nu_0$ is
\[
 \phi_{\mathrm N}(W)=\frac{D}{p_0}(\Delta Y-\nu_1)
                    -\frac{1-D}{1-p_0}(\Delta Y-\nu_0),
\]
estimated with the two sample means and $\widehat p$. Oracle uses the full-curve influence function with true nuisances, substituting $\widehat p$ and its own effect estimate for $p_0$ and $\tau_0$.

\Needspace{6\baselineskip}
\subsection*{Simultaneous intervals at observed time points}
\begin{corollary}[Fixed-coordinate joint intervals]\label{cor:supp_fixed_grid}
Let $\Delta Y\in\mathbb R^m$ be fully observed, with $m$ fixed. Assume the sampling, fold-proportion, moment, and working-model estimation conditions of Proposition~\ref{prop:supp_working}, with $\cV=\mathbb R^m$ and $r_{n,\cV}=0$. Define $\phi_a$ by \eqref{eq:supp_if}, and estimate it by the foldwise contributions \eqref{eq:supp_plugin}. Let
\[
 \Sigma_a=P(\phi_a\phi_a^\top),\qquad
 \widehat\Sigma_a=\frac1n\sum_i
  (\widehat\phi_{a,i}-\bar{\widehat\phi}_a)
  (\widehat\phi_{a,i}-\bar{\widehat\phi}_a)^\top,
\]
and suppose $\Sigma_{a,hh}>0$ for $h=1,\ldots,m$. From $B_n\to\infty$ independent conditional draws $G^*\sim N_m(0,\widehat\Sigma_a)$, let $\widehat c_{1-\alpha}$ be the empirical $(1-\alpha)$ quantile of $\max_h|G_h^*|/\sqrt{\widehat\Sigma_{a,hh}}$. Then
\[
 P\left\{ |\widehat\tau_{a,h}-\tau_{a,h}^*|
       \leq\widehat c_{1-\alpha}\sqrt{\widehat\Sigma_{a,hh}/n}
       \text{ for every }h=1,\ldots,m\right\}\longrightarrow1-\alpha.
\]
For a fixed vector $v$ with $v^\top\Sigma_av>0$, the marginal normal interval
$v^\top\widehat\tau_a\pm z_{1-\alpha/2}\sqrt{v^\top\widehat\Sigma_av/n}$ is also asymptotically valid for $v^\top\tau_a^*$.
\end{corollary}
\begin{proof}
The Taylor and fold-averaging arguments in Propositions~\ref{prop:supp_working}--\ref{prop:supp_foldwise} apply directly in $\mathbb R^m$. Thus
\[
 \sqrt n(\widehat\tau_a-\tau_a^*)\Rightarrow G\sim N_m(0,\Sigma_a),\qquad
 P_n\|\widehat\phi_a-\phi_a\|^2=o_p(1),\qquad
 \widehat\Sigma_a\to_p\Sigma_a.
\]
The last assertion follows from Cauchy--Schwarz and the law of large numbers. Conditionally, the Gaussian multiplier sum has covariance $\widehat\Sigma_a$, so its law converges to that of $G$. Positive coordinate variances permit studentization and continuous mapping through the finite maximum.

The limiting maximum $T_m=\max_h|G_h|/\sqrt{\Sigma_{a,hh}}$ has a continuous, strictly increasing distribution on $(0,\infty)$, even if $\Sigma_a$ is singular. To see this, write $G=LZ$, where $L$ has full column rank $r\geq1$ and $Z\sim N_r(0,I_r)$. The map $z\mapsto\max_h|(Lz)_h|/\sqrt{\Sigma_{a,hh}}$ is a norm. In the polar decomposition $Z=RU$, $R$ has a strictly positive chi density on $(0,\infty)$ and is independent of $U$; the norm of $U$ is bounded above and away from zero. Conditioning on $U$ proves the asserted continuity and strict increase. Conditional distribution convergence and $B_n\to\infty$ therefore imply quantile consistency, giving the joint coverage statement. Applying the scalar CLT and Slutsky's theorem to $v^\top\widehat\tau_a$ gives the marginal interval.
\end{proof}
The intervals target $\tau_a^*$ under the working-model conditions. Under the identification assumptions for the population analyzed, CF--DR has $\tau_a^*$ equal to its ATT if either nuisance model is correct. The same finite-dimensional argument applies to Naive DiD using its two-group influence function, with target $\nu_1-\nu_0$.

\section*{Additional simulation results}
\subsection*{Bias, standard errors, and Monte Carlo uncertainty}
Replication and multiplier counts follow main-paper Section~5.4.
For successful replications $\mathcal S$, let $R_s=|\mathcal S|$,
$e_r(t)=\widehat\tau_r(t)-\tau_0(t)$, and $\widehat{\mathrm{se}}_r(t)$ denote the estimated standard error. Define
\begin{align*}
 \bar e(t)&=R_s^{-1}\sum_{r\in\mathcal S}e_r(t),&
 s_{\mathrm{MC}}^2(t)&=(R_s-1)^{-1}\sum_{r\in\mathcal S}\{e_r(t)-\bar e(t)\}^2,\\
 \mathrm{SBias}(t)&=\bar e(t)/s_{\mathrm{MC}}(t),&
 \mathrm{SEratio}(t)&=\frac{R_s^{-1}\sum_{r\in\mathcal S}\widehat{\mathrm{se}}_r(t)}{s_{\mathrm{MC}}(t)}.
\end{align*}
SBias expresses mean error in Monte Carlo standard-deviation units; SEratio compares the mean estimated
standard error with the Monte Carlo standard deviation. Their reference values are zero and one, respectively.
Both require $R_s\geq2$ and $s_{\mathrm{MC}}(t)>0$. Tables~\ref{tab:supp_bias_se}--\ref{tab:supp_bias_se_other} report the three
selected times. The omitted S3 regression term is nonzero at $t=0.50$ but vanishes at $t=0.25,0.75$.

Firth fitting uses at most 250 iterations, relative update tolerance $10^{-8}$, and line search.
Solved matrices require reciprocal condition number greater than $10^{-10}$; the adjusted-score Jacobian
uses symmetric differences with increments $10^{-5}(1+|\widehat\beta_j|)$.
Empty required groups or nonfinite coefficients or estimates constitute point-estimation failure;
inference additionally requires positive finite grid variances and a finite multiplier critical value.
Failures count as noncoverage, and samples are never replaced. Empty nonparametric bins use the
prespecified training-group fallback. No procedures failed, so $N_{\cap}=R$ for paired comparisons.

For coverage count $k$ and $\widehat p=k/R$, Table~\ref{tab:supp_mcse} reports
$100\sqrt{\widehat p(1-\widehat p)/R}$ percentage points.
At observed coverage of zero or one, this plug-in MCSE is zero; uncertainty is described by the 95\% Wilson interval
\[
 \frac{\widehat p+z^2/(2R)\ \pm\ z\sqrt{\widehat p(1-\widehat p)/R+z^2/(4R^2)}}{1+z^2/R},
 \qquad z=1.96.
\]
Its boundary values are approximately $[0,0.38]\%$ and $[99.62,100]\%$.
Coverage differences $\Delta C$ use the paired MCSE in main-paper Section~5.4; both are in percentage points.
Error and width differences average within-replication differences over the $N_{\cap}$ jointly successful
replications; these summaries are undefined when $N_{\cap}=0$.
Width ratios use the same set. AW denotes mean simultaneous band width.
Inputs, assignments, folds, and multiplier streams are shared wherever the comparison permits.

\subsection*{Additional simulation designs}
Comparisons of nuisance fits, variance correction, true-curve inputs, and grid resolution cover S1--S6.
Observation-density, noise, basis-truncation, and clipping checks focus on S4 and S6 and change one feature
at a time. The nonparametric experiment is a separate configuration.

Table~\ref{tab:supp_nuisance} uses the true curves $Y_{iT}$ with true nuisances (TT), estimated propensity
only (ET), estimated control regression only (TE), or both estimated nuisances (EE). The R fit estimates
both nuisances using $\widehat Y_{iT}$ from the noisy measurements. Estimated blocks receive the correction; TT and R reuse the principal
Oracle and CF--DR runs. Table~\ref{tab:supp_correction} holds point estimates, nuisance fits, data, folds,
and multiplier draws fixed. Table~\ref{tab:supp_latent} compares each feasible estimator using
$\widehat Y_{iT}$ and $Y_{iT}$, refitting outcome-dependent nuisances for each input. Paired differences
are results based on $\widehat Y_{iT}$ minus those based on $Y_{iT}$.

The maximum curve error is $r_{n,H^1}=\max_{i,T}\|\widehat Y_{iT}-Y_{iT}\|_{H^1}$, computed from the known basis
and its derivative Gram matrix. Regular-grid Fourier least squares gives coefficient noise
$N(0,\sigma_\varepsilon^2I_{13}/m)$ for $m$ observations. Nested grids of sizes $m,M$ have cross-covariance
$\sigma_\varepsilon^2I_{13}/\max(m,M)$; periods and units are independent.
Density comparisons use this joint law, while noise comparisons rescale the same errors. Each setting
uses 13 basis functions and the principal grid, with nuisance fits repeated; Table~\ref{tab:supp_observation}
reports mean $\sqrt n\,r_{n,H^1}$. Grid refinement reevaluates the fitted curves and multiplier draws
on $\{j/(2n):j=0,\ldots,2n\}$ without refitting (Table~\ref{tab:supp_grid}).

Let $\Pi_9$ project onto Fourier frequencies $k\leq4$. Under Proposition~\ref{prop:supp_working}'s
conditions for $W^{(9)}=(\Pi_9\Delta Y,D,X)$ and correctness of either projected-data nuisance model,
CF--DR targets $\Pi_9\tau_0$. Thus
\[
 \widehat\tau_9-\tau_0=(\widehat\tau_9-\Pi_9\tau_0)+(\Pi_9\tau_0-\tau_0).
\]
The first term has the working-model expansion; the second is truncation bias.
Table~\ref{tab:supp_truncation} evaluates the same bands against both targets, refitting nuisances for
each basis dimension. At $L=13$, the projected and full targets coincide.

Clipping comparisons report point estimation and weight summaries (Table~\ref{tab:supp_clipping}).
CF--DR and IPW retain the same propensity fits and vary the score weights across thresholds.
OR and Naive DiD are unaffected by clipping, and Oracle uses true propensities.
The largest threshold remains active in S6. The clipped fraction counts all held-out predictions changed by truncation.
For $w_i=(1-D_i)\widehat\pi_i/(1-\widehat\pi_i)$, control ESS is $(\sum_iw_i)^2/\sum_iw_i^2$, reported relative
to $n_0=\sum_i(1-D_i)$. Panel B reports replication means of the clipped percentage, relative ESS,
and within-replication maximum odds; these summaries are common to CF--DR and IPW.

The nonparametric design follows main-paper Section~5.3. CF--DR, OR, and IPW share applicable bin fits;
Oracle uses true curves and nuisances. Inference is reported for CF--DR using the efficient-score
influence function in main-paper Algorithm~1.

\begin{table}[!htbp]
\centering
\caption{Standardized bias and standard-error ratios at $t=0.25$.}
\label{tab:supp_bias_se}
\begingroup
\footnotesize
\renewcommand{\arraystretch}{1}
\begin{tabular*}{\textwidth}{@{\extracolsep{\fill}}ll*{6}{r}@{}}
\toprule
Scenario & Method & \multicolumn{3}{c}{SBias} & \multicolumn{3}{c}{SEratio} \\
\cmidrule(lr){3-5}\cmidrule(lr){6-8}
 &  & $n{=}200$ & $n{=}400$ & $n{=}800$ & $n{=}200$ & $n{=}400$ & $n{=}800$ \\
\midrule
S1 & CF--DR & -0.038 & 0.024 & -0.040 & 1.124 & 1.072 & 1.055 \\
 & OR & -0.050 & 0.015 & -0.041 & 0.939 & 0.979 & 1.007 \\
 & IPW & -0.201 & -0.100 & -0.129 & 1.029 & 1.027 & 1.034 \\
 & Naive DiD & 0.428 & 0.741 & 0.985 & 0.976 & 0.998 & 0.994 \\
 & Oracle & -0.050 & 0.009 & -0.037 & 0.982 & 1.004 & 1.018 \\
\addlinespace[4pt]
S2 & CF--DR & -0.012 & 0.048 & 0.001 & 1.143 & 1.017 & 1.039 \\
 & OR & -0.012 & 0.024 & 0.006 & 0.967 & 0.942 & 0.990 \\
 & IPW & -0.234 & -0.148 & -0.197 & 1.048 & 0.973 & 1.016 \\
 & Naive DiD & 0.485 & 0.700 & 1.010 & 1.007 & 0.967 & 1.019 \\
 & Oracle & -0.014 & 0.023 & 0.006 & 0.992 & 0.958 & 0.992 \\
\addlinespace[4pt]
S3 & CF--DR & -0.017 & -0.016 & 0.024 & 1.172 & 1.057 & 1.001 \\
 & OR & -0.012 & -0.004 & 0.020 & 0.983 & 0.972 & 0.957 \\
 & IPW & -0.171 & -0.132 & -0.057 & 1.076 & 1.015 & 0.980 \\
 & Naive DiD & 0.501 & 0.722 & 0.994 & 1.027 & 1.016 & 0.963 \\
 & Oracle & -0.023 & -0.018 & 0.024 & 1.024 & 0.991 & 0.972 \\
\addlinespace[4pt]
S4 & CF--DR & 0.037 & 0.034 & -0.041 & 1.211 & 1.127 & 1.032 \\
 & OR & 0.038 & 0.023 & -0.036 & 0.901 & 0.966 & 0.963 \\
 & IPW & -0.203 & -0.154 & -0.167 & 1.076 & 1.061 & 1.000 \\
 & Naive DiD & 0.473 & 0.731 & 0.944 & 0.942 & 1.006 & 0.990 \\
 & Oracle & 0.026 & 0.031 & -0.035 & 0.965 & 0.988 & 0.981 \\
\addlinespace[4pt]
S5 & CF--DR & 0.020 & 0.017 & -0.084 & 1.201 & 1.154 & 1.025 \\
 & OR & 0.027 & 0.016 & -0.088 & 0.935 & 1.018 & 0.961 \\
 & IPW & -0.230 & -0.169 & -0.218 & 1.050 & 1.077 & 0.999 \\
 & Naive DiD & 0.523 & 0.729 & 0.877 & 0.996 & 1.023 & 0.966 \\
 & Oracle & 0.031 & 0.034 & -0.075 & 0.989 & 1.029 & 0.969 \\
\addlinespace[4pt]
S6 & CF--DR & 0.034 & 0.010 & -0.010 & 0.841 & 0.773 & 1.055 \\
 & OR & 0.020 & 0.002 & -0.014 & 0.924 & 0.932 & 0.983 \\
 & IPW & -0.135 & -0.156 & -0.147 & 0.662 & 0.824 & 1.006 \\
 & Naive DiD & 1.882 & 2.685 & 3.522 & 1.002 & 1.006 & 0.947 \\
 & Oracle & 0.007 & -0.017 & -0.042 & 0.858 & 0.915 & 0.965 \\
\bottomrule
\end{tabular*}
\endgroup
\end{table}

\begin{table}[!htbp]
\centering
\caption{Standardized bias and standard-error ratios at $t=0.50$ and $t=0.75$.}
\label{tab:supp_bias_se_other}
\begingroup
\footnotesize
\setlength{\tabcolsep}{2pt}
\renewcommand{\arraystretch}{1}
\begin{tabular*}{\textwidth}{@{\extracolsep{\fill}}ll@{\hspace{7pt}}*{12}{r}@{}}
\toprule
Scenario & Method & \multicolumn{3}{c}{SBias $(0.50)$} & \multicolumn{3}{c}{SEratio $(0.50)$} & \multicolumn{3}{c}{SBias $(0.75)$} & \multicolumn{3}{c}{SEratio $(0.75)$} \\
\cmidrule(lr){3-5}\cmidrule(lr){6-8}\cmidrule(lr){9-11}\cmidrule(lr){12-14}
 & & $n{=}200$ & $n{=}400$ & $n{=}800$ & $n{=}200$ & $n{=}400$ & $n{=}800$ & $n{=}200$ & $n{=}400$ & $n{=}800$ & $n{=}200$ & $n{=}400$ & $n{=}800$ \\
\midrule
S1 & CF--DR & 0.030 & 0.031 & -0.001 & 1.122 & 1.091 & 1.012 & 0.054 & -0.008 & -0.009 & 1.175 & 1.045 & 1.033 \\
 & OR & 0.019 & 0.026 & 0.002 & 0.964 & 1.014 & 0.964 & 0.041 & -0.019 & -0.004 & 0.986 & 0.973 & 0.997 \\
 & IPW & -0.070 & -0.047 & -0.046 & 1.035 & 1.050 & 0.990 & 0.035 & -0.016 & -0.008 & 1.085 & 1.012 & 1.010 \\
 & Naive DiD & 0.253 & 0.386 & 0.488 & 1.008 & 1.055 & 0.980 & -0.482 & -0.741 & -1.026 & 1.004 & 0.984 & 1.001 \\
 & Oracle & 0.011 & 0.019 & 0.008 & 0.989 & 1.042 & 0.977 & 0.025 & -0.015 & -0.005 & 1.024 & 0.991 & 0.997 \\
\addlinespace[4pt]
S2 & CF--DR & -0.007 & -0.010 & 0.008 & 1.129 & 1.054 & 1.022 & 0.030 & -0.009 & -0.021 & 1.153 & 1.074 & 1.008 \\
 & OR & 0.004 & -0.014 & -0.004 & 0.960 & 0.966 & 0.983 & 0.018 & 0.000 & -0.027 & 0.991 & 0.983 & 0.965 \\
 & IPW & -0.141 & -0.142 & -0.136 & 1.043 & 1.014 & 1.001 & 0.056 & 0.044 & 0.044 & 1.067 & 1.026 & 0.985 \\
 & Naive DiD & 0.234 & 0.333 & 0.445 & 1.007 & 0.994 & 0.973 & -0.488 & -0.700 & -1.019 & 1.010 & 0.984 & 0.995 \\
 & Oracle & -0.001 & -0.007 & 0.004 & 1.021 & 0.995 & 0.992 & 0.000 & -0.017 & -0.019 & 1.010 & 1.004 & 0.983 \\
\addlinespace[4pt]
S3 & CF--DR & -0.022 & 0.008 & 0.010 & 1.143 & 1.077 & 1.028 & 0.009 & 0.008 & -0.015 & 1.143 & 1.069 & 1.051 \\
 & OR & -0.134 & -0.183 & -0.256 & 0.944 & 0.998 & 0.998 & 0.023 & 0.009 & -0.024 & 0.962 & 0.995 & 1.019 \\
 & IPW & -0.108 & -0.065 & -0.039 & 1.042 & 1.035 & 1.013 & 0.016 & 0.006 & -0.021 & 1.040 & 1.033 & 1.030 \\
 & Naive DiD & 0.199 & 0.345 & 0.490 & 0.994 & 1.012 & 1.016 & -0.514 & -0.714 & -1.066 & 1.018 & 1.011 & 1.027 \\
 & Oracle & -0.003 & -0.001 & 0.027 & 0.996 & 1.013 & 1.019 & 0.017 & 0.018 & -0.028 & 1.006 & 1.009 & 1.023 \\
\addlinespace[4pt]
S4 & CF--DR & 0.040 & 0.056 & -0.026 & 1.185 & 1.082 & 1.053 & 0.050 & -0.019 & -0.022 & 1.182 & 1.106 & 1.073 \\
 & OR & 0.026 & 0.050 & -0.029 & 0.920 & 0.964 & 0.984 & 0.018 & -0.026 & -0.027 & 0.942 & 0.974 & 1.026 \\
 & IPW & -0.159 & -0.085 & -0.129 & 1.071 & 1.026 & 1.017 & 0.026 & -0.029 & -0.031 & 1.048 & 1.039 & 1.050 \\
 & Naive DiD & 0.818 & 1.237 & 1.649 & 0.958 & 1.005 & 1.012 & -0.454 & -0.718 & -1.072 & 1.018 & 1.006 & 1.040 \\
 & Oracle & 0.027 & 0.058 & -0.032 & 0.979 & 0.998 & 1.002 & 0.023 & -0.027 & -0.032 & 0.993 & 1.007 & 1.039 \\
\addlinespace[4pt]
S5 & CF--DR & 0.004 & -0.015 & -0.046 & 1.193 & 1.121 & 1.086 & 0.023 & -0.004 & 0.013 & 1.261 & 1.106 & 1.076 \\
 & OR & 0.002 & -0.015 & -0.050 & 0.938 & 0.984 & 1.026 & 0.032 & -0.004 & -0.002 & 0.967 & 0.964 & 1.011 \\
 & IPW & -0.208 & -0.156 & -0.149 & 1.073 & 1.038 & 1.051 & 0.030 & -0.011 & 0.008 & 1.114 & 1.034 & 1.046 \\
 & Naive DiD & 0.838 & 1.174 & 1.624 & 0.989 & 1.013 & 1.015 & -0.463 & -0.690 & -0.981 & 1.011 & 0.983 & 1.003 \\
 & Oracle & 0.000 & -0.020 & -0.047 & 0.991 & 1.013 & 1.040 & 0.046 & -0.011 & 0.001 & 1.010 & 0.991 & 1.026 \\
\addlinespace[4pt]
S6 & CF--DR & -0.015 & 0.058 & 0.036 & 0.784 & 0.777 & 1.054 & 0.023 & -0.002 & 0.027 & 0.843 & 0.873 & 1.029 \\
 & OR & -0.042 & 0.022 & -0.006 & 0.897 & 0.924 & 1.006 & -0.033 & -0.059 & -0.036 & 0.916 & 0.934 & 0.987 \\
 & IPW & -0.058 & -0.011 & -0.031 & 0.596 & 0.782 & 1.000 & 0.071 & 0.072 & 0.053 & 0.685 & 0.783 & 0.957 \\
 & Naive DiD & -0.005 & 0.023 & 0.003 & 1.046 & 0.981 & 0.979 & -1.904 & -2.658 & -3.863 & 1.006 & 0.986 & 1.021 \\
 & Oracle & -0.033 & -0.002 & 0.024 & 0.866 & 0.947 & 0.969 & -0.052 & -0.011 & 0.019 & 0.864 & 0.893 & 0.944 \\
\bottomrule
\end{tabular*}
\endgroup
\end{table}

\begin{table}[!htbp]
\centering
\caption{Entry-specific Monte Carlo standard errors for the coverage percentages in main-paper Table~3.}
\label{tab:supp_mcse}
\begingroup
\footnotesize
\setlength{\tabcolsep}{2pt}
\renewcommand{\arraystretch}{1}
\begin{tabular*}{\textwidth}{@{\extracolsep{\fill}}ll@{\hspace{7pt}}*{12}{r}@{}}
\toprule
Scenario & Method & \multicolumn{3}{c}{$t=0.25$} & \multicolumn{3}{c}{$t=0.50$} & \multicolumn{3}{c}{$t=0.75$} & \multicolumn{3}{c}{SCB} \\
\cmidrule(lr){3-5}\cmidrule(lr){6-8}\cmidrule(lr){9-11}\cmidrule(lr){12-14}
 & & $n{=}200$ & $n{=}400$ & $n{=}800$ & $n{=}200$ & $n{=}400$ & $n{=}800$ & $n{=}200$ & $n{=}400$ & $n{=}800$ & $n{=}200$ & $n{=}400$ & $n{=}800$ \\
\midrule
S1 & CF--DR & 0.56 & 0.63 & 0.59 & 0.49 & 0.59 & 0.65 & 0.47 & 0.67 & 0.66 & 0.46 & 0.65 & 0.60 \\
 & OR & 0.77 & 0.73 & 0.70 & 0.76 & 0.69 & 0.73 & 0.70 & 0.79 & 0.66 & 0.86 & 0.79 & 0.71 \\
 & IPW & 0.69 & 0.67 & 0.63 & 0.61 & 0.61 & 0.70 & 0.58 & 0.70 & 0.66 & 0.60 & 0.73 & 0.68 \\
 & Naive DiD & 0.84 & 1.02 & 1.21 & 0.63 & 0.69 & 0.88 & 0.86 & 1.06 & 1.23 & 0.94 & 1.07 & 1.24 \\
 & Oracle & 0.75 & 0.66 & 0.63 & 0.64 & 0.61 & 0.74 & 0.57 & 0.75 & 0.68 & 0.66 & 0.69 & 0.70 \\
\addlinespace[4pt]
S2 & CF--DR & 0.51 & 0.64 & 0.63 & 0.57 & 0.67 & 0.66 & 0.47 & 0.56 & 0.67 & 0.43 & 0.63 & 0.60 \\
 & OR & 0.73 & 0.80 & 0.71 & 0.79 & 0.72 & 0.70 & 0.76 & 0.68 & 0.71 & 0.81 & 0.81 & 0.72 \\
 & IPW & 0.64 & 0.68 & 0.71 & 0.66 & 0.67 & 0.70 & 0.57 & 0.61 & 0.70 & 0.56 & 0.73 & 0.70 \\
 & Naive DiD & 0.83 & 1.03 & 1.19 & 0.73 & 0.75 & 0.85 & 0.85 & 1.02 & 1.24 & 0.91 & 1.07 & 1.21 \\
 & Oracle & 0.69 & 0.74 & 0.70 & 0.67 & 0.67 & 0.70 & 0.67 & 0.60 & 0.66 & 0.62 & 0.72 & 0.65 \\
\addlinespace[4pt]
S3 & CF--DR & 0.43 & 0.67 & 0.72 & 0.52 & 0.54 & 0.63 & 0.52 & 0.54 & 0.60 & 0.41 & 0.59 & 0.65 \\
 & OR & 0.69 & 0.77 & 0.80 & 0.79 & 0.70 & 0.73 & 0.76 & 0.70 & 0.64 & 0.85 & 0.85 & 0.78 \\
 & IPW & 0.56 & 0.72 & 0.73 & 0.63 & 0.63 & 0.65 & 0.64 & 0.60 & 0.63 & 0.62 & 0.66 & 0.68 \\
 & Naive DiD & 0.84 & 0.95 & 1.21 & 0.71 & 0.71 & 0.80 & 0.82 & 0.97 & 1.20 & 0.84 & 1.00 & 1.26 \\
 & Oracle & 0.62 & 0.77 & 0.70 & 0.66 & 0.63 & 0.62 & 0.70 & 0.68 & 0.65 & 0.61 & 0.75 & 0.70 \\
\addlinespace[4pt]
S4 & CF--DR & 0.45 & 0.53 & 0.65 & 0.43 & 0.57 & 0.63 & 0.37 & 0.56 & 0.56 & 0.40 & 0.50 & 0.51 \\
 & OR & 0.82 & 0.73 & 0.76 & 0.83 & 0.80 & 0.69 & 0.74 & 0.73 & 0.63 & 1.00 & 0.81 & 0.68 \\
 & IPW & 0.58 & 0.69 & 0.71 & 0.57 & 0.66 & 0.66 & 0.59 & 0.65 & 0.60 & 0.59 & 0.65 & 0.59 \\
 & Naive DiD & 0.93 & 0.97 & 1.13 & 1.14 & 1.35 & 1.54 & 0.77 & 0.98 & 1.17 & 1.04 & 1.32 & 1.48 \\
 & Oracle & 0.77 & 0.73 & 0.70 & 0.70 & 0.71 & 0.66 & 0.65 & 0.66 & 0.63 & 0.63 & 0.73 & 0.60 \\
\addlinespace[4pt]
S5 & CF--DR & 0.41 & 0.50 & 0.66 & 0.45 & 0.51 & 0.61 & 0.34 & 0.49 & 0.56 & 0.43 & 0.43 & 0.58 \\
 & OR & 0.77 & 0.67 & 0.76 & 0.77 & 0.70 & 0.69 & 0.70 & 0.73 & 0.61 & 0.90 & 0.79 & 0.76 \\
 & IPW & 0.57 & 0.59 & 0.73 & 0.49 & 0.64 & 0.67 & 0.45 & 0.60 & 0.58 & 0.53 & 0.58 & 0.63 \\
 & Naive DiD & 0.88 & 0.97 & 1.16 & 1.09 & 1.27 & 1.52 & 0.84 & 0.98 & 1.15 & 1.04 & 1.22 & 1.50 \\
 & Oracle & 0.68 & 0.64 & 0.73 & 0.60 & 0.68 & 0.66 & 0.63 & 0.67 & 0.58 & 0.67 & 0.69 & 0.66 \\
\addlinespace[4pt]
S6 & CF--DR & 0.34 & 0.44 & 0.60 & 0.42 & 0.48 & 0.54 & 0.41 & 0.55 & 0.56 & 0.47 & 0.67 & 0.67 \\
 & OR & 0.86 & 0.83 & 0.73 & 0.87 & 0.77 & 0.68 & 0.80 & 0.78 & 0.69 & 1.09 & 0.90 & 0.76 \\
 & IPW & 0.62 & 0.60 & 0.66 & 0.46 & 0.53 & 0.52 & 0.52 & 0.57 & 0.57 & 0.68 & 0.75 & 0.71 \\
 & Naive DiD & 1.58 & 1.34 & 0.66 & 0.66 & 0.73 & 0.70 & 1.58 & 1.34 & 0.56 & 1.58 & 1.20 & 0.36 \\
 & Oracle & 0.45 & 0.61 & 0.64 & 0.49 & 0.56 & 0.51 & 0.56 & 0.63 & 0.59 & 0.45 & 0.41 & 0.45 \\
\bottomrule
\end{tabular*}
\par\smallskip
\begin{minipage}{\textwidth}
\footnotesize Entries are MCSEs in percentage points.
\end{minipage}
\endgroup
\end{table}

\begin{table}[!htbp]
\centering
\caption{CF--DR nuisance-component comparisons.}\label{tab:supp_nuisance}
\begingroup
\footnotesize
\setlength{\tabcolsep}{2pt}
\renewcommand{\arraystretch}{1}
\begin{tabular*}{\textwidth}{@{\extracolsep{\fill}}ll*{9}{r}@{}}
\toprule
Scenario & Input & \multicolumn{3}{c}{MAE} & \multicolumn{3}{c}{SCB (\%)} & \multicolumn{3}{c}{AW} \\
\cmidrule(lr){3-5}\cmidrule(lr){6-8}\cmidrule(lr){9-11}
 &  & $n{=}200$ & $n{=}400$ & $n{=}800$ & $n{=}200$ & $n{=}400$ & $n{=}800$ & $n{=}200$ & $n{=}400$ & $n{=}800$ \\
\midrule
S1 & TT & 0.111 & 0.078 & 0.056 & 95.5 & 95.0 & 94.9 & 0.799 & 0.565 & 0.399 \\
 & ET & 0.122 & 0.082 & 0.057 & 95.9 & 94.6 & 95.3 & 0.920 & 0.604 & 0.412 \\
 & TE & 0.115 & 0.080 & 0.056 & 94.7 & 94.6 & 94.9 & 0.854 & 0.585 & 0.406 \\
 & EE & 0.122 & 0.081 & 0.057 & 97.3 & 95.3 & 96.1 & 0.999 & 0.627 & 0.419 \\
 & R & 0.122 & 0.082 & 0.057 & 97.8 & 95.6 & 96.2 & 1.001 & 0.627 & 0.419 \\
\addlinespace[4pt]
S2 & TT & 0.112 & 0.081 & 0.057 & 96.0 & 94.5 & 95.6 & 0.814 & 0.571 & 0.404 \\
 & ET & 0.123 & 0.085 & 0.058 & 97.0 & 94.7 & 95.5 & 0.940 & 0.617 & 0.421 \\
 & TE & 0.114 & 0.082 & 0.057 & 95.9 & 95.1 & 95.5 & 0.868 & 0.593 & 0.412 \\
 & EE & 0.122 & 0.084 & 0.057 & 97.9 & 96.0 & 96.3 & 1.012 & 0.633 & 0.423 \\
 & R & 0.122 & 0.084 & 0.057 & 98.1 & 95.8 & 96.2 & 1.015 & 0.633 & 0.423 \\
\addlinespace[4pt]
S3 & TT & 0.111 & 0.078 & 0.054 & 96.1 & 94.0 & 94.9 & 0.803 & 0.564 & 0.399 \\
 & ET & 0.122 & 0.082 & 0.056 & 96.0 & 94.6 & 95.0 & 0.917 & 0.604 & 0.412 \\
 & TE & 0.117 & 0.080 & 0.055 & 95.5 & 94.9 & 95.3 & 0.865 & 0.591 & 0.411 \\
 & EE & 0.124 & 0.082 & 0.056 & 98.2 & 96.4 & 95.6 & 1.008 & 0.635 & 0.424 \\
 & R & 0.124 & 0.082 & 0.056 & 98.3 & 96.4 & 95.6 & 1.010 & 0.635 & 0.424 \\
\addlinespace[4pt]
S4 & TT & 0.113 & 0.080 & 0.055 & 95.8 & 94.4 & 96.3 & 0.810 & 0.573 & 0.405 \\
 & ET & 0.136 & 0.087 & 0.057 & 96.0 & 95.7 & 96.5 & 1.057 & 0.646 & 0.428 \\
 & TE & 0.121 & 0.082 & 0.056 & 95.0 & 95.2 & 96.2 & 0.905 & 0.607 & 0.417 \\
 & EE & 0.139 & 0.087 & 0.057 & 98.4 & 97.3 & 97.4 & 1.230 & 0.689 & 0.441 \\
 & R & 0.139 & 0.087 & 0.057 & 98.4 & 97.4 & 97.3 & 1.233 & 0.689 & 0.441 \\
\addlinespace[4pt]
S5 & TT & 0.112 & 0.078 & 0.056 & 95.3 & 95.0 & 95.5 & 0.811 & 0.573 & 0.406 \\
 & ET & 0.136 & 0.084 & 0.058 & 97.1 & 96.7 & 96.1 & 1.068 & 0.645 & 0.428 \\
 & TE & 0.120 & 0.080 & 0.056 & 95.6 & 95.6 & 95.3 & 0.905 & 0.607 & 0.417 \\
 & EE & 0.140 & 0.084 & 0.057 & 98.0 & 98.1 & 96.5 & 1.239 & 0.688 & 0.440 \\
 & R & 0.140 & 0.084 & 0.057 & 98.1 & 98.1 & 96.5 & 1.242 & 0.688 & 0.441 \\
\addlinespace[4pt]
S6 & TT & 0.253 & 0.184 & 0.131 & 97.9 & 98.3 & 97.9 & 1.666 & 1.229 & 0.913 \\
 & ET & 0.556 & 0.272 & 0.158 & 98.4 & 97.3 & 96.6 & 4.488 & 2.064 & 1.177 \\
 & TE & 0.309 & 0.204 & 0.139 & 91.1 & 91.8 & 93.4 & 2.249 & 1.455 & 0.998 \\
 & EE & 0.666 & 0.304 & 0.165 & 97.5 & 95.0 & 95.2 & 6.377 & 2.549 & 1.291 \\
 & R & 0.667 & 0.304 & 0.165 & 97.7 & 95.3 & 95.3 & 6.386 & 2.551 & 1.291 \\
\bottomrule
\end{tabular*}
\par\smallskip
\begin{minipage}{\textwidth}
\footnotesize SCB is coverage (\%); AW is mean simultaneous band width.
\end{minipage}
\endgroup
\end{table}

\begin{table}[!htbp]
\centering
\caption{Paired corrected and uncorrected CF--DR bands.}\label{tab:supp_correction}
\begingroup
\footnotesize
\setlength{\tabcolsep}{2pt}
\renewcommand{\arraystretch}{1}
\begin{tabular*}{\textwidth}{@{\extracolsep{\fill}}ll*{5}{r}@{}}
\toprule
Scenario & $n$ & $C_+$ & $C_-$ & $\Delta C$ & MCSE & $W_+/W_-$ \\
\midrule
S1 & 200 & 97.8 & 95.6 & 2.2 & 0.46 & 1.093 \\
 & 400 & 95.6 & 94.2 & 1.4 & 0.37 & 1.040 \\
 & 800 & 96.2 & 95.2 & 1.0 & 0.31 & 1.019 \\
\addlinespace[4pt]
S2 & 200 & 98.1 & 96.9 & 1.2 & 0.37 & 1.084 \\
 & 400 & 95.8 & 94.3 & 1.5 & 0.38 & 1.029 \\
 & 800 & 96.2 & 96.0 & 0.2 & 0.20 & 1.005 \\
\addlinespace[4pt]
S3 & 200 & 98.3 & 95.6 & 2.7 & 0.51 & 1.092 \\
 & 400 & 96.4 & 95.4 & 1.0 & 0.31 & 1.041 \\
 & 800 & 95.6 & 94.9 & 0.7 & 0.26 & 1.018 \\
\addlinespace[4pt]
S4 & 200 & 98.4 & 95.1 & 3.3 & 0.57 & 1.180 \\
 & 400 & 97.4 & 94.8 & 2.6 & 0.50 & 1.074 \\
 & 800 & 97.3 & 96.5 & 0.8 & 0.28 & 1.033 \\
\addlinespace[4pt]
S5 & 200 & 98.1 & 96.8 & 1.3 & 0.36 & 1.180 \\
 & 400 & 98.1 & 96.0 & 2.1 & 0.45 & 1.073 \\
 & 800 & 96.5 & 95.3 & 1.2 & 0.34 & 1.031 \\
\addlinespace[4pt]
S6 & 200 & 97.7 & 80.1 & 17.6 & 1.20 & 1.652 \\
 & 400 & 95.3 & 84.8 & 10.5 & 0.98 & 1.337 \\
 & 800 & 95.3 & 89.7 & 5.6 & 0.74 & 1.156 \\
\bottomrule
\end{tabular*}
\par\smallskip
\begin{minipage}{\textwidth}
\footnotesize $C_+$ and $C_-$ are corrected and uncorrected coverage (\%), with $\Delta C=C_+-C_-$. $W_+/W_-$ is the corresponding mean band-width ratio.
\end{minipage}
\endgroup
\end{table}

\begin{table}[!htbp]
\centering
\caption{Effect of estimating curves from noisy measurements.}\label{tab:supp_latent}
\begingroup
\footnotesize
\setlength{\tabcolsep}{2pt}
\renewcommand{\arraystretch}{1}
\begin{tabular*}{\textwidth}{@{\extracolsep{\fill}}ll*{9}{r}@{}}
\toprule
Scenario & Method & \multicolumn{3}{c}{$\Delta$MAE} & \multicolumn{3}{c}{$\Delta C$} & \multicolumn{3}{c}{MCSE} \\
\cmidrule(lr){3-5}\cmidrule(lr){6-8}\cmidrule(lr){9-11}
 &  & $n{=}200$ & $n{=}400$ & $n{=}800$ & $n{=}200$ & $n{=}400$ & $n{=}800$ & $n{=}200$ & $n{=}400$ & $n{=}800$ \\
\midrule
S1 & CF--DR & 0.0001 & 0.0001 & 0.0000 & 0.5 & 0.3 & 0.1 & 0.26 & 0.17 & 0.10 \\
 & OR & 0.0002 & 0.0001 & 0.0000 & 0.1 & -0.1 & 0.3 & 0.30 & 0.17 & 0.17 \\
 & IPW & 0.0001 & 0.0001 & 0.0000 & 0.2 & -0.2 & 0.0 & 0.24 & 0.20 & 0.14 \\
 & Naive DiD & 0.0002 & 0.0000 & 0.0000 & 0.2 & -0.2 & 0.0 & 0.37 & 0.24 & 0.25 \\
\addlinespace[4pt]
S2 & CF--DR & 0.0003 & 0.0000 & 0.0000 & 0.2 & -0.2 & -0.1 & 0.14 & 0.14 & 0.10 \\
 & OR & 0.0003 & 0.0000 & 0.0000 & -0.1 & 0.0 & 0.2 & 0.26 & 0.14 & 0.14 \\
 & IPW & 0.0002 & 0.0000 & 0.0000 & -0.3 & 0.2 & -0.1 & 0.26 & 0.24 & 0.10 \\
 & Naive DiD & 0.0002 & 0.0000 & 0.0000 & -0.4 & -0.3 & 0.2 & 0.37 & 0.26 & 0.20 \\
\addlinespace[4pt]
S3 & CF--DR & 0.0002 & 0.0000 & 0.0000 & 0.1 & 0.0 & 0.0 & 0.17 & 0.00 & 0.00 \\
 & OR & 0.0002 & 0.0000 & 0.0000 & -0.2 & -0.3 & -0.1 & 0.40 & 0.26 & 0.10 \\
 & IPW & 0.0002 & 0.0000 & 0.0000 & -0.4 & 0.0 & -0.1 & 0.20 & 0.14 & 0.10 \\
 & Naive DiD & 0.0002 & 0.0000 & 0.0000 & 0.2 & 0.1 & 0.1 & 0.35 & 0.22 & 0.10 \\
\addlinespace[4pt]
S4 & CF--DR & 0.0003 & 0.0000 & 0.0000 & 0.0 & 0.1 & -0.1 & 0.20 & 0.22 & 0.10 \\
 & OR & 0.0002 & 0.0000 & 0.0000 & -0.3 & 0.0 & 0.0 & 0.30 & 0.20 & 0.00 \\
 & IPW & 0.0003 & 0.0000 & 0.0000 & 0.1 & -0.4 & -0.1 & 0.22 & 0.20 & 0.10 \\
 & Naive DiD & 0.0001 & 0.0000 & 0.0000 & 0.3 & 0.4 & -0.2 & 0.44 & 0.35 & 0.24 \\
\addlinespace[4pt]
S5 & CF--DR & 0.0002 & 0.0000 & 0.0000 & 0.1 & 0.0 & 0.0 & 0.10 & 0.14 & 0.00 \\
 & OR & 0.0002 & 0.0000 & 0.0000 & -0.3 & 0.4 & -0.1 & 0.33 & 0.24 & 0.10 \\
 & IPW & 0.0002 & 0.0000 & 0.0000 & 0.2 & 0.1 & -0.1 & 0.14 & 0.17 & 0.10 \\
 & Naive DiD & 0.0000 & 0.0000 & 0.0000 & -0.1 & -0.3 & 0.1 & 0.22 & 0.22 & 0.17 \\
\addlinespace[4pt]
S6 & CF--DR & 0.0003 & 0.0003 & 0.0000 & 0.2 & 0.3 & 0.1 & 0.24 & 0.17 & 0.10 \\
 & OR & 0.0003 & 0.0001 & 0.0000 & -0.3 & -0.3 & -0.1 & 0.41 & 0.26 & 0.17 \\
 & IPW & -0.0004 & 0.0001 & 0.0000 & 0.4 & 0.2 & 0.3 & 0.28 & 0.20 & 0.17 \\
 & Naive DiD & 0.0001 & 0.0000 & 0.0000 & 0.2 & 0.3 & 0.0 & 0.40 & 0.26 & 0.00 \\
\bottomrule
\end{tabular*}
\par\smallskip
\begin{minipage}{\textwidth}
\footnotesize $\Delta C$ and its paired MCSE are in percentage points.
\end{minipage}
\endgroup
\end{table}

\begin{table}[!htbp]
\centering
\caption{CF--DR sensitivity to observation density and measurement noise.}\label{tab:supp_observation}
\begingroup
\footnotesize
\setlength{\tabcolsep}{2pt}
\renewcommand{\arraystretch}{1}
\begin{tabular*}{\textwidth}{@{\extracolsep{\fill}}ll*{9}{r}@{}}
\toprule
Scenario & Setting & \multicolumn{3}{c}{MAE} & \multicolumn{3}{c}{SCB (\%)} & \multicolumn{3}{c}{$\sqrt n\,r_{H^1}$} \\
\cmidrule(lr){3-5}\cmidrule(lr){6-8}\cmidrule(lr){9-11}
 &  & $n{=}200$ & $n{=}400$ & $n{=}800$ & $n{=}200$ & $n{=}400$ & $n{=}800$ & $n{=}200$ & $n{=}400$ & $n{=}800$ \\
\midrule
S4 & Base & 0.139 & 0.087 & 0.057 & 98.4 & 97.4 & 97.3 & 22.30 & 16.37 & 11.98 \\
 & D-- & 0.140 & 0.087 & 0.057 & 98.7 & 97.4 & 97.1 & 44.65 & 32.74 & 24.00 \\
 & D+ & 0.139 & 0.087 & 0.057 & 98.3 & 97.3 & 97.3 & 11.11 & 8.19 & 6.00 \\
 & N-- & 0.139 & 0.087 & 0.057 & 98.2 & 97.3 & 97.4 & 11.15 & 8.19 & 5.99 \\
 & N+ & 0.140 & 0.087 & 0.057 & 98.6 & 97.3 & 97.3 & 44.59 & 32.75 & 23.97 \\
\addlinespace[4pt]
S6 & Base & 0.667 & 0.304 & 0.165 & 97.7 & 95.3 & 95.3 & 22.39 & 16.34 & 12.01 \\
 & D-- & 0.671 & 0.304 & 0.165 & 97.2 & 95.0 & 95.5 & 44.41 & 32.61 & 24.03 \\
 & D+ & 0.666 & 0.304 & 0.165 & 97.7 & 95.3 & 95.2 & 11.12 & 8.22 & 5.99 \\
 & N-- & 0.666 & 0.304 & 0.165 & 97.7 & 95.0 & 95.3 & 11.19 & 8.17 & 6.01 \\
 & N+ & 0.669 & 0.304 & 0.165 & 97.5 & 95.1 & 95.5 & 44.77 & 32.68 & 24.02 \\
\bottomrule
\end{tabular*}
\par\smallskip
\begin{minipage}{\textwidth}
\footnotesize Settings $(m,\sigma_\varepsilon)$: Base $(m_n,0.20)$; D-- $(m_n/4,0.20)$; D+ $(4m_n,0.20)$; N-- $(m_n,0.10)$; N+ $(m_n,0.40)$.
\end{minipage}
\endgroup
\end{table}

\begin{table}[!htbp]
\centering
\caption{CF--DR sensitivity to a doubled-resolution inference grid.}\label{tab:supp_grid}
\begingroup
\small
\setlength{\tabcolsep}{2pt}
\renewcommand{\arraystretch}{1}
\begin{tabular*}{\textwidth}{@{\extracolsep{\fill}}ll*{5}{r}@{}}
\toprule
Scenario & $n$ & $C_{2n}$ & $\Delta C$ & MCSE & $W_{2n}$ & $\Delta W$ \\
\midrule
S1 & 200 & 97.8 & 0.0 & 0.00 & 1.002 & 0.0001 \\
 & 400 & 95.6 & 0.0 & 0.00 & 0.627 & 0.0000 \\
 & 800 & 96.2 & 0.0 & 0.00 & 0.419 & 0.0000 \\
\addlinespace[4pt]
S2 & 200 & 98.1 & 0.0 & 0.00 & 1.015 & 0.0001 \\
 & 400 & 95.8 & 0.0 & 0.00 & 0.633 & 0.0000 \\
 & 800 & 96.2 & 0.0 & 0.00 & 0.423 & 0.0000 \\
\addlinespace[4pt]
S3 & 200 & 98.3 & 0.0 & 0.00 & 1.010 & 0.0001 \\
 & 400 & 96.4 & 0.0 & 0.00 & 0.635 & 0.0000 \\
 & 800 & 95.6 & 0.0 & 0.00 & 0.424 & 0.0000 \\
\addlinespace[4pt]
S4 & 200 & 98.4 & 0.0 & 0.00 & 1.233 & 0.0002 \\
 & 400 & 97.4 & 0.0 & 0.00 & 0.689 & 0.0000 \\
 & 800 & 97.3 & 0.0 & 0.00 & 0.441 & 0.0000 \\
\addlinespace[4pt]
S5 & 200 & 98.1 & 0.0 & 0.00 & 1.242 & 0.0002 \\
 & 400 & 98.1 & 0.0 & 0.00 & 0.688 & 0.0000 \\
 & 800 & 96.5 & 0.0 & 0.00 & 0.441 & 0.0000 \\
\addlinespace[4pt]
S6 & 200 & 97.7 & 0.0 & 0.00 & 6.391 & 0.0047 \\
 & 400 & 95.3 & 0.0 & 0.00 & 2.551 & 0.0003 \\
 & 800 & 95.3 & 0.0 & 0.00 & 1.291 & 0.0000 \\
\bottomrule
\end{tabular*}
\par\smallskip
\begin{minipage}{\textwidth}
\footnotesize $C_{2n}$ and $W_{2n}$ are fine-grid coverage (\%) and mean band width. Differences are fine-grid minus principal-grid results.
\end{minipage}
\endgroup
\end{table}

\begin{table}[!htbp]
\centering
\caption{Basis truncation: full-target error and coverage for full and projected targets.}\label{tab:supp_truncation}
\begingroup
\footnotesize
\setlength{\tabcolsep}{2pt}
\renewcommand{\arraystretch}{1}
\begin{tabular*}{\textwidth}{@{\extracolsep{\fill}}ll*{9}{r}@{}}
\toprule
Scenario / $L$ & Method & \multicolumn{3}{c}{MAE (full)} & \multicolumn{3}{c}{$C_{\rm full}$} & \multicolumn{3}{c}{$C_{\rm proj}$} \\
\cmidrule(lr){3-5}\cmidrule(lr){6-8}\cmidrule(lr){9-11}
 &  & $n{=}200$ & $n{=}400$ & $n{=}800$ & $n{=}200$ & $n{=}400$ & $n{=}800$ & $n{=}200$ & $n{=}400$ & $n{=}800$ \\
\midrule
S4 / 13 & CF--DR & 0.139 & 0.087 & 0.057 & 98.4 & 97.4 & 97.3 & 98.4 & 97.4 & 97.3 \\
 & OR & 0.120 & 0.082 & 0.056 & 88.8 & 92.9 & 95.2 & 88.8 & 92.9 & 95.2 \\
 & IPW & 0.139 & 0.087 & 0.058 & 96.4 & 95.6 & 96.4 & 96.4 & 95.6 & 96.4 \\
 & Naive DiD & 0.122 & 0.098 & 0.078 & 87.6 & 77.6 & 67.6 & 87.6 & 77.6 & 67.6 \\
\addlinespace[4pt]
S4 / 9 & CF--DR & 0.146 & 0.099 & 0.076 & 95.9 & 86.1 & 60.1 & 98.5 & 96.9 & 97.2 \\
 & OR & 0.128 & 0.095 & 0.075 & 77.7 & 67.5 & 44.2 & 89.8 & 92.6 & 95.1 \\
 & IPW & 0.146 & 0.100 & 0.076 & 90.7 & 79.9 & 53.6 & 97.0 & 96.0 & 96.3 \\
 & Naive DiD & 0.130 & 0.108 & 0.091 & 66.9 & 40.7 & 10.5 & 87.7 & 77.8 & 67.2 \\
\addlinespace[4pt]
S6 / 13 & CF--DR & 0.667 & 0.304 & 0.165 & 97.7 & 95.3 & 95.3 & 97.7 & 95.3 & 95.3 \\
 & OR & 0.192 & 0.134 & 0.091 & 86.3 & 91.0 & 93.8 & 86.3 & 91.0 & 93.8 \\
 & IPW & 0.586 & 0.286 & 0.165 & 95.2 & 94.1 & 94.6 & 95.2 & 94.1 & 94.6 \\
 & Naive DiD & 0.182 & 0.171 & 0.163 & 49.9 & 17.3 & 1.3 & 49.9 & 17.3 & 1.3 \\
\addlinespace[4pt]
S6 / 9 & CF--DR & 0.659 & 0.304 & 0.171 & 96.1 & 93.6 & 90.0 & 97.4 & 96.1 & 95.8 \\
 & OR & 0.196 & 0.141 & 0.103 & 82.2 & 81.3 & 76.6 & 88.1 & 91.0 & 93.4 \\
 & IPW & 0.581 & 0.287 & 0.171 & 93.0 & 88.5 & 86.7 & 95.0 & 93.5 & 95.4 \\
 & Naive DiD & 0.187 & 0.176 & 0.167 & 28.6 & 2.4 & 0.0 & 50.7 & 16.6 & 1.1 \\
\bottomrule
\end{tabular*}
\par\smallskip
\begin{minipage}{\textwidth}
\footnotesize MAE uses the full target $\tau_0$. $C_{\rm full}$ and $C_{\rm proj}$ are coverage percentages for $\tau_0$ and $\Pi_L\tau_0$, respectively.
\end{minipage}
\endgroup
\end{table}

\begin{table}[!htbp]
\centering
\caption{Propensity-clipping sensitivity: estimation error and weight summaries.}\label{tab:supp_clipping}
\begingroup
\footnotesize
\setlength{\tabcolsep}{2pt}
\renewcommand{\arraystretch}{1}
\noindent Panel A: point estimation\par\smallskip
\begin{tabular*}{\textwidth}{@{\extracolsep{\fill}}ll*{3}{r}@{}}
\toprule
Scenario / $c$ & Method & \multicolumn{3}{c}{MAE} \\
\cmidrule(lr){3-5}
 &  & $n{=}200$ & $n{=}400$ & $n{=}800$ \\
\midrule
S4 / 0.001 & CF--DR & 0.139 & 0.087 & 0.057 \\
 & IPW & 0.139 & 0.087 & 0.058 \\
S4 / 0.005 & CF--DR & 0.139 & 0.087 & 0.057 \\
 & IPW & 0.139 & 0.087 & 0.058 \\
S4 / 0.020 & CF--DR & 0.139 & 0.087 & 0.057 \\
 & IPW & 0.139 & 0.087 & 0.058 \\
\addlinespace[4pt]
S6 / 0.001 & CF--DR & 0.667 & 0.304 & 0.165 \\
 & IPW & 0.586 & 0.286 & 0.165 \\
S6 / 0.005 & CF--DR & 0.564 & 0.283 & 0.165 \\
 & IPW & 0.488 & 0.270 & 0.165 \\
S6 / 0.020 & CF--DR & 0.385 & 0.232 & 0.148 \\
 & IPW & 0.333 & 0.221 & 0.146 \\
\bottomrule
\end{tabular*}
\par\medskip
\noindent Panel B: weight summaries\par\smallskip
\begin{tabular*}{\textwidth}{@{\extracolsep{\fill}}ll*{9}{r}@{}}
\toprule
Scenario & $c$ & \multicolumn{3}{c}{Clip (\%)} & \multicolumn{3}{c}{ESS$/n_0$} & \multicolumn{3}{c}{Max odds} \\
\cmidrule(lr){3-5}\cmidrule(lr){6-8}\cmidrule(lr){9-11}
 &  & $n{=}200$ & $n{=}400$ & $n{=}800$ & $n{=}200$ & $n{=}400$ & $n{=}800$ & $n{=}200$ & $n{=}400$ & $n{=}800$ \\
\midrule
S4 & 0.001 & 0.000 & 0.000 & 0.000 & 0.594 & 0.663 & 0.710 & 6.2 & 4.9 & 4.2 \\
 & 0.005 & 0.002 & 0.000 & 0.000 & 0.594 & 0.663 & 0.710 & 6.2 & 4.9 & 4.2 \\
 & 0.020 & 0.012 & 0.001 & 0.000 & 0.594 & 0.663 & 0.710 & 6.2 & 4.9 & 4.2 \\
\addlinespace[4pt]
S6 & 0.001 & 0.198 & 0.017 & 0.000 & 0.096 & 0.083 & 0.078 & 63.9 & 55.3 & 52.6 \\
 & 0.005 & 1.607 & 0.425 & 0.076 & 0.097 & 0.083 & 0.078 & 51.7 & 51.1 & 52.5 \\
 & 0.020 & 10.048 & 7.464 & 6.120 & 0.104 & 0.090 & 0.084 & 31.7 & 36.9 & 40.8 \\
\bottomrule
\end{tabular*}
\endgroup
\end{table}

\FloatBarrier
\subsection*{Additional findings}
In S6, the center-point CF--DR SEratio is 0.784, 0.777, and 1.054 at $n=200,400,800$, respectively
(Table~\ref{tab:supp_bias_se_other}; Figure~\ref{fig:supp_s6_bias_se}). At the two smaller sample sizes,
mean estimated standard errors are below the Monte Carlo standard deviations, despite coverage above 95\%.

In S4, mean $\sqrt n\,r_{n,H^1}$ decreases from 22.30 to 11.98; this measure includes the maximum over
units and the derivative norm. Using true curves or varying density and noise produces small changes; doubling the grid changes mean widths
by less than 0.08\% without changing coverage. With nine basis functions in S4, full-target coverage
falls from 95.9\% to 60.1\%, while projected-target coverage remains above 95\%
(Tables~\ref{tab:supp_latent}--\ref{tab:supp_truncation}).

In the nonparametric experiment (Table~\ref{tab:supp_nonparametric}; Figure~\ref{fig:supp_np}), five
replications at $n=200$ have an empty control bin and use the prespecified fallback.
Maximum control odds remain at or below 49 in these runs.

\begin{table}[!htbp]
\centering\small
\caption{Nonparametric bin-learner experiment: estimation and CF--DR inference.}
\label{tab:supp_nonparametric}
\begingroup
\setlength{\tabcolsep}{4pt}
\renewcommand{\arraystretch}{1}
\noindent Panel A: estimation accuracy\par\smallskip
\begin{tabular*}{\textwidth}{@{\extracolsep{\fill}}llrrrr@{}}
\toprule
$n$ & Method & Bias & MAE & $\overline{\rm ISE}$ & SupErr \\
\midrule
200 & CF--DR & 0.010 & 0.146 & 0.038 & 0.346 \\
 & OR & 0.008 & 0.127 & 0.026 & 0.301 \\
 & IPW & -0.014 & 0.140 & 0.032 & 0.331 \\
 & Naive DiD & 0.137 & 0.193 & 0.054 & 0.407 \\
 & Oracle & 0.001 & 0.120 & 0.023 & 0.283 \\
\addlinespace[3pt]
400 & CF--DR & 0.002 & 0.097 & 0.015 & 0.229 \\
 & OR & 0.003 & 0.090 & 0.013 & 0.213 \\
 & IPW & -0.013 & 0.097 & 0.015 & 0.228 \\
 & Naive DiD & 0.134 & 0.177 & 0.044 & 0.362 \\
 & Oracle & -0.001 & 0.088 & 0.012 & 0.206 \\
\addlinespace[3pt]
800 & CF--DR & 0.002 & 0.064 & 0.006 & 0.152 \\
 & OR & 0.002 & 0.062 & 0.006 & 0.148 \\
 & IPW & -0.007 & 0.064 & 0.007 & 0.153 \\
 & Naive DiD & 0.136 & 0.175 & 0.040 & 0.337 \\
 & Oracle & -0.001 & 0.060 & 0.006 & 0.144 \\
\bottomrule
\end{tabular*}
\par\medskip
\noindent Panel B: CF--DR inference\par\smallskip
\begin{tabular*}{\textwidth}{@{\extracolsep{\fill}}l*{7}{r}@{}}
\toprule
$n$ & $C_{.25}$ & $C_{.50}$ & $C_{.75}$ & SCB (\%) & AW & SBias$_{.50}$ & SEratio$_{.50}$ \\
\midrule
200 & 95.5 & 95.5 & 96.8 & 95.2 & 1.099 & 0.040 & 0.981 \\
400 & 96.4 & 94.8 & 95.4 & 95.7 & 0.719 & 0.018 & 0.977 \\
800 & 96.9 & 96.2 & 95.6 & 94.8 & 0.483 & 0.026 & 1.076 \\
\bottomrule
\end{tabular*}
\par\smallskip
\begin{minipage}{\textwidth}
\footnotesize $C_t$ is pointwise coverage at $t$; all coverage entries are percentages.
AW is mean simultaneous band width.
\end{minipage}
\endgroup
\end{table}

\FloatBarrier
\begin{figure}[p]
\centering
\includegraphics[width=\textwidth]{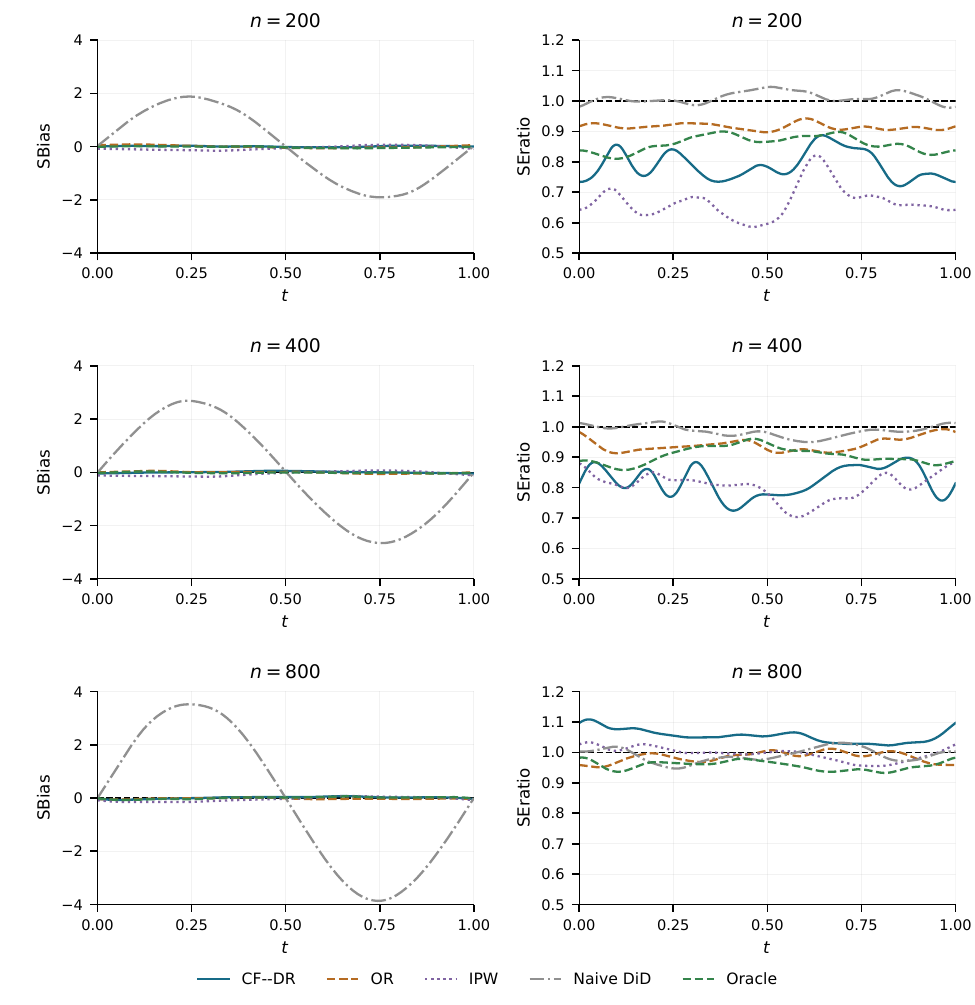}
\caption{Standardized bias and standard-error ratios under limited overlap (S6). Reference lines mark SBias $=0$ and SEratio $=1$.}
\label{fig:supp_s6_bias_se}
\end{figure}

\begin{figure}[p]
\centering
\includegraphics[width=\textwidth]{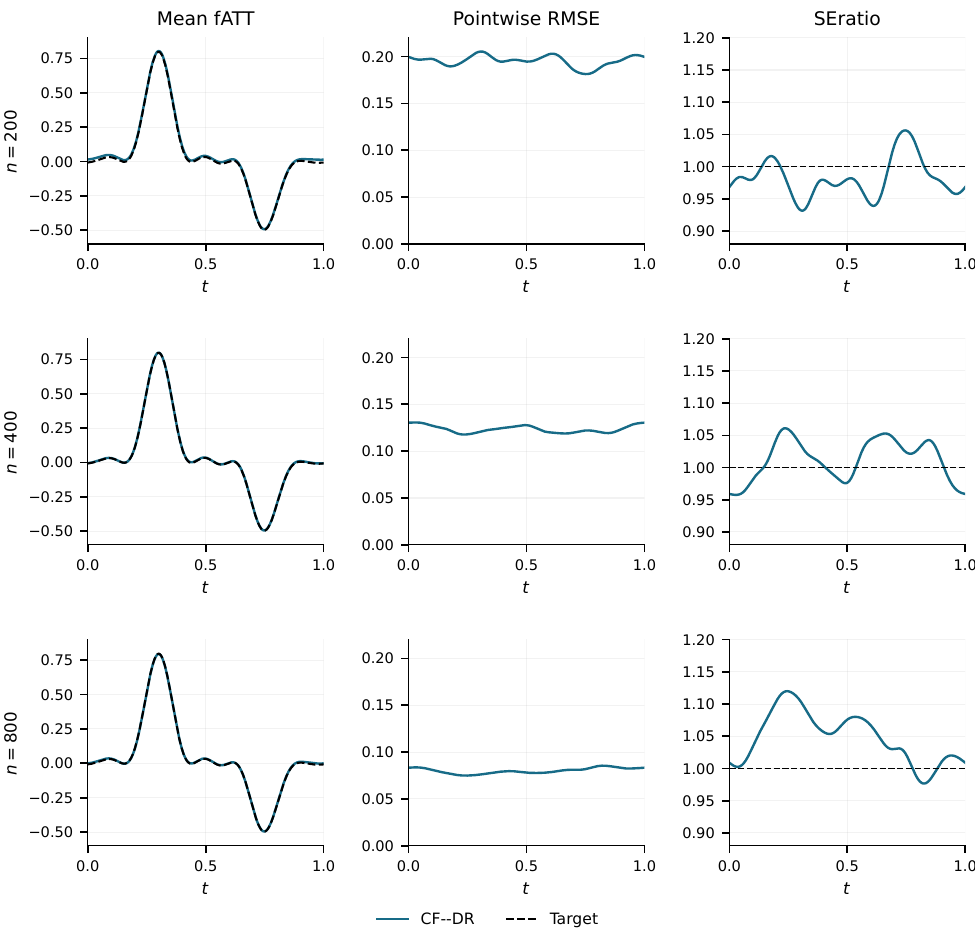}
\caption{Nonparametric CF--DR results: mean estimated curves and target, pointwise RMSE, and standard-error ratios.}
\label{fig:supp_np}
\end{figure}

\FloatBarrier
\section*{Application: implementation and sensitivity analyses}
\subsection*{Data preparation and analysis population}
The analysis in main-paper Section~6 uses August 2012 covariates and compares September--October 2012 with the same months in 2013. Dynamic-tariff participation ($D=1$) was voluntary; comparison households have source label \texttt{Std}. The study collected readings between November 2011 and February 2014 and reports 5567 recruits; the public archive contains 5,566 identifiers.

Recorded half-hour labels are retained, excluding October 28, 2012 and October 27, 2013 daylight-saving transitions. A complete day has 48 distinct, finite, nonnegative readings. Identical duplicates are collapsed; conflicting duplicates, off-grid timestamps, or missing readings invalidate the day. Profiles use complete days without interpolation or zero imputation. Eligibility requires at least 80\% complete eligible days and two complete dates for each of the seven days of the week in every study month.

For retained dates $\mathcal D_{imw}$ in month $m$ and day of the week $w$, define
\[
 \overline Y_{im}(h)=\frac17\sum_{w=1}^7
 \frac1{|\mathcal D_{imw}|}\sum_{d\in\mathcal D_{imw}}Y_{id}(h),
 \qquad h\in\mathcal G_{48}:=\{0,\tfrac12,\ldots,23\tfrac12\}.
\]
Averaging the two monthly profiles equally gives $Y_{i0}$ and $Y_{i1}$, with $\Delta Y_i=Y_{i1}-Y_{i0}$. These weights standardize calendar composition. For eligibility indicator $R$, the target is $\tau_{\rm E}(h)=\E\{Y_1^1(h)-Y_1^0(h)\mid D=1,R=1\}$. In addition to consistency, no anticipation, and no interference, identification requires
\begin{gather*}
 \E(Y_1^0-Y_0^0\mid X,D=1,R=1)=\E(Y_1^0-Y_0^0\mid X,D=0,R=1),\\
 c_{\rm E}\leq P(D=1\mid X,R=1)\leq1-c_{\rm E},\qquad c_{\rm E}>0.
\end{gather*}
These conditions and all nuisance functions concern the selected population and its observation process.

For $B_i(h)=\overline Y_{i,\mathrm{Aug\,2012}}(h)$, the six covariates summarize its mean, first two Fourier harmonics, and weekend--weekday contrast. The harmonic summaries are
\[
 b_{i,\sin k}=\frac1{48}\sum_{h\in\mathcal G_{48}}B_i(h)\sin(2\pi kh/24),\qquad
 b_{i,\cos k}=\frac1{48}\sum_{h\in\mathcal G_{48}}B_i(h)\cos(2\pi kh/24),\quad k=1,2.
\]
The mean is $48^{-1}\sum_hB_i(h)$; the weekend contrast subtracts the Monday--Friday average from the Saturday--Sunday average. In order, the summaries correspond to \texttt{baseline\_mean}, \texttt{baseline\_sin1}, \texttt{baseline\_cos1}, \texttt{baseline\_sin2}, \texttt{baseline\_cos2}, and \texttt{baseline\_weekend}. Each enters as $X_{ij}=\tanh(b_{ij}/0.5)$, with scale 0.5 kWh.

\subsection*{Estimation and comparisons}
CF--DR, OR, and IPW share five household folds assigned independently of the data, the six covariates, Firth logistic propensities clipped to $[0.001,0.999]$, and group-specific least-squares change regressions. Each household's entire profile remains in one fold. Both models include intercepts. Naive DiD is unadjusted. The primary fit requires at least 100 eligible households per group. Control ESS uses the formula above with out-of-fold odds weights.

Centered household influence vectors from Propositions~\ref{prop:supp_working}--\ref{prop:supp_foldwise} provide covariance estimate $\widehat\Sigma$. Corollary~\ref{cor:supp_fixed_grid}, with 5000 Gaussian multiplier draws, gives 95\% joint intervals over the 48 recorded coordinates. Sampling inference assumes independent households conditional on realized calendar-level conditions and the tariff schedule. For a set of half-hours $A$,
\[
 \widehat\theta_A=\sum_{h\in A}\widehat\tau(h),\qquad
 \widehat{\operatorname{se}}(\widehat\theta_A)
 =\{\bm1_A^\top\widehat\Sigma\bm1_A/n\}^{1/2}.
\]
Full-day, morning $[7,10)$, and evening $[17,23)$ contrasts use marginal 95\% normal intervals. Each reading measures half-hour energy in kWh, so the contrast is a sum over $A$.

The October-minus-September 2012 pretreatment comparison retains August covariates, households, and folds and refits outcome regressions. Expressing $\widehat\tau_{\rm pre}$ in Wh per half-hour, define $\mathrm{RMS}=\{48^{-1}\sum_h\widehat\tau_{\rm pre}(h)^2\}^{1/2}$ and the maximum absolute contrast $\max_h|\widehat\tau_{\rm pre}(h)|$. These summarize departures from zero in the pretreatment contrast; the joint test evaluates zero contrast at all 48 coordinates. They concern pretreatment differences, while bias and RMSE of the treatment-effect estimates cannot be measured without the true effects. Prespecified checks use September alone, October alone, an expanded design adding $X_1^2,X_1X_2,X_1X_3$ to both nuisance models, and a 90\% completeness threshold. Only the last changes the eligible population.

\subsection*{Pretreatment comparison and prespecified sensitivity analyses}
The analysis retains 79.6\% of dynamic-tariff and 78.7\% of comparison households. OR's daily standard error is 2.8\% smaller than CF--DR's (main-paper Tables~5--6). CF--DR's maximum absolute pretreatment contrast is 21.7\% below Naive DiD's; IPW has the smallest RMS and maximum contrasts (Table~\ref{tab:lcl_diagnostics}).

\begin{table}[!htbp]
\centering\small
\caption{Pretreatment contrasts for October minus September 2012.}
\label{tab:lcl_diagnostics}
\begin{tabular}{lrrr}
\toprule
Method & RMS (Wh) & Max. absolute (Wh) & Joint $p$ \\
\midrule
CF--DR & 4.481 & 8.609 & 0.123 \\
OR & 4.569 & 9.233 & 0.063 \\
IPW & 4.328 & 8.044 & 0.084 \\
Naive DiD & 4.604 & 10.994 & 0.001 \\
\bottomrule
\end{tabular}
\par\smallskip
\begin{minipage}{\textwidth}
\footnotesize Contrasts are in Wh per half-hour. Naive DiD's $p=0.0014$ is rounded to 0.001.
\end{minipage}
\end{table}
All prespecified daily estimates are negative, but the September interval includes zero (Table~\ref{tab:lcl_sensitivity}; Figure~\ref{fig:lcl_diagnostics}). Expanded nuisance models and stricter completeness give similar estimates; the latter retains 4345 households.

\begin{table}[!htbp]
\centering\small
\caption{CF--DR prespecified sensitivity analyses and exploratory checks.}
\label{tab:lcl_sensitivity}
\begin{tabular}{lrrr}
\toprule
Specification & $n$ & Energy contrast & 95\% CI \\
\midrule
\multicolumn{4}{l}{Panel A: primary and prespecified analyses} \\
\addlinespace
Main & 4389 & -0.249 & [-0.413, -0.085] \\
September & 4389 & -0.160 & [-0.331, 0.012] \\
October & 4389 & -0.339 & [-0.554, -0.124] \\
Expanded features & 4389 & -0.267 & [-0.433, -0.101] \\
90\% completeness & 4345 & -0.268 & [-0.435, -0.102] \\
\addlinespace[4pt]
\multicolumn{4}{l}{Panel B: exploratory checks (nominal intervals)} \\
\addlinespace
Night-feature adjustment & 4389 & -0.254 & [-0.420, -0.088] \\
Baseline night share $\leq 50\%$ & 4358 & -0.248 & [-0.412, -0.083] \\
47-slot energy (omit 00:00) & 4389 & -0.234 & [-0.396, -0.073] \\
\bottomrule
\end{tabular}
\par\smallskip
\begin{minipage}{\textwidth}
\footnotesize Contrasts and marginal 95\% intervals are in kWh per standardized day. The 47-slot contrast omits midnight.
\end{minipage}
\end{table}

\begin{figure}[!htbp]
\centering
\includegraphics[width=0.95\textwidth]{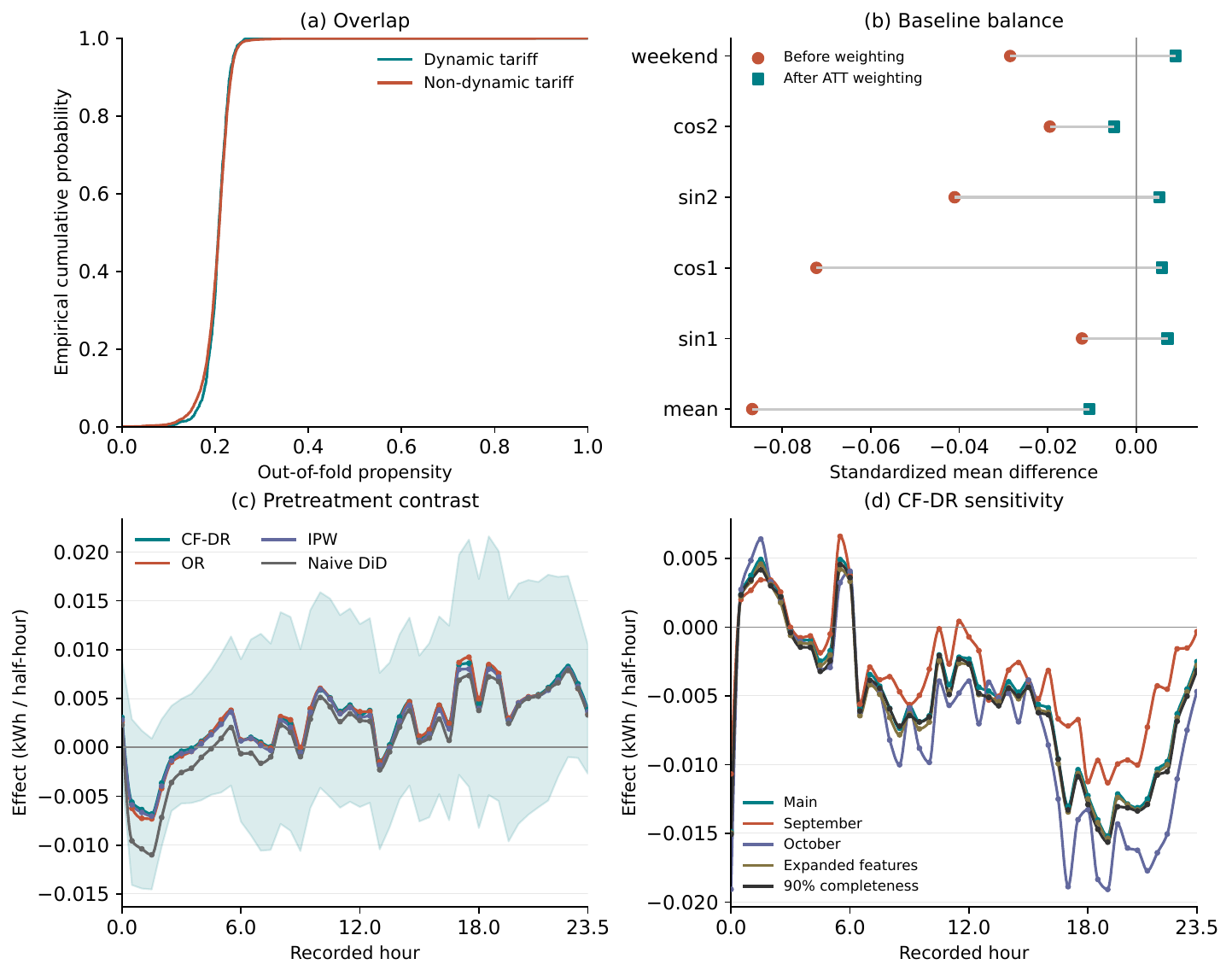}
\caption{Household analysis: (a) out-of-fold propensity distributions; (b) baseline balance before and after ATT weighting; (c) pretreatment contrasts with CF--DR joint intervals; (d) CF--DR estimates under the prespecified sensitivity specifications.}
\label{fig:lcl_diagnostics}
\end{figure}

\FloatBarrier
\subsection*{Exploratory sensitivity analyses}
The following analyses were added after inspecting the initial results and use the same specifications across methods. Their intervals and $p$-values do not account for this selection.

For the August profile, define night-energy share and boundary contrast
\[
 q_i=\frac{\sum_{h\in[0,7)}B_i(h)}{\sum_{h\in\mathcal G_{48}}B_i(h)},\qquad
 j_i=B_i(0)-\{B_i(\tfrac12)+B_i(23\tfrac12)\}/2,
\]
with $q_i=0$ for zero baseline energy. The fields are \texttt{baseline\_\allowbreak night\_\allowbreak share} and \texttt{baseline\_\allowbreak boundary\_\allowbreak kwh}. Night-feature adjustment adds $q_i$ and $\tanh(j_i/0.5)$ to both nuisance models, retaining the primary cohort. The separate restriction $q_i\leq0.5$ uses the original six features and excludes two treated and 29 comparison households, leaving 892 and 3466. Original fold labels are retained (sizes 873, 873, 870, 874, 868). Both analyses give negative estimates (Table~\ref{tab:lcl_sensitivity}; Figure~\ref{fig:supp_lcl_robustness}); the restriction gives Naive DiD pretreatment $p=0.065$.

For each study month, define $a_{im}=\overline Y_{im}(0)-\{\overline Y_{im}(\tfrac12)+\overline Y_{im}(23\tfrac12)\}/2$. The rule $\max_m|a_{im}|>1$ kWh identifies 31 comparison households and no treated households; these records remain in the analysis. The cause of the midnight anomalies is unresolved without verified tariff and meter-fault histories. The 47-slot contrast omits midnight from the fitted sum and covariance quadratic form. It estimates energy over the remaining 47 half-hours and remains negative.

For $a=\mathrm{CF\text{--}DR}$, center \eqref{eq:supp_plugin} as $v_i=\widehat\phi_{a,i}-\bar{\widehat\phi}_a$ and set $u_i=\sum_{h\in\mathcal G_{48}}v_i(h)$. The ten largest variance shares $u_i^2/\sum_ju_j^2$ total 27.8\%. Omitting each of the three highest-ranked households separately, with nuisance refitting and unchanged remaining fold labels, gives daily estimates from $-0.272$ to $-0.227$ kWh.

\begin{figure}[!htbp]
\centering
\includegraphics[width=0.92\textwidth]{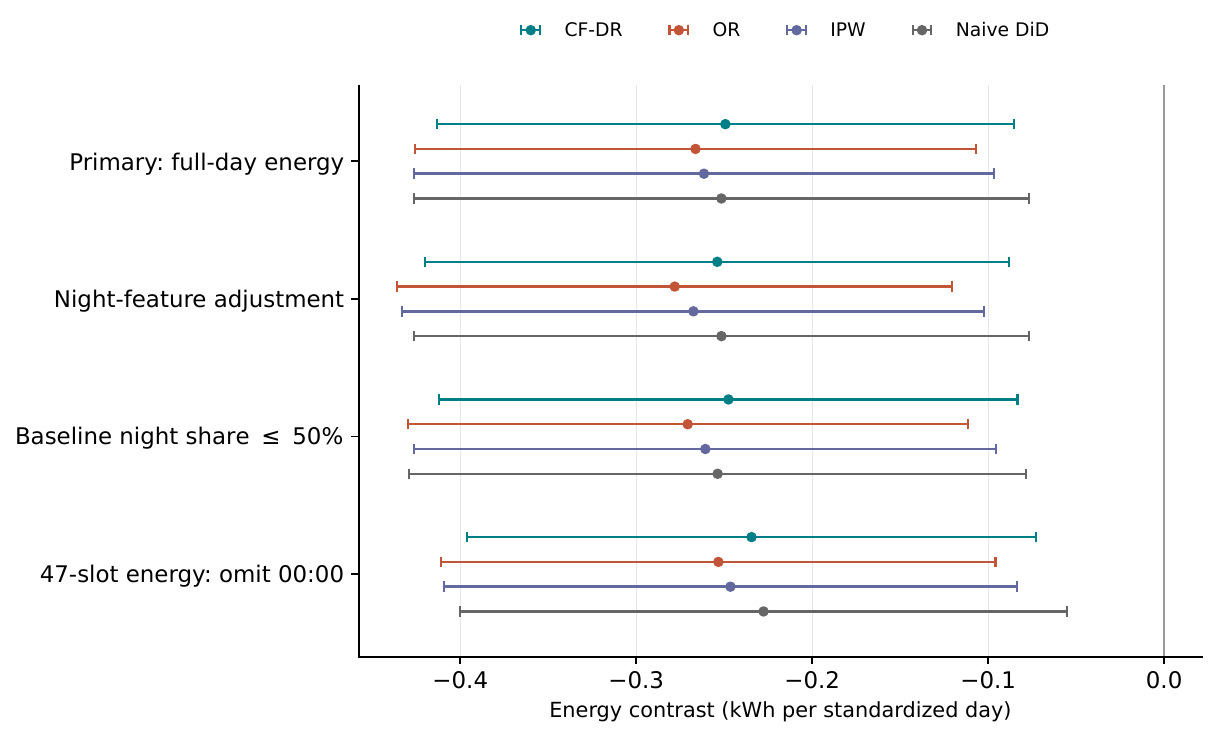}
\caption{Energy contrasts with marginal 95\% intervals. Night-share restriction retains 4,358 households; the 47-slot contrast excludes midnight. Exploratory intervals are nominal.}
\label{fig:supp_lcl_robustness}
\end{figure}

\FloatBarrier
\end{document}